\documentclass[12pt]{article}
\usepackage[utf8]{inputenc}
\usepackage{amsmath}
\usepackage{amssymb}
\usepackage{amsthm}
\usepackage{amsfonts}
\usepackage{graphicx}
\usepackage{multirow}
\usepackage{enumerate}
\usepackage{url} 
\usepackage{float}
\usepackage{xcolor}
\usepackage{booktabs}
\usepackage{tikz}
\usetikzlibrary{calc}
\usepackage{authblk}
\usepackage{array}
\usepackage[top=1in,bottom=1in,left=1in,right=1in]{geometry}

\usepackage{natbib}
\usepackage[unicode=true,bookmarks=true,bookmarksnumbered=false,bookmarksopen=false,breaklinks=false,pdfborder={0 0 1},backref=false,colorlinks=true]{hyperref}
\hypersetup{citecolor=blue}

\newtheorem{theorem}{Theorem}

\newtheorem{assumption}{Assumption}

\newtheorem{remark}{Remark}
\newtheorem{example}{Example}

\newcommand{\bX}{\boldsymbol{X}}\newcommand{\bO}{\boldsymbol{O}}
\newcommand{\bW}{\boldsymbol{W}}
\newcommand{\bI}{\boldsymbol{I}}

\newcommand{\bD}{\boldsymbol{D}}

\newcommand{\bY}{\boldsymbol{Y}}

\newcommand{\bo}{\boldsymbol{o}}

\newcommand{\bzero}{\boldsymbol{0}}
\newcommand{\balpha}{\boldsymbol{\alpha}}
\newcommand{\bbeta}{\boldsymbol{\beta}}
\newcommand{\btheta}{\boldsymbol{\theta}}
\newcommand{\bTheta}{\boldsymbol{\Theta}}
\newcommand{\bpsi}{\boldsymbol{\psi}}

\newcommand{\bfB}{\mathbf{B}}
\newcommand{\bfV}{\mathbf{V}}

\title{Asymptotic inference with flexible covariate adjustment under rerandomization and stratified rerandomization}
\author{Bingkai Wang$^{1}$ and Fan Li$^{2,3}$
\vspace{10pt}

$^1$Department of Biostatistics, School of Public Health, University of Michigan, Ann Arbor, MI, USA

$^{2}$Department of Biostatistics, Yale School of Public Health, New Haven, CT, USA

$^{3}$Center for Methods in Implementation and Prevention Science, Yale School of Public Health, New Haven, CT, USA}

\begin{document}
\def\spacingset#1{\renewcommand{\baselinestretch}%
{#1}\small\normalsize} \spacingset{1}

\date{\vspace{-5ex}}

\maketitle
\begin{abstract}
Rerandomization is an effective treatment allocation procedure to control for baseline covariate imbalance. For estimating the average treatment effect, rerandomization has been previously shown to improve the precision of the unadjusted and the linearly-adjusted estimators over simple randomization without compromising consistency. However, it remains unclear whether such results apply more generally to the class of M-estimators, including the g-computation formula with generalized linear regression and doubly-robust methods, and more broadly, to efficient estimators with data-adaptive machine learners. In this paper, we develop the asymptotic theory for a more general class of covariate-adjusted estimators under rerandomization and its stratified extension. We prove that the asymptotic linearity and the influence function remain identical for any M-estimator under simple randomization and rerandomization, but rerandomization may lead to a non-Gaussian asymptotic distribution. We further explain, drawing examples from several common M-estimators, that asymptotic normality can be achieved if rerandomization variables are appropriately adjusted for in the final estimator. These results are extended to stratified rerandomization. Finally, we study the asymptotic theory for efficient estimators based on data-adaptive machine learners, and prove their efficiency optimality under rerandomization and stratified rerandomization. Our results are demonstrated via simulations and re-analyses of a cluster-randomized experiment that used stratified rerandomization.
\end{abstract}
\noindent%
\vspace{20pt}

\spacingset{1.7}
\setcounter{page}{1}
\section{Introduction}
In randomized experiments, rerandomization \citep{morgan2012rerandomization} refers to a restricted randomization procedure that discards treatment allocation schemes corresponding to large baseline imbalance. That is, treatment assignments will be randomly re-generated until the balance statistics fall within a pre-specified threshold. In biomedical research, such a randomization procedure is also referred to as covariate-constrained randomization \citep{li2016evaluation,li2017evaluation} and increasingly applied to cluster randomization. Compared with simple randomization that assigns treatment by independent coin flips, rerandomization can provide effective design-based control of imbalance associated with a number of covariates, possibly of different types, and improve the study power. Compared to covariate-adaptive randomization, such as stratified randomization \citep{Zelen1974} for balancing discrete baseline variables, rerandomization can easily handle continuous baseline variables with continuously valued design parameters to control the randomization space, thereby offering additional flexibility. Given these potential advantages, rerandomization has become popular in economics \citep{bruhn2009pursuit} and biomedical research \citep{ivers2012allocation}.

Under rerandomization, the statistical properties of the unadjusted and linearly-adjusted estimators for the average treatment effect have been extensively studied under a finite-population design-based framework. \cite{li2018asymptotic} first established the design-based asymptotic theory for the unadjusted estimator (i.e., the two-sample difference-in-means estimator) under rerandomization. \cite{li2020rerandomization} generalized their results to covariate-adjusted estimators based on linear regression, and recommended the combination of rerandomization and regression adjustment to optimize statistical precision. Rerandomization has also been extended to accommodate tiers of covariates with varying importance \citep{morgan2015rerandomization}, sequential rerandomization in batches \citep{zhou2018sequential}, high-dimensional settings with a diverging number of covariates \citep{wang2022rerandomization}, split-plot designs \citep{shi2022rerandomization}, {survey experiments \citep{cytrynbaum2021optimal, yang2023rejective},} cluster rerandomization \citep{lu2023design}, stratified rerandomization \citep{wang2023rerandomization}, rerandomization based on p-values from regression-based balance assessment \citep{zhao2024no}, {finely stratified designs \citep{bai2022optimality,cytrynbaum2024finely}}, {and best-choice rerandomization \citep{WangLi2025}}.  A notable feature of much of this literature is its focus on unadjusted and linearly-adjusted estimators, natural estimator choices under the finite-population design-based framework. 

However, the analysis of randomized experiments in practice frequently involves more general covariate adjustment strategies than linear regression \citep{benkeser2021improving}, and the validity and efficiency of such general estimators remain unexplored in the existing literature on rerandomization. For example, with binary outcomes, the recent guidance from the U.S. Food and Drug Administration \citep{FDA2023} pointed out that ``\emph{Nonlinear models such as logistic regression \dots are commonly used in many clinical settings}''. In this guidance, the g-computation estimator based on logistic regression has been delineated with details as a reliable method for covariate adjustment and model-robust inference.
When the outcomes are missing at random, doubly-robust estimators \citep{robins2007} offer two opportunities for consistent estimation of the average treatment effect; i.e., consistency holds when either the missingness propensity score model or the outcome model is correctly specified but not necessarily both. In cluster-randomized experiments, mixed-effects models are used by 52\% of the primary analysis \citep{fiero2016statistical} and can simultaneously adjust for covariates and account for the intracluster correlation \citep{wang2021mixed}. Finally, data-adaptive machine learners have shown promise to maximally leverage baseline information to achieve full efficiency in randomized experiments 
\citep{chernozhukov2018double,wang2023CRT}. Only for selected scenarios under cluster randomization, previous simulations have demonstrated that rerandomization improved precision for mixed-effects models and generalized estimating equations estimators \citep{li2016evaluation, li2017evaluation}. However, to the best of our knowledge, there has been no formal development of the asymptotic theory for a more general class of covariate-adjusted estimators under rerandomization; thus few explicit recommendations exist for rerandomized experiments when inference involves flexible covariate adjustment.

In this paper, we introduce formal asymptotic results for the class of M-estimators as well as efficient estimators (nuisance models estimated via data-adaptive machine learners) under rerandomization and stratified rerandomization. First, we prove that, under standard regularity conditions for M-estimators, the asymptotic linearity and the influence function remain identical under simple randomization and rerandomization, but rerandomization may lead to a non-Gaussian asymptotic distribution. This result justifies the consistency of M-estimators under rerandomization, and obviates the need to re-derive the influence function under rerandomization. Next, we survey several common M-estimators in randomized experiments, and show that asymptotic normality can be achieved if rerandomization variables are appropriately adjusted for in the final estimator. This result clarifies when rerandomization is asymptotically \emph{ignorable} during the analysis stage, in which case one can conveniently draw inference under rerandomization as if simple randomization were carried out. Then, we develop parallel results to stratified rerandomization---a restricted randomization procedure combining stratification and rerandomization to achieve stronger control of chance imbalance. Finally, we turn to asymptotically efficient estimators based on the efficient influence functions and data-adaptive machine learners. 
We prove that, as long as the rerandomization variables are adjusted for in the machine learners, the efficient estimators continue to be optimally efficient without the need to invoke any additional regularity assumptions.
A key challenge in establishing these asymptotic results is addressing the correlation among individual observations induced by (stratified) rerandomization, for which we develop general technical arguments that accommodate flexible forms of covariate adjustment.


Throughout, we adopt a super-population sampling-based framework, which is an effective framework to study more general covariate-adjusted estimators in randomized experiments \citep{wang2023model}. We refer to \cite{robins2002} and \cite{ding2017bridging} for a comparison between the super-population sampling-based and finite-population design-based framework, and {more recent work of \cite{xu2021potential} and \cite{qu2025randomization} on alternative, designed-based asymptotic theory for M-estimators.} When the unadjusted and linear regression estimators are considered under rerandomization and stratified randomization, our results can be viewed as the sampling-based counterparts of earlier design-based results \citep{li2018asymptotic,li2020rerandomization,lu2023design,wang2023rerandomization}. However, our results go beyond linear regression and address a wider class of covariate-adjusted estimators under rerandomization and stratified rerandomization. {Finally, after posting the first version of our preprint, we became aware of the independent work of \citet{cytrynbaum2024finely}, who develops asymptotic theory for M-estimators under finely stratified randomization \citep{bai2022optimality} and whose result coincides with our Theorem~\ref{thm1} under rerandomization. \citet{cytrynbaum2024finely} further studies finite-population causal estimands and nonlinear imbalance criteria, whereas our analysis emphasizes super-population inference, identifies conditions under which M-estimators have normal limits without ex-post adjustment, and covers theory for addressing missing outcomes and debiased machine-learning estimators under (stratified) rerandomization. As such, these two parallel works remain distinct and are complementary contributions to the growing literature on rerandomization.}

The remainder of the paper is organized as follows. In Section \ref{sec:setup}, we introduce the super-population framework, randomization procedures, and M-estimation. Section~\ref{sec: main} characterizes the asymptotic distributions for M-estimators under rerandomization, and Section~\ref{sec:examples} surveys familiar M-estimators and provides sufficient conditions to achieve asymptotic normality under rerandomization. Section~\ref{sec: stratified-reran} presents parallel development under stratified rerandomization. Section~\ref{sec: eff} introduces the asymptotic theory for estimators based on the efficient influence function under rerandomization and stratified rerandomization. We report simulations in Section~\ref{sec:sim} and re-analyses of a completed cluster-randomized experiment in Section~\ref{sec:data}. Section~\ref{sec: discussion} concludes. R code for reproducing the numerical results is available at \url{https://github.com/BingkaiWang/Rerandomization}. 

\section{Notation, rerandomization, and M-estimation}\label{sec:setup}
\subsection{Assumptions and estimands}
We consider a randomized experiment with $n$ individuals. Each individual $i\ (i=1,\dots,n)$ has an outcome $Y_i$, a non-missing indicator of outcome $R_i$ ($R_i=1$ if $Y_i$ observed and $0$ if missing), a treatment assignment indicator $A_i$ ($A_i=1$ if assigned treatment and $0$ otherwise), and a vector of baseline covariates $\bX_i$. The observed data are $(\bO_1,\dots, \bO_n)$ with $\bO_i = (R_iY_i, R_i, A_i, \bX_i)$. We adopt the potential outcomes framework and define $Y_i(a)$ as the potential outcome if individual $i$ were assigned to treatment group $a, a\in \{0,1\}$. Similarly, we denote $R_i(a)$ as the potential non-missing indicator given treatment $a$. We assume causal consistency such that $Y_i = A_iY_i(1)+(1-A_i)Y_i(0)$ and $R_i = A_iR_i(1)+(1-A_i)R_i(0)$. Denoting the complete data vector for each individual as $\bW_i = (Y_i(1), Y_i(0), R_i(1), R_i(0),\bX_i)$, we make the following assumptions on $(\bW_1,\dots, \bW_n)$.

\begin{assumption}[Super-population sampling]\label{assp1}
    Each complete data vector $\bW_i, i=1,\dots,n$ is an independent and identically distributed draw from an unknown distribution $\mathcal{P}^{\bW}$ on $\bW = (Y(1), Y(0), R(1), R(0),\bX)$ with finite second moments.
\end{assumption}

\begin{assumption}[Missing at random]\label{assp2}
For each $a =0,1$, $R(a)$ is independent of $Y(a)$ given $\bX$ and $P(R(a)=1|\bX)$ is uniformly bounded away from 0.
\end{assumption}

Assumption~\ref{assp1} is a standard condition for making inference under a sampling-based framework. It postulates a notional super-population of units, from which the observed sample is randomly drawn and for which the target estimands may be defined. 
We invoke Assumption~\ref{assp2} as a standard condition to accommodate missing outcomes in randomized experiments. This assumption is required for the doubly-robust estimators in Theorems~\ref{thm2}, \ref{thm4}, \ref{thm: ml-rerandomization}, and \ref{thm: ml-stratified-rerandom}, but is not required for other results without missing data. 

Our goal is to estimate the average treatment effect, defined in a general form as
\begin{equation*}
    \Delta^* = f(E[Y(1)], E[Y(0)])
\end{equation*}
where $f(x,y)$ is a pre-specified function defining the scale of effect measure. For example, $f(x,y) = x-y$ corresponds to the treatment effect on the difference scale, whereas $f(x,y) = x/y$ corresponds to the risk ratio scale, which is a standard choice when the outcome is binary. 
Finally, by Assumption~\ref{assp1}, the expectation operator in the estimand definition (and throughout the paper) is defined with respect to the super-population distribution $\mathcal{P}^W$.

\subsection{Simple randomization and rerandomization}\label{sec: trt-assign}
Simple randomization assigns treatment via independent coin flips, that is, $A_i$ is independently determined by a Bernoulli distribution $\mathcal{P}^A$ with $P(A_i=1)=\pi\in (0,1)$.  By design, each $A_i$ is independent of $(\bW_1,\dots, \bW_n)$, and the observed data $(\bO_1,\dots, \bO_n)$ are independent and identically distributed given Assumption~\ref{assp1}.

Rerandomization improves simple randomization by controlling imbalance on a subset of measured baseline covariates, which we denote as $\bX^r$ with $\bX^r \subset \bX$ \citep{morgan2012rerandomization}; we refer to $\bX^r$ as rerandomization variables.  
Rerandomization involves the following three steps. First, we independently generate $A_1^*,\dots, A_n^*$ from a Bernoulli distribution with $P(A_i^*=1) = \pi\in (0,1)$ as in simple randomization. In the second step, we compute the imbalance statistic on $\bX^r_i$ and its variance estimator as
\begin{align*}
    \bI &= \frac{1}{N_1} \sum_{i=1}^n A^*_i \bX^r_i - \frac{1}{N_0} \sum_{i=1}^n (1-A_i^*) \bX^r_i, \\
    \widehat{Var}(\bI) &= \frac{1}{N_1N_0} \sum_{i=1}^n (\bX^r_i - \overline{\bX^r})(\bX^r_i - \overline{\bX^r})^\top,
\end{align*}
where $N_1 = \sum_{i=1}^n A_i^*$,  $N_0 = n - N_1$ and $\overline{\bX^r} = n^{-1} \sum_{i=1}^n \bX^r_i$. Lastly, given a pre-specified balance threshold $t >0$, we check whether $\bI^\top \{\widehat{Var}(\bI)\}^{-1} \bI < t$. If true, the final treatment assignment $(A_1,\dots, A_n)$ is set to be $(A_1^*,\dots, A_n^*)$; otherwise, we repeat the steps until we obtain the first randomization scheme that satisfies the balance criterion. 

The above rerandomization procedure considers the Mahalanobis distance to define the balance criterion, and $\bI^\top \{\widehat{Var}(\bI)\}^{-1} \bI$ follows a chi-squared distribution asymptotically. The balance threshold $t$ can then be selected to adjust the rejection rate. For example, letting $q$ denote the dimension of $\bX^r$, then $t$ being the $\alpha$-quantile of $\mathcal{X}^2_q$ leads to an approximate rejection rate $1-\alpha$. In general, a smaller $t$ corresponds to stronger control over imbalance; if $t\rightarrow +\infty$, then rerandomization reduces to simple randomization. From a statistical perspective, rerandomization introduces correlation among $(A_1,\dots, A_n)$ and hence correlation among observed data $(\bO_1,\dots,\bO_n)$, leading to a key complication in studying large-sample results of treatment effect estimators.
{Of note, this rerandomization scheme is based on simple randomization, and thus differs from complete randomization, which fixes the numbers of treated and control units. The latter can be viewed as a special case of stratified randomization with a single stratum, as discussed in Section~\ref{sec: stratified-reran}.}

Beyond the Mahalanobis distance, rerandomization can also be more generally based on $\bI^\top \widehat{\mathbf{H}}^{-1} \bI < t$, where $\widehat{\mathbf{H}}$ is any positive definite matrix (e.g., a diagonal matrix collecting the sample variance for each component of $\bI$). 
For our main theorems, we focus on rerandomization using the Mahalanobis distance, which leads to the most interpretable results. Nevertheless, we show how our results can be extended to accommodate such a more general distance function in Remarks~\ref{remark1} and \ref{remark3}. 


\subsection{M-estimation}
M-estimator \citep[Section 5]{vaart_1998} refers to a wide class of estimators that are solutions to estimating equations. Specifically, let $\btheta = (\Delta, \bbeta) \in \mathbb{R}^{l+1}$ be an $(l+1)$-dimensional vector of parameters, where $\Delta \in \mathbb{R}$ is the parameter of interest and  $\bbeta \in \mathbb{R}^l$ is an $l$-dimensional vector of nuisance parameters. An M-estimator $\widehat{\btheta} = (\widehat{\Delta}, \widehat{\bbeta})$ for $\btheta$ is the solution to $$\sum_{i=1}^n \bpsi(\bO_i;\btheta) = \bzero,$$
where $\bpsi$ is a pre-specified $(l+1)$-dimensional estimating function. For example, $\bpsi$ is the score function if $\widehat{\btheta}$ is obtained by maximum likelihood estimation. The M-estimator for $\Delta^*$ is $\widehat{\Delta}$. For $a= 0,1$, we denote $\bO(a) = (R(a)Y(a), R(a), a, \bX)$ and $||\cdot||_2$ as the $L_2$ matrix norm. Throughout, we assume the following regularity conditions for M-estimators.
\begin{assumption}[Regularity conditions]\label{assp3} (1) $\btheta \in \bTheta$, a compact subset of $\mathbb{R}^{l+1}$. \\
(2) $E[||\bpsi(\bO(a);\btheta)||_2^2] < \infty$ for all $\btheta \in \bTheta$ and $a \in \{0,1\}$.
(3) There exists a unique solution $\underline{\btheta} = (\underline{\Delta}, \underline{\bbeta})$, an inner point of $\bTheta$, to the equations $\pi E[\bpsi(\bO(1);\btheta)] + (1-\pi)E[\bpsi(\bO(0);\btheta)] = \bzero$.
(4) The function $\btheta \mapsto \bpsi(\bo;\btheta)$ is twice continuously differentiable for every $\bo$ in the support of $\bO$ and is dominated by an integrable function. 
(5) There exist an integrable function $v(\bo)$ that dominates the first and second derivatives of $\bpsi(\bo;\btheta)$ in  $||\btheta-\underline{\btheta}||_2 < C$ for some $C>0$. 
(6) $E\left[||\frac{\partial}{\partial \btheta} \bpsi(\bO(a);\btheta)\big|_{\btheta=\underline{\btheta}}||_2^2\right] < \infty$ for $a \in \{0,1\}$, and $\pi E\left[\frac{\partial}{\partial \btheta} \bpsi(\bO(1);\btheta)\big|_{\btheta=\underline{\btheta}}\right] + (1-\pi)E\left[\frac{\partial}{\partial \btheta} \bpsi(\bO(0);\btheta)\big|_{\btheta=\underline{\btheta}}\right]$ is invertible.
\end{assumption}

Assumption~\ref{assp3} are largely moment and continuity assumptions to rule out irregular or degenerate M-estimators, which are similar to those invoked in Section 5.3 of \cite{vaart_1998}. It is important to note that when the M-estimator is defined based on a working regression model, Assumption~\ref{assp3} does not necessarily imply the working model is correctly specified. For example, if the analysis of covariance (ANCOVA, Example \ref{eg:ancova} in Section~\ref{sec:examples}) is used, Assumption~\ref{assp3} (2)-(6) reduce to assuming $\bW$ has finite fourth moments and invertible covariance matrices, rather than assuming the ANCOVA model is correctly specified.

Under simple randomization, it has been established that Assumptions~\ref{assp1} and \ref{assp3} yield $\widehat{\Delta} \xrightarrow{p} \underline{\Delta}$, i.e., $\widehat{\Delta}$ converge  in probability to a limit quantity $\underline{\Delta}$, and $\sqrt{n}(\widehat{\Delta} - \underline{\Delta}) \xrightarrow{d} N(0,V)$, i.e., $\sqrt{n}$-rate asymptotic normality of $\widehat{\Delta}$ \citep{vaart_1998} with a certain asymptotic variance matrix $V$. To achieve consistent estimation, i.e., $\widehat{\Delta} \xrightarrow{p} \Delta^*$, one needs to choose appropriate working models (e.g., ANCOVA in randomized experiments) or make additional assumptions (e.g., Assumption~\ref{assp2} in the presence of missing outcomes) such that $\Delta^* = \underline{\Delta}$; we refer to Section~\ref{sec:examples} for detailed examples. In the next section, we formally describe the asymptotic behavior of $\widehat{\Delta}$ under rerandomization.

\section{Asymptotics for M-estimators under rerandomization}\label{sec: main}
\begin{theorem}\label{thm1}
    Given Assumptions~\ref{assp1} and \ref{assp3}, and under either simple randomization or rerandomization, we have $\widehat{\Delta} \xrightarrow{p} \underline{\Delta}$ and asymptotic linearity, i.e., 
    \begin{equation}
    \sqrt{n}(\widehat{\Delta} - \underline{\Delta}) = \frac{1}{\sqrt{n}}\sum_{i=1}^nIF(\bO_i) + o_p(1), \label{asymptoticlinearity}
    \end{equation} 
	where the influence function $IF(\bO_i)$ is the first entry of $-\bfB^{-1}\bpsi(\bO_i;\underline{\btheta})$ with\\ $\bfB = \pi E\left[\frac{\partial}{\partial \btheta} \bpsi(\bO(1);\btheta)\big|_{\btheta=\underline{\btheta}}\right] + (1-\pi)E\left[\frac{\partial}{\partial \btheta} \bpsi(\bO(0);\btheta)\big|_{\btheta=\underline{\btheta}}\right]$.

 Furthermore, under rerandomization with Mahalanobis distance as the balance criterion,
 \begin{equation}\label{eq: asymptotic distribution}
     \sqrt{n}(\widehat{\Delta} - \underline\Delta) \xrightarrow{d} \sqrt{V}\left(\sqrt{1-R^2}\ z + \sqrt{R^2}\ r_{q,t}\right),
 \end{equation}
where $z \sim N(0,1)$, $r_{q,t} \sim D_1|( \bD^\top  \bD < t)$ for $\bD \sim N(\bzero, \mathbf{I}_q)$ with $D_1$ being the first element of $\bD$ and $r_{q,t}$ being independent of $z$, $V$ is the asymptotic variance of $\widehat{\Delta}$ under simple randomization, and
\begin{equation*}
    R^2 = \frac{\pi(1-\pi)}{V} Cov[IF\{\bO(1)\}- IF\{\bO(0)\}, \bX^r]Var(\bX^r)^{-1}Cov[\bX^r, IF\{\bO(1)\}- IF\{\bO(0)\}].
\end{equation*}
Consistent estimators for $V$ and $R^2$ are provided in Supplementary Material A.
\end{theorem}

A central finding in Theorem \ref{thm1} is that the convergence in probability and asymptotic linearity properties remain identical under simple randomization and rerandomization. 
As a result, if an M-estimator is consistent to the average treatment effect ($\underline{\Delta} = \Delta^*$) under simple randomization, then it remains consistent under rerandomization. Furthermore, the influence function of an M-estimator remains unchanged under rerandomization.

In addition, Equation~\eqref{eq: asymptotic distribution} provides the asymptotic distribution of an M-estimator under rerandomization. 
This result resembles Theorem 1 of \cite{li2018asymptotic} and Theorem 1 of \cite{li2020rerandomization} (derived under a finite-population design-based framework) if $\Delta^*=E[Y(1)-Y(0)]$ and $\widehat{\Delta}$ is the unadjusted estimator or the linearly-adjusted estimator. 
Our new contribution in Theorem \ref{thm1} is to establish the asymptotic theory for a broad class of M-estimators and estimands with different scales of effect measure under rerandomization.

Different from simple randomization that leads to asymptotic normality of an M-estimator, rerandomization may correspond to a non-normal asymptotic distribution. By Theorem \ref{thm1}, the M-estimator converges weakly to a distribution given by a summation of two independent components---a normal variate $\sqrt{V}\sqrt{1-R^2}\ z$ and an independent \emph{truncated normal variate} $\sqrt{V}\sqrt{R^2}\ r_{q,t}$. 
{Here, \(R^2\in[0,1]\) measures the degree of alignment between the estimator and the rerandomized covariates \(\bX^r\). With a larger \(R^2\), rerandomization can yield a larger efficiency gain \citep{li2018asymptotic}. The magnitude of \(R^2\) is determined by the choice of estimator and by how strongly \(\bX^r\) are associated with the potential outcomes. For example, for the unadjusted estimator under equal allocation, \(R^2\) is driven by the covariance between \(\bX^r\) and $Y(1)-Y(0)$; thus, more prognostic rerandomization variables generally lead to greater variance reduction. In contrast, for covariate-adjusted estimators that already fully account for \(\bX^r\), the resulting influence function may be nearly orthogonal to the imbalance vector \(\bI\), leading to \(R^2\) close to zero. In this case, the efficiency gain from rerandomization overlaps with that from covariate adjustment, as discussed in Section~\ref{sec:examples}.}


Based on \cite{morgan2012rerandomization} and \citet{li2018asymptotic}, the density of $\sqrt{V}\left(\sqrt{1-R^2}\ z + \sqrt{R^2}\ r_{q,t}\right)$ is symmetric about zero and bell-shaped, resembling a normal distribution. However, its variance is $V\{1- (1-v_{q,t})R^2\}$, where $v_{q,t} = P(\mathcal{X}^2_{q+2} < t)/P(\mathcal{X}^2_{q} < t) < 1$, and therefore $V\{1- (1-v_{q,t})R^2\} \le V$. As a result, compared to simple randomization, rerandomization does not increase the asymptotic variance of an M-estimator. Under a design-based framework, Theorem 2 of \cite{li2018asymptotic} 
showed that the variance reduction for an unadjusted estimator is non-decreasing in $R^2$ and non-increasing in both $t$ and $q$. {Correspondingly, rerandomization can also reduce the asymptotic quantile ranges of the asymptotic distribution.} Owing to Theorem \ref{thm1}, these previous observations made for the unadjusted estimator continue to hold for general M-estimators. 

For statistical inference, we propose consistent estimators $\widehat{V}$ and $\widehat{R}^2$ in Supplementary Material A. The 95\% confidence interval for $\widehat{\Delta}$ can then be approximated by the following Monte Carlo approach. We first generate a large number of draws ($m=10,000$) of $z$ and $r_{q,t}$, e.g., via rejection sampling {or following the more efficient construction by Proposition 2 of \cite{li2018asymptotic}}. We then compute the interval $\widehat{\Delta}+\sqrt{\widehat{V}/n} \left(\sqrt{1-\widehat{R}^2}\ z + \sqrt{\widehat{R}^2}\ r_{q,t}\right)$ for each draw, and take the $2.5\%$ and $97.5\%$ quantiles of its empirical distribution to obtain the final confidence interval estimator $\widehat{CI}$. As $m,n\rightarrow \infty$, we have asymptotic nominal coverage such that $P(\Delta^*\in \widehat{CI}) \rightarrow 0.95$.

\begin{remark}\label{remark1}
When rerandomization is based on a more general balance criterion $\bI^\top \widehat{\mathbf{H}}^{-1} \bI < t$, Theorem~\ref{thm1} remains the same, except that the asymptotic distribution in \eqref{eq: asymptotic distribution} is now $\sqrt{V}\sqrt{1-R^2}\ z +  Cov[IF\{\bO(1)\}- IF\{\bO(0)\}, \bX^r]\bfV_I^{-1/2} \boldsymbol{R}_{\mathbf{H},t}$, where $\boldsymbol{R}_{\mathbf{H},t} \sim \bD|(\bD^\top\bfV_I^{1/2} \underline{\mathbf{H}}^{-1}\bfV_I^{1/2} \bD < t)$ and $\bfV_I = \{\pi(1-\pi)\}^{-1} Var(\bX^r)$ with $\underline{\mathbf{H}}$ being the probability limit of $n\widehat{\mathbf{H}}$. 
Compared to Equation~\eqref{eq: asymptotic distribution}, the first component of the asymptotic distribution remains the same, but the second component becomes more complex and depends on the choice of $\widehat{\mathbf{H}}$. Nevertheless, as long as $\underline{\mathbf{H}}$ is positive definite, rerandomization still leads to no increase in the asymptotic variance. In addition, the confidence interval can be similarly constructed with a Monte Carlo approach. Similar results were first pointed out by \cite{lu2023design} under a design-based framework with cluster rerandomization (for unadjusted and linearly-adjusted estimators); we extend their results to the class of M-estimators (proof in Supplementary Material B).
\end{remark}

\begin{remark}\label{rmk: tiers-of-covariates}
{Theorem 1 can be extended to accommodate rerandomization 
with tiers of covariates \citep{morgan2015rerandomization}. Specifically, we partition the $q$ rerandomization variables into $B$ non-overlapping groups. For each group $b = 1, \dots, B$, we denote the rerandomization variables as $\bX^{r,b}$ with dimension $q_b \leq q$, so that $\sum_{b=1}^B q_b = q$. 
Let $\mathbf{G}_b$ be the $q \times q_b$ selection matrix satisfying $\bX^{r,b} = \mathbf{G}_b^\top \bX^r$, where the columns of $\mathbf{G}_b$ are the standard basis vectors in $\mathbb{R}^q$ corresponding to the indices of tier $b$. The $q_b$-dimensional imbalance statistic is $\bI_b = \mathbf{G}_b^\top \bI$, where $\bI \in \mathbb{R}^q$ is defined in Section~\ref{sec: trt-assign}.
We define the $q_b \times q_b$ weighting matrix $\widehat{\mathbf{H}}_b$ and balancing threshold $t_b > 0$, and accept a randomization if $\bI_b^\top \widehat{\mathbf{H}}_b^{-1} \bI_b < t_b$ for all $b = 1, \dots, B$, allowing different tiers of covariates to be assigned separate balance criteria. Under this design, Theorem 1 still holds with the asymptotic distribution described in Remark 1, except that $\boldsymbol{R}_{\mathbf{H},t} \sim \bD|(\bD^\top\bfV_I^{1/2} \mathbf{G}_b \underline{\mathbf{H}}_b^{-1} \mathbf{G}_b^\top \bfV_I^{1/2} \bD < t_b: b = 1,\dots, B)$, with $\underline{\mathbf{H}}_b$ defined as the probability limit of 
$n\widehat{\mathbf{H}}_b$.}
\end{remark}


\section{Examples}\label{sec:examples}
In this section, we survey common M-estimators used for analyzing randomized experiments and provide important special cases where $R^2=0$, meaning that the M-estimator has fully accounted for rerandomization. 
Therefore, in these cases, each M-estimator remains asymptotically normal under rerandomization.
These examples serve as important clarifications on when standard inference may be carried out ignoring rerandomization. Among the four examples, estimators in Examples \ref{eg:ancova}, \ref{eg:gcomp}, and \ref{eg:mixed} are model-robust and hence consistent under arbitrary working model misspecification under simple randomization; see discussions in \citet{Tsiatis2008, colantuoni2015leveraging, wang2021mixed}; this model-robustness property is carried under rerandomization as implied by Theorem~\ref{thm1}. Example \ref{eg:dr} focuses on experiments with missing outcomes, and we additionally require Assumption~\ref{assp2} to establish the double robustness property under simple randomization \citep{robins2007}; this property remains to hold under rerandomization as a result of Theorem \ref{thm1}. 

\begin{example}[\textbf{ANCOVA estimator}]\label{eg:ancova}
In randomized experiments, the ANCOVA estimator considers the working model $E[Y|A,\bX] = \beta_0 + \beta_AA + \bbeta_{\bX}^\top \bX$ and uses ordinary least squares to obtain estimators $(\widehat\beta_0, \widehat\beta_A, \widehat{\bbeta}_X)$ for parameters $(\beta_0, \beta_A, \bbeta_X)$. In this context, $\widehat{\Delta}=\widehat\beta_A$ is the ANCOVA estimator for the average treatment effect estimand, $\Delta^* = E[Y(1)-Y(0)]$. 
\end{example}

\begin{example}[\textbf{G-computation estimator with working logistic regression}]\label{eg:gcomp}
When the outcome is binary, the logistic regression model, $\text{logit}\{P(Y=1|A,\bX)\} = \beta_0 + \beta_AA + \bbeta_{\bX}^\top \bX$, is commonly used to analyze randomized experiments. With maximum likelihood estimators $(\widehat\beta_0, \widehat\beta_A, \widehat{\bbeta}_X)$, we construct $\widehat{p}_i(a) = \text{logit}^{-1}(\widehat{\beta}_0 + \widehat{\beta}_Aa + \widehat{\bbeta}_{\bX}^\top \bX_i)$ for $a = 0,1$. The g-computation estimator for the average treatment effect on any scale $\Delta^*=f(E[Y(1)], E[Y(0)])$ is $\widehat{\Delta} = f(n^{-1}\sum_{i=1}^n \widehat{p}_i(1), n^{-1}\sum_{i=1}^n \widehat{p}_i(0))$. 
\end{example}

\begin{example}[\textbf{DR-WLS estimator for handling missing outcomes}]\label{eg:dr}
In the presence of missing outcomes, the doubly-robust weighted least squares (DR-WLS) estimator involves the following steps. First, we fit a logistic regression working model to predict missingness, i.e., $\text{logit}\{P(R=1|\bX)\} = \alpha_0 + \alpha_A A + \balpha_X \bX$, and obtain the missingness propensity score as $\widehat{P}(R=1|A,\bX) = \text{logit}^{-1}(\widehat{\alpha}_0 + \widehat{\alpha}_A A + \widehat{\balpha}_X \bX)$. Second, we fit an outcome regression working model, $E[Y|A,\bX] = g^{-1}(\beta_0 + \beta_A A + \bbeta_X \bX)$, using data with $R_i=1$ and canonical link function $g$, weighted by the inverse propensity score $\{\widehat{P}(R=1|A_i, \bX_i)\}^{-1}$. We denote the model fit as $\widehat{E}[Y|A,\bX] = g^{-1}(\widehat{\beta}_0 + \widehat{\beta}_A A + \widehat{\bbeta}_X \bX)$. Finally, we apply a g-computation estimator $\widehat{\Delta} = f\left(n^{-1}\sum_{i=1}^n \widehat{E}[Y|1,\bX_i], n^{-1}\sum_{i=1}^n \widehat{E}[Y|0,\bX_i]\right)$ to estimate the average treatment effect estimand, $\Delta^*=f(E[Y(1)], E[Y(0)])$. We assume missing at random (Assumption~\ref{assp2}) and at least one of the two working models is correctly specified.
\end{example}

\begin{example}[\textbf{Mixed-ANCOVA estimator under cluster randomization}]\label{eg:mixed}
When rerandomization is at the cluster level, the outcomes for each cluster become an $N_i$-dimensional vector $\bY_i = (Y_{i1},\dots, Y_{iN_i})$, where $N_i$ is the size of cluster $i$. In addition, the covariates $\bX_i$ is a collection of individual covariates $\{\bX_{i1}, \dots, \bX_{iN_i}\}$. For simplicity, we assume non-informative cluster sizes, i.e., $N_i$ is independent of all $(Y_{ij}, A_i, \bX_{ij})$, and all $(Y_{ij}, A_i, \bX_{ij})$ are {marginally identically distributed within and across clusters}, but they can have an arbitrary with-cluster correlation structure. Then, rerandomization is based on $\bX_i^r$ being a summary function of $(\bX_{i1}^r,\dots, \bX_{iN_i}^r)$, e.g., ${N_i}^{-1} \sum_{j=1}^{N_i} \bX_{ij}^r$.
The mixed-ANCOVA estimator involves fitting the linear mixed model $Y_{ij} = {\beta}_0 + {\beta}_A A_i + {\bbeta}_X \bX_{ij} + \delta_i + \varepsilon_{ij}$, where $\delta_i \sim N(0,\tau^2)$ is the random intercept and $\varepsilon_{ij} \sim N(0,\sigma^2)$ is the independent error. Under maximum likelihood estimation, we can construct a model-robust estimator $\widehat{\Delta} = \widehat{\beta}_A$ for the average treatment effect estimand $\Delta^* = E[Y_{ij}(1)-Y_{ij}(0)]$ \citep{wang2021mixed}. 
{With informative cluster sizes and heterogeneous within-cluster distributions, the mixed-ANCOVA estimator can be adapted to target cluster-level or individual-level estimands following \cite{wang2023CRT}.}
\end{example}

Each estimator described in Examples \ref{eg:ancova}-\ref{eg:mixed} is an M-estimator (with the corresponding estimating function provided in Supplementary Material B.1.3), and its consistency and asymptotic linearity properties directly carry from simple randomization to rerandomization. Furthermore, Theorem~\ref{thm2} clarifies, for each estimator, when asymptotic normality holds under rerandomization and thus inference can proceed by ignoring rerandomization.

\begin{theorem}\label{thm2}
    Assume Assumptions~\ref{assp1}-\ref{assp3} and that $\bX$ includes $\bX^r$ as a subset for covariate adjustment. Then, for the ANCOVA, logistic regression, DR-WLS, and mixed-ANCOVA estimators described in Examples \ref{eg:ancova}-\ref{eg:mixed}, we have $\sqrt{n}(\widehat{\Delta} - \Delta^*) \xrightarrow{d} N(0,V)$ under rerandomization if at least one of the following conditions holds: (1) $\pi = 0.5$ and $\Delta^* = E[Y(1)-Y(0)]$ is the estimand on the difference scale, or (2) the outcome regression model further includes a treatment-by-covariate interaction term $\bbeta_{AX^r}A\bX^r$ for $\bX^r$, or (3) the outcome regression model is correctly specified. In other words, $R^2=0$ under any of conditions (1), (2), or (3).
\end{theorem}

\begin{remark}
While the ANCOVA and mixed-ANCOVA estimators target $\Delta^* = E[Y(1)-Y(0)]$ through the treatment coefficient, they can be adapted to estimate treatment effect estimands on other scales through a g-computation procedure (as described in Examples \ref{eg:gcomp} and \ref{eg:dr}). When treatment-by-covariate interaction terms are included, they become the ANCOVA2 \citep{Tsiatis2008} and mixed-ANCOVA2 estimators \citep{wang2021mixed}, for which g-computation (instead of simply interpreting $\widehat{\beta}_A$) is needed for consistent estimation. {For the DR-WLS estimator, Theorem~\ref{thm2} can still hold if $\bX^r$ is omitted from the missingness model. However, including $\bX^r$ is recommended as it facilitates the correct specification of the missingness model and increases the chance for consistency. 
}
\end{remark}

Theorem~\ref{thm2} highlights the benefit of adjusting for baseline variables that are used for rerandomization. In this case, inference with many commonly used estimators under rerandomization can be performed as if simple randomization were used in the design stage; thus rerandomization becomes asymptotically ignorable in the analysis stage. This result has two implications. First, the choice of threshold $t$ or the weighting matrix $\widehat{\mathbf{H}}$ has no impact asymptotically for the estimators considered in Examples \ref{eg:ancova}-\ref{eg:mixed}. Therefore, the question of how one should select these design parameters for optimizing the asymptotic precision is immaterial in large samples. Second, many existing results about variance under simple randomization naturally extend to rerandomization. For example, for estimating the average treatment effect on the difference scale, the ANCOVA estimator under condition (1) or (2) in Theorem \ref{thm2} does not reduce the asymptotic precision by adjusting for baseline covariates under simple randomization \citep{Tsiatis2008}; the same result holds under rerandomization as long as the adjustment set included $\bX^r$. As another example, when condition (1) holds, not only ANCOVA and mixed-ANCOVA provide consistent point estimates, their model-based variance estimators are also consistent \citep{wang2019analysis,wang2021mixed}; such properties hold under rerandomization by Theorem \ref{thm2}. 

{In addition, covariate adjustment in Examples \ref{eg:ancova}-\ref{eg:mixed} sets $R^2 = 0$, which does not mean that the asymptotic variance in Theorem~\ref{thm1} is increased. This is because adjusting for prognostic covariates can reduce $V$, which often dominates the precision gain. Since rerandomization and covariate adjustment target the same prognostic variation, when the analysis correctly orthogonalizes with respect to the rerandomized covariates (so the relevant $R^2$ terms vanish), rerandomization offers no further asymptotic precision gain, and the variance equals that under simple randomization (Remark~\ref{remark:bound}). This finding refines the traditional view that rerandomization always improves efficiency: such gains arise for the unadjusted difference-in-means estimator \citep{morgan2012rerandomization} but asymptotically disappear once covariate adjustment fully accounts for the rerandomized variables \citep{li2020rerandomization}. In practice, however, covariate adjustment may be under-specified or constrained by model form or sample size, in which case rerandomization continues to deliver meaningful finite-sample gains and better covariate balance, as also observed in our simulations. 

On the contrary, if $\bX^r$ is not included in the adjustment set, Theorem~\ref{thm2} may not hold, and the asymptotic distribution of the estimator is non-normal (Theorem \ref{thm1}). 
Intuitively, this is the scenario where $\bX^r$ can further explain the variance of $\widehat{\Delta}$ so that $R^2> 0$. If inference is based on asymptotic results derived under simple randomization, the resulting confidence interval may be too wide, leading to overly conservative inference.
We thus recommend adjusting for rerandomization variables whenever using M-estimators described in Examples  \ref{eg:ancova}-\ref{eg:mixed}, such that conventional inferential procedure based on asymptotic normality can be justified by Theorem~\ref{thm2}.

\section{Extensions to stratified rerandomization}\label{sec: stratified-reran}
\subsection{Defining stratified rerandomization}\label{subsec: st-reran}
Stratified rerandomization is effective when we want to control the imbalance on $\bX^r$ at level $t$ but eliminate the imbalance on certain categorical variables via stratification \citep{wang2023rerandomization}. 
To formalize this scheme, we first introduce stratified randomization without rerandomization. 
{Stratified (block) randomization refers to a scheme that achieves exact balance within each pre-specified stratum. Let $S \in \bX$ be a categorical baseline variable encoding the randomization strata with domain $\mathcal{S}$.
For example, $S$ takes values in $\mathcal{S}=\{\textup{male smoker},\ \textup{female smoker},\ \textup{male non-smoker},\ \textup{female non-smoker}\}$ if the strata are defined by sex and smoking status. We assume the domain $\mathcal{S}$ has fixed size with $P(S =s ) > 0$ for each $s \in \mathcal{S}$, which precludes matched-pair designs. Assumption 1 implies that $S_1,\dots, S_n$ are independent and identically distributed.   }

{Within each randomization stratum, a size-$k$ permutation block with $\pi k$ ones and $(1-\pi)k$ zeros in random order is independently sampled to assign the treatment for the first $k$ participants (1 for treatment and 0 for control). When a permutation block is exhausted, a new one is sampled to assign the treatment for the next $k$ participants. In practice, the block size $k$ is chosen to make $\pi k$ an integer, and a smaller $k$ yields a stronger balance at intermediate stages of enrollment but increases predictability of future outcomes. However, 
{size-$k$ permutation blocks will be first-order equivalent to complete randomization asymptotically \citep{bugni2018}, and thus different choices of $k$ do not influence our asymptotic results.}
Unlike simple randomization where treatment assignment is independent, stratified randomization directly introduces correlation among treatment assignment and observed data $(\bO_1,\dots, \bO_n)$.}

Stratified rerandomization combines stratified randomization and rerandomization through the following steps. First, we generate $\widetilde{A}_1,\dots, \widetilde{A}_n$ by stratified randomization based on randomization strata $\{S_1,\dots, S_n\}$ and parameter $\pi \in (0,1)$. Second, we compute
    \begin{align*}
        \widetilde{\bI} &= \frac{1}{\widetilde{N}_1}\sum_{i=1}^n \widetilde{A}_i \bX_i^r - \frac{1}{\widetilde{N}_0} \sum_{i=1}^n (1-\widetilde{A}_i) \bX_i^r, \\
        \widehat{Var}(\widetilde{\bI}) &=  \frac{n}{\widetilde{N}_1\widetilde{N}_0} \left[\frac{1}{n}\sum_{i=1}^n \bX^r_i\bX^r_i{}^\top - \sum_{s \in \mathcal{S}} \frac{1}{n}\sum_{i=1}^n I\{S_i=s\}\overline{\bX_s^r}\overline{\bX_s^r}^\top\right],
    \end{align*}
where $\widetilde{N}_1 = \sum_{i=1}^n \widetilde{A}_i$, $\widetilde{N}_0 = n - \widetilde{N}_1$, and $\overline{\bX_s^r} = (\sum_{i=1}^n I\{S_i=s\})^{-1}\sum_{i=1}^n \bX_i^r I\{S_i=s\}$. In the last step, if $\widetilde{\bI}^\top \{\widehat{Var}(\widetilde{\bI})\}^{-1} \widetilde{\bI} < t$ for a pre-specified balance threshold $t$, we set the final treatment assignment $(A_1,\dots, A_n)$ to be $(\widetilde{A}_1,\dots, \widetilde{A}_n)$; otherwise, we return to the first step to obtain the first randomization scheme that satisfies the balance criterion.

Compared to rerandomization defined in Section~\ref{sec: trt-assign}, stratified rerandomization involves two changes: the first step replaces simple randomization by stratified randomization, and the variance estimator $\widehat{Var}(\widetilde{\bI})$ is updated to reflect the true variance of the imbalance statistic under stratification, which is smaller than under simple randomization. 
{This design is asymptotically equivalent to the first stratified rerandomization criterion (overall Mahalanobis distance) of \cite{wang2023rerandomization} since we assume $\pi$ is constant across randomization strata---the most common setting in randomized experiments.} 
Without loss of generality, we assume $S \notin \bX^r$; otherwise, rerandomization is effectively controlling for imbalance on $\bX^r\setminus S$.
Under stratified rerandomization, we next develop the large-sample properties for M-estimators regarding the general estimand $\Delta^*$.

\subsection{Asymptotics theory under stratified rerandomization}
\begin{theorem}\label{thm3}
    Given Assumptions~\ref{assp1} and \ref{assp3}, and under stratified rerandomization, we have $\widehat{\Delta} \xrightarrow{p} \underline{\Delta}$ and asymptotic linearity as in Equation~\eqref{asymptoticlinearity}, and
     \begin{align}\label{eq: stratified-rerandom}
     \sqrt{n}(\widehat{\Delta} - \underline\Delta) \xrightarrow{d} \sqrt{\widetilde{V}}\left\{\sqrt{1-\widetilde{R}^2}\ z + \sqrt{\widetilde{R}^2}\ r_{q,t} \right\},
 \end{align}
 where  $z \sim N(0,1)$, $r_{q,t} \sim D_1|( \bD^\top  \bD < t)$ for $\bD \sim N(\bzero, \mathbf{I}_q)$ with $D_1$ being the first element of $\bD$ and $r_{q,t}$ being independent of $z$, and
     \begin{align*}
        \widetilde{V} &= V - \pi(1-\pi) E\left[E[IF\{\bO(1)\}- IF\{\bO(0)\}|S]^2\right], \\
         \widetilde{R}^2 &= \frac{\pi(1-\pi)}{\widetilde{V}} \boldsymbol{C}^\top E[Var(\bX^r|S)]^{-1}\boldsymbol{C}.
    \end{align*}
with $\boldsymbol{C} = E[Cov[\bX^r, IF\{\bO(1)\}- IF\{\bO(0)\}|S]]$ and $V$ being the asymptotic variance under simple randomization. Consistent estimators for $\widetilde{V}$ and $\widetilde{R}^2$ are provided in Supplementary Material A.
\end{theorem}

Theorem~\ref{thm3} parallels Theorem~\ref{thm1} in terms of convergence in probability, asymptotic linearity, and weak convergence results. The major difference lies in the constant factors of the asymptotic distribution. In particular, $\widetilde{V}$ substitutes $V$ in Equation~\eqref{eq: stratified-rerandom} to 
account for the variance reduction from stratification. Likewise, compared to Theorem \ref{thm1}, $\widetilde{R}^2$ is also updated by replacing marginal covariance with the expectation of conditional covariance. Despite these changes, the properties of the asymptotic distribution remain similar to those under rerandomization, and statistical inference considerations under rerandomization discussed in Section~\ref{sec: main} still apply to stratified rerandomization. Finally, as a special case, if we consider the unadjusted estimator, the difference estimand $\Delta^*=E[Y(1)-Y(0)]$, and fixed strata categories, Theorem~\ref{thm3} becomes the counterpart of Theorem 3 of \cite{wang2023rerandomization} under a sampling-based framework. 

\begin{remark}\label{remark3}
    If rerandomization is based on $\bI^\top \widehat{\mathbf{H}}^{-1} \bI < t$, Theorem~\ref{thm3} remains the same with the asymptotic distribution in Equation~\eqref{eq: stratified-rerandom} updated to  $\sqrt{\widetilde{V}}\sqrt{1-\widetilde{R}^2}\ z +  \boldsymbol{C}^\top\widetilde{\bfV}_I^{-1/2} \widetilde{\boldsymbol{R}}_{\mathbf{H},t}$, where $\widetilde{\boldsymbol{R}}_{\mathbf{H},t} \sim \bD|(\bD^\top\widetilde{\bfV}_I^{1/2} \underline{\mathbf{H}}^{-1}\widetilde{\bfV}_I^{1/2} \bD < t)$ and $\widetilde{\bfV}_I = \{\pi(1-\pi)\}^{-1} E[Var(\bX^r|S)]$, and $\underline{\mathbf{H}}$ is the probability limit of $n\widehat{\mathbf{H}}$. As long as $\underline{\mathbf{H}}$ is positive definite, stratified rerandomization does not lead to asymptotic precision reduction over stratified or simple randomization.
\end{remark}

\begin{remark}
    Like rerandomization, stratified rerandomization can also accommodate tiers of covariates in a similar way as in Remark~\ref{rmk: tiers-of-covariates}. When we enforce $\sum_{i=1}^n A_i = \pi n$ in rerandomization as done in \cite{li2018asymptotic}, rerandomization with tiers of covariates accommodates stratified randomization as a special case \citep{wang2023rerandomization}; these schemes all fit into stratified rerandomization with tiers of covariates under our framework. 
    
\end{remark}

\subsection{Continued examples}

For Examples 1-4 in Section~\ref{sec:examples}, their model-robustness or double-robustness property remains the same under stratified rerandomization. We thus extend Theorem~\ref{thm3} for stratified rerandomization given below. 

\begin{theorem}\label{thm4}
Assume Assumptions~\ref{assp1}-\ref{assp3} and that $\bX$ includes $\bX^r$ as a subset and $S$ as dummy variables for covariate adjustment.  Then, for the ANCOVA, logistic regression, DR-WLS, or mixed-ANCOVA estimator in Examples \ref{eg:ancova}-\ref{eg:mixed}, we have $\sqrt{n}(\widehat{\Delta} - \Delta^*) \xrightarrow{d} N(0,V)$ under stratified rerandomization if at least one of the following conditions holds: (1) $\pi = 0.5$ and $\Delta^* = E[Y(1)-Y(0)]$, or (2) the outcome regression model further includes treatment-by-covariate interaction terms for both $\bX^r$ and $S$, or (3) the outcome regression model is correctly specified. In other words, $\widetilde{R}^2=0$ and $\widetilde{V} = V$ under condition (1), (2), or (3).    
\end{theorem}

Theorem~\ref{thm4} clarifies that adjusting for both $\bX^r$ and $S$ in the working models offers asymptotic normality of the considered M-estimators under stratified rerandomization. In this case, the asymptotic distribution for each estimator is no different from that under simple randomization, and asymptotic results developed for those estimators under simple randomization can directly apply even though the actual assignment comes from a stratified rerandomization procedure. 
Among the three conditions, conditions (1) and (3) remain the same as Theorem~\ref{thm3}, but condition (2) is modified to further include additional treatment-by-stratum interaction terms (as intuitively, the stratum variable is part of the randomization procedure). Similar to Theorem~\ref{thm3}, the threshold $t$ and weighting matrix $\widehat{\mathbf{H}}$ do not contribute to the asymptotic distribution, suggesting that the choice of these design parameters in general does not impact the asymptotic inference of the considered treatment effect estimators. 
Due to the appeal of applying normal approximation for inference, we maintain our recommendation that adjusted for rerandomization variables (including rerandomization variables and dummy stratum variables) in the analysis.

\section{Efficient inference with machine learning}\label{sec: eff}
Beyond traditional M-estimators based on parametric working models, data-adaptive machine learning models have emerged for analysis of randomized experiments under simple randomization, typically through the vehicle of efficient influence function \citep{van2011targeted, chernozhukov2018double}. 
While these estimators provide flexible covariate adjustment and achieve the asymptotic efficiency lower bound under simple randomization, 
their properties under rerandomization or stratified rerandomization remain unknown. In this section, we turn our focus to a class of efficient estimators via the efficient influence function and obtain their asymptotic distribution under rerandomization and stratified rerandomization. We focus on the setting where outcomes are missing at random (Assumption~\ref{assp2}), but our theory can be straightforwardly extended to accommodate other settings (e.g., cluster-randomized experiments), where efficient estimation has been previously investigated under simple randomization \citep{wang2023CRT}. 

\subsection{Efficient inference under rerandomization}\label{subsec: eff-rerandom}
Under rerandomization, we focus on the efficient estimator that uses double machine learning (DML) for estimating nuisance functions with cross-fitting \citep{chernozhukov2018double}. Let $\widehat{\eta}_a(\bX;\mathcal{D}), \widehat{\kappa}_a(\bX;\mathcal{D})$ denote the pre-specified estimators for $E[Y(a)|\bX], E[R(a)|\bX]$ trained on data $\mathcal{D}$ and evaluated at $\bX$. For cross-fitting, we randomly partition the index set $\{1,\dots,n\}$ into $K$ folds with approximately equal sizes. Specifically, let $\mathcal{I}_k$ denote the index set for the $k$-th fold, we have $\{1,\dots,n\}=\cup_{k=1}^K \mathcal{I}_k$ and $|\mathcal{I}_k| = K^{-1}n$. If $K$ does not divide $n$, we require $||\mathcal{I}_k|- K^{-1}n| \le 1$ instead. Defining $\mathcal{O}_{a,-k} = \{(R_iY_i,R_i,\bX_i): i = 1,\dots, n;i \notin \mathcal{I}_k; A_i =a\}$ as the training data for fold $k$ within treatment group $a$, we specify the nuisance function estimators as $\widehat{\eta}_a(\bX_i) = \sum_{k=1}^K I\{i \in \mathcal{I}_k\}\widehat{\eta}_a(\bX_i;\mathcal{O}_{a,-k})$ and $\widehat{\kappa}_a(\bX_i) = \sum_{k=1}^K I\{i \in \mathcal{I}_k\}\widehat{\kappa}_a(\bX_i;\mathcal{O}_{a,-k})$. Next, the efficient estimator for $E[Y(a)]$ is constructed based on the efficient influence function as
\begin{equation}\label{eq: eff}
    \widehat{\mu}_a^{\textup{dml}} = \frac{1}{n}\sum_{i=1}^n \left[\frac{I\{A_i=a\}}{\pi^a(1-\pi)^{1-a}}\frac{R_i}{\widehat{\kappa}_a(\bX_i)}\{Y_i - \widehat{\eta}_a(\bX_i)\} + \widehat{\eta}_a(\bX_i)\right],
\end{equation}
and output $\widehat{\Delta}^{\textup{dml}} = f(\widehat{\mu}_1^{\textup{dml}}, \widehat{\mu}_0^{\textup{dml}})$.  For this estimator, we assume the following conditions. 

\begin{assumption}[Conditions for nuisance function estimators]\label{assp: machine-learning}
  For $a = 0,1$, we assume (1) $E[\left\{\widehat{\eta}_a(\bX;\mathcal{D}_n)-E[Y(a)|\bX]\right\}^2]^{1/2} = o(n^{-1/4})$ and $E[\left\{\widehat{\kappa}_a(\bX;\mathcal{D}_n)-E[R(a)|\bX]\right\}^2]^{1/2} = o(n^{-1/4})$ for $\mathcal{D}_n = \{R_i(a)Y_i(a), R_i(a),\bX_i\}_{i=1}^n$; (2) $ \{\widehat{\kappa}_a(\bX;\mathcal{D}_n)\}^{-1}$ and $ E[Y^2(a)|\bX]$ are uniformly bounded with probability 1.
\end{assumption}

Assumption~\ref{assp: machine-learning}(1) requires that the nuisance function estimators converge to their target in $L_2$-norm with a rate faster than $n^{1/4}$, if the training data consist of independent and identically distributed data. This rate can be achieved by existing machine learning methods, including random forests \citep{wager2018estimation}, deep neural networks \citep{Farrell2021}, and highly adaptive lasso \citep{benkeser2016highly}, under mild regularity conditions. 
Of note, for the outcome model, the training data $\mathcal{D}_n$ are used to learn  $E[R(a)Y(a)|R(a), \bX]$, which yields $E[Y(a)|R(a)=1, \bX]$ and thus $E[Y(a)|\bX]$ under Assumption~\ref{assp2}.
In Assumption~\ref{assp: machine-learning}(2), we assume uniform boundedness on several components of $\mathcal{P}^W$ and the estimators so that the remainder terms can be appropriately controlled. Overall, Assumption~\ref{assp: machine-learning} resembles the standard assumptions for efficient inference made in \cite{chernozhukov2018double}, and importantly, no extra condition is specifically assumed for addressing rerandomization. Theorem~\ref{thm: ml-rerandomization} below provides the asymptotic results for $\widehat{\Delta}^{\textup{dml}}$.

\begin{theorem}\label{thm: ml-rerandomization}
 Given Assumptions~\ref{assp1}, \ref{assp2}, and \ref{assp: machine-learning}, and under rerandomization, we have consistency, i.e., $\widehat{\Delta}^{\textup{dml}} \xrightarrow{p} \Delta^*$, and asymptotic linearity, i.e., $\sqrt{n}(\widehat{\Delta}^{\textup{dml}}-\Delta^*) = \frac{1}{\sqrt{n}} \sum_{i=1}^n EIF(\bO_i) + o_p(1)$ with 
\begin{equation}\label{eq: eif}
    EIF(\bO_i) = \sum_{a=0}^1 f_a' \left[\frac{I\{A_i=a\}}{\pi^a(1-\pi)^{1-a}}\frac{R_i}{E[R_i(a)|\bX_i]}\{Y_i - E[Y_i(a)|\bX_i]\} + E[Y_i(a)|\bX_i] - E[Y_i(a)]\right],
\end{equation}
where $f_a'$ is the partial derivative of $f$ with respect to $E[Y(a)]$.

Furthermore, $\sqrt{n}(\widehat{\Delta}^{\textup{dml}}-\Delta^*) \xrightarrow{d} \sqrt{V}\left(\sqrt{1-R^2}\ z + \sqrt{R^2}\ r_{q,t}\right)$ as described in Equation~\eqref{eq: asymptotic distribution} with $V$ being the asymptotic variance under simple randomization and $R^2$ defined in Theorem \ref{thm1} by setting $IF(\bO_i)=EIF(\bO_i)$. Additionally, $R^2=0$ if the adjustment set $\bX$ includes the rerandomization variables $\bX^r$. Consistent estimators for $V$ and $R^2$ are provided in Supplementary Material A.
\end{theorem}

Theorem~\ref{thm: ml-rerandomization} justifies the validity of efficient estimation of the average treatment effect estimand under rerandomization. Additionally, the overall structure of its asymptotic expansion resembles that for M-estimators in Theorem~\ref{thm1}, with the only difference being that the influence function is given in Equation~\eqref{eq: eif}, which is the efficient influence function under simple randomization. If covariate adjustment includes the rerandomization variables, then $\widehat{\Delta}^{\textup{dml}}$ is asymptotically normal with asymptotic variance $V$---the usual efficiency lower bound given by the variance of $EIF(\bO_i)$ under simple randomization. On the contrary, if $\bX$ does not include all rerandomization variables, then  $V$ does not fully account for the precision gain from rerandomization, which leads to a positive $R^2$ and non-normal asymptotic distribution.



{Finally, although efficient estimators are asymptotically optimal, they may require large sample sizes to realize efficiency gains and to prevent overfitting of data-adaptive nuisance estimators. In comparison, when parametric models are correctly specified or nearly correct, M-estimators often outperform efficient estimators in finite samples with greater numerical stability and smaller variance estimates. The relative performance thus depends on the sample size, model misspecification, and complexity of the nuisance structure, and should be evaluated on a case-by-case basis. Consistent with prior recommendations for efficient inference in cluster-randomized experiments \citep{wang2023CRT}, efficient estimators are more likely to show improvement when the sample size per arm is relatively large (e.g., $n>100$).}

\begin{remark}\label{remark:bound}
{\cite{bai2023efficiency} in Theorem 4.1  showed that the semiparametric lower bound is unchanged across a variety of randomization schemes, including rerandomization schemes we study. In Supplementary Material B.3.3, we provide additional results that the observed-data tangent space and efficient influence function are also invariant to randomization schemes. Therefore, rerandomization may reduce the variance of estimators that are not orthogonal to prognostic covariates, but it cannot lower the classical efficiency bound. In particular, the DML estimator $\widehat{\Delta}^{\textup{dml}}$ remains semiparametrically efficient under rerandomization and stratified rerandomization.}
\end{remark}

\begin{remark}\label{rmk: ml}
    With no missing outcomes, we have $R_i \equiv1$, and the efficient estimator is simplified by setting $\widehat{\kappa}_a(\bX;\mathcal{D}_n) \equiv 1$. In this case, Assumption~\ref{assp: machine-learning}(1) is relaxed to only requiring a consistency assumption $E[\left\{\widehat{\eta}_a(\bX;\mathcal{D}_n)- E[Y(a)|\bX]\right\}^2] = o(1)$ without the need for regulating rate of convergence. Under this modified Assumption~\ref{assp: machine-learning}, Theorem~\ref{thm: ml-rerandomization} still holds and extends a special case of \cite{chernozhukov2018double} from simple randomization to rerandomization.
\end{remark}

\subsection{Results under stratified rerandomization}
When using stratified rerandomization, the efficient estimator in Section~\ref{subsec: eff-rerandom} requires modification to achieve our target asymptotic results. This is because stratification introduces additional correlation among observed data, which cannot be handled by the standard cross-fitting procedure. To overcome this challenge, we perform cross-fitting within each stratum $s \in \mathcal{S}$ and treatment group $a \in \{0,1\}$, an approach proposed by \cite{rafi2023efficient} to address stratification. 
Specifically, for each $s \in \mathcal{S}$ and $a \in \{0,1\}$, we randomly partition $\{i=1,\dots, n: S_i=s, A_i=a\}$ into $k$ folds and denote $\mathcal{I}_{ask}$ as the index set of the $k$-th fold. Similarly, we require roughly equal fold sizes across $k$ given $a$ and $s$. We next define the $k$-th fold in stratum $s$ as $\mathcal{I}_{sk} = \mathcal{I}_{1sk} \cup \mathcal{I}_{0sk}$ and define the training data for this fold as $\mathcal{O}_{as,-k} = \{(R_iY_i,Y_i,\bX_i): i = 1,\dots, n;i \notin \mathcal{I}_{sk}; A_i =a;S_i=s\}$. The nuisance function estimators are updated to $\widehat{\eta}_a(\bX_i) = \sum_{k=1}^K \sum_{s\in \mathcal{S}} I\{i \in \mathcal{I}_{sk}, S_i=s\}\widehat{\eta}_a(\bX_i;\mathcal{O}_{as,-k})$ and $\widehat{\kappa}_a(\bX_i) = \sum_{k=1}^K \sum_{s\in \mathcal{S}} I\{i \in \mathcal{I}_{sk}, S_i=s\}\widehat{\kappa}_a(\bX_i;\mathcal{O}_{as,-k})$. Then $\widehat{\mu}_a^{\textup{dml}}$ is defined as in Equation~\eqref{eq: eff} and $\widehat{\Delta}^{\textup{dml}} = f(\widehat{\mu}_1^{\textup{dml}}, \widehat{\mu}_0^{\textup{dml}})$. 
Based on this modified cross-fitting scheme, Theorem~\ref{thm: ml-stratified-rerandom} gives the counterpart of Theorem~\ref{thm: ml-rerandomization} under stratified rerandomization.

\begin{theorem}\label{thm: ml-stratified-rerandom}
Given Assumptions~\ref{assp1}, \ref{assp2}, and \ref{assp: machine-learning}, and under stratified rerandomization, we have the same consistency and asymptotic linearity as in Theorem~\ref{thm: ml-rerandomization}. Furthermore, $\sqrt{n}(\widehat{\Delta}^{\textup{dml}}-\Delta^*) \xrightarrow{d} \sqrt{\widetilde{V}}\left(\sqrt{1-\widetilde{R}^2}\ z + \sqrt{\widetilde{R}^2}\ r_{q,t}\right)$ as described in Equation~\eqref{eq: stratified-rerandom} with $\widetilde{V}$ and $\widetilde{R}^2$ now defined by setting $IF(\bO_i)=EIF(\bO_i)$. Additionally, $\widetilde{R}^2=0$ if the adjustment set $\bX$ includes the rerandomization variables $\bX^r$ and the dummy stratum variables $S$. Consistent estimators for $\widetilde{V}$ and $\widetilde{R}^2$ are provided in the Supplementary Material A.
\end{theorem}

\section{Simulations}\label{sec:sim}
We conduct simulation studies to demonstrate our asymptotic results. The first simulation focuses on continuous outcomes with a difference estimand $\Delta^*=E[Y(1)-Y(0)]$, and the second simulation focuses on binary outcomes with a ratio estimand $\Delta^*=E[Y(1)]/E[Y(0)]$. Both simulations consider complete outcomes and outcomes under missing at random, with simple randomization, rerandomization, and stratified rerandomization. 

\subsection{Simulation design}
For the first simulation with continuous outcomes, we set $n=400$ and independently generate three baseline covariates $X_{i1} \sim \mathcal{N}(1,1)$, $S_i|X_{i1} \sim \mathcal{B}(0, 0.4 + 0.2 I\{X_{i1}<1\})$, and $X_{i2}|(X_{i1},S_i) \sim \mathcal{N}(0,1)$ for $i=1,\dots, n$, where $\mathcal{N}, \mathcal{B}$ represent normal and Bernoulli distributions, respectively. Next, we independently generate $Y_i(a) \sim \mathcal{N}\left(4aS_iX_{i2}^2 + 2 e^{X_{i1}} + |X_{i2}|, 1\right)$, $R_i(a) \sim \mathcal{B}\left\{\textup{expit}(0.6 + 0.6a + X_{i2} + S_i)\right\}$ 
for $a = 0,1$ and $i =1,\dots, n$, where $\textup{expit}(x) = 1/(1+e^{-x})$. For treatment allocation, we set the randomization variables as $\bX^r_i = (X_{i1},X_{i2})$ and stratification variable as $S_i$. Then $(A_1,\dots, A_n)$ are generated under simple randomization, rerandomization, and stratified rerandomization as described in Sections~\ref{sec: trt-assign} and \ref{subsec: st-reran} with Mahalanobis distance, $\pi=0.5$, and threshold $t=1$ (corresponding to an acceptance rate of approximately $40\%$). Denoting $Y_i = A_iY_i(1) + (1-A_i)Y_i(0)$ and $R_i = A_iR_i(1) + (1-A_i)R_i(0)$, the observed data are $\{Y_i, A_i, X_{i1},X_{i2}, S_i:i=1,\dots, n\}$ under the no missing outcome scenario, and $\{R_iY_i, R_i, A_i, X_{i1},X_{i2}, S_i:i=1,\dots, n\}$ under the missing outcome scenario. We repeat the above procedure to generate $1000$ simulated data sets.

For each simulated data set, we implement the following three estimators under the no-missing-outcome scenario. The unadjusted estimator is given by the difference in mean outcomes between treatment groups and ignores covariates. The ANCOVA estimator (Example \ref{eg:ancova}) adjusts for all covariates $\bX_i=(X_{i1},X_{i2},S_i)$. The ML estimator refers to the efficient estimator described in Section~\ref{sec: eff} but setting $R_i = \widehat{\kappa}_a(\bX_i)=1$ as discussed in Remark~\ref{rmk: ml}. For the outcome regression model $\widehat{\eta}_a(\bX_i)$, we adopt an ensemble learner of generalized linear models, regression trees, and neural networks via SuperLearner \citep{van2007super}. Under the outcome missing at random scenario, we instead compute the DR-WLS estimator (described in Example \ref{eg:dr}) and the efficient DML estimator (as described in Section~\ref{sec: eff}); we use the same ensemble learners for estimating the two nuisance functions for the latter. Of note, since the outcome missingness is generated by a generalized linear model, the missingness propensity score model is correctly specified in both the DR-WLS and DML estimators.

For each estimator, we report the following performance metrics: bias, empirical standard error (ESE), average of standard error estimators assuming simple randomization (ASE$^*$, i.e., no extra adjustment is performed for rerandomization or stratified rerandomization), coverage probabilities and average confidence-interval lengths (computed either using normal approximations under simple randomization, or using the derived asymptotic distribution that accounts for the actual randomization procedure). 

In the second simulation study with binary outcomes, the data generating process is the same as the first simulation except that $Y_i(a) \sim \mathcal{B}\left\{\textup{expit}(-4+4aS_iX_{i2}^2 + e^{X_{i1}}-|X_{i2}|)\right\}$. Here, the unadjusted estimator is the ratio of mean outcomes between treatment groups. The outcome model used in the ANCOVA estimator and DR-WLS estimator is substituted by GLM2, which denotes logistic regression in Example~\ref{eg:gcomp} with all treatment-by-covariate interaction terms. The ML and DML estimators are also modified to target the ratio estimand, but the nuisance function estimators remain the same as in the first simulation study. 

\subsection{Simulation results}

Table~\ref{tab:sim1} summarizes the results of the first simulation study. Across different settings, all estimators have negligible bias and achieve nominal coverage probability based on our derived asymptotic distributions in Theorems~\ref{thm1}-\ref{thm: ml-stratified-rerandom}, thereby empirically supporting our theoretical results. For the unadjusted estimator, we observe that rerandomization and stratified rerandomization can improve precision, reflected by ESE being $\sim 20\%$ smaller than ASE$^*$ and CP-Normal being close to 1. Under rerandomization and stratified rerandomization, the ANCOVA, ML, DR-WLS, and DML estimators all have similar performance compared to simple randomization, and the corresponding coverage probabilities are all close to 0.95; this is expected because asymptotic normality holds for these covariate-adjusted estimators under rerandomization. These results further support Theorems~\ref{thm2}, \ref{thm4}, \ref{thm: ml-rerandomization}, and \ref{thm: ml-stratified-rerandom}, by which the ANCOVA, ML, DR-WLS, and DML estimators have been shown to fully account for the variance reduction brought by the rerandomization procedure ($R^2=0$). Of note, the empirical standard error tends to be slightly larger than the average standard error estimators for the adjusted estimators (especially when there are missing outcomes and data-adaptive estimation is used); this is driven by outlier estimates from a few simulated data sets, which have negligible impact on the coverage probability. We have repeated the simulation with $n=1000$, and this variance underestimation issue disappears (results omitted for brevity). Finally, in terms of finite-sample efficiency, we observe that parametric working models can improve precision over the unadjusted estimator, and machine learning estimators often lead to further variance reduction.

\begin{table}[ht]
    \centering
    \caption{\small Simulation results for continuous outcomes. ESE denotes the empirical standard error. ASE$^*$ denotes the average estimated standard error computed under the assumption of simple randomization. {CP-Normal and wCI-Normal denote the coverage probability and average width of confidence intervals based on normal approximations, respectively. CP-True and wCI-True denote the coverage probability and average width of confidence intervals based on the derived asymptotic distribution that accounts for rerandomization, respectively.} }\label{tab:sim1}
    \resizebox{\textwidth}{!}{
    \begin{tabular}{cccccccccc}
    \hline
    \begin{tabular}[c]{@{}c@{}}Missing\\     Outcomes?\end{tabular}   & Randomization & Estimator & Bias & ESE & ASE$^*$ & CP-Normal & CP-True & wCI-Normal & wCI-True\\
    \hline
    \multirow{9}{*}{No}     &  \multirow{3}{*}{\shortstack{Simple\\ randomization}} & Unadjusted & 0.00 & 1.23 & 1.19 & 0.95 & 0.95 & 4.67 & 4.62 \\ 
          & & ANCOVA & 0.01 & 0.82 & 0.79 & 0.95 & 0.96 & 3.10 & 3.07 \\ 
          & & ML & -0.00 & 0.43 & 0.40 & 0.94 & 0.95  & 1.57 & 1.56 \\ 
    \cmidrule(lr){2-10}
&  \multirow{3}{*}{\shortstack{Rerandomization}} & Unadjusted & -0.01 & 0.95 & 1.19 & 0.99 & 0.94 & 4.67 & 3.49 \\ 
          & & ANCOVA & -0.00 & 0.83 & 0.79 & 0.95 & 0.95 & 3.10 & 3.07 \\ 
          & & ML & 0.02 & 0.45 & 0.41 & 0.94 & 0.94 & 1.59 & 1.56 \\ 
    \cmidrule(lr){2-10}
&  \multirow{3}{*}{\shortstack{Stratified\\rerandomization}} & Unadjusted & 0.02 & 0.92 & 1.19 & 0.99 & 0.94 & 4.67 & 3.45 \\  
          & & ANCOVA & 0.02 & 0.82 & 0.79 & 0.94 & 0.95 & 3.10 & 3.07 \\ 
          & & ML & 0.01 & 0.51 & 0.47 & 0.95 & 0.94  & 1.83 & 1.78 \\  
          \hline
\multirow{6}{*}{Yes} & \multirow{2}{*}{\shortstack{Simple\\ randomization}} & DR-WLS & -0.00 & 1.00 & 0.94 & 0.95 & 0.94 & 3.67 & 3.74 \\ 
          & &  DML & 0.05 & 0.69 & 0.61 & 0.96 & 0.95 & 2.38 & 2.42 \\  
              \cmidrule(lr){2-10}
& \multirow{2}{*}{\shortstack{Rerandomization}} & DR-WLS & 0.03 & 0.97 & 0.94 & 0.95 & 0.95 & 3.68 & 3.74 \\ 
          & &  DML & 0.07 & 0.87 & 0.63 & 0.95 & 0.95 & 2.45 & 2.47 \\ 
              \cmidrule(lr){2-10}
& \multirow{2}{*}{\shortstack{Stratified\\ rerandomization}} & DR-WLS & 0.06 & 0.99 & 0.94 & 0.96 & 0.95 & 3.67 & 3.73 \\  
          & &  DML & 0.10 & 0.82 & 0.69 & 0.95 & 0.94 & 2.72 & 2.72 \\ 
\hline
    \end{tabular}
    }
\end{table}

Table~\ref{tab:sim2} presents the results for the second simulation, with similar overall findings to those from the first simulation. Here, since the ANCOVA estimator is replaced by logistic regression with treatment-by-covariate interactions, the GLM2 estimator is asymptotically normal as stated in Theorem~\ref{thm2}(2) and Theorem~\ref{thm4}(2). 

\begin{table}[ht]
    \centering
    \caption{\small Simulation results for binary outcomes. ESE denotes the empirical standard error. ASE$^*$ denotes the average estimated standard error computed under the assumption of simple randomization. {CP-Normal and wCI-Normal denote the coverage probability and average width of confidence intervals based on normal approximations, respectively. CP-True and wCI-True denote the coverage probability and average width of confidence intervals based on the derived asymptotic distribution that accounts for rerandomization, respectively.} }\label{tab:sim2}
    \resizebox{\textwidth}{!}{
    \begin{tabular}{cccccccccc}
    \hline
    \begin{tabular}[c]{@{}c@{}}Missing\\     Outcomes?\end{tabular}   & Randomization & Estimator & Bias & ESE & ASE$^*$ & CP-Normal & CP-True &wCI-Normal & wCI-True\\
    \hline
    \multirow{9}{*}{No}     &  \multirow{3}{*}{\shortstack{Simple\\ randomization}} & Unadjusted & 0.02 & 0.25 & 0.24 & 0.94 & 0.93 & 0.93 & 0.93 \\ 
          & & GLM2 &0.01 & 0.19 & 0.18 & 0.94 & 0.94& 0.72 & 0.71 \\  
          & & ML & 0.01 & 0.19 & 0.18 & 0.94 & 0.94 & 0.73 & 0.72 \\ 
    \cmidrule(lr){2-10}
&  \multirow{3}{*}{\shortstack{Rerandomization}} & Unadjusted & 0.02 & 0.21 & 0.24 & 0.97 & 0.94 & 0.93 & 0.82 \\ 
          & & GLM2 & 0.01 & 0.19 & 0.18 & 0.94 & 0.94 & 0.72 & 0.71 \\  
          & & ML & 0.02 & 0.20 & 0.19 & 0.94 & 0.94 & 0.77 & 0.76 \\  
    \cmidrule(lr){2-10}
&  \multirow{3}{*}{\shortstack{Stratified\\rerandomization}} & Unadjusted & 0.01 & 0.21 & 0.24 & 0.97 & 0.93 & 0.93 & 0.90 \\  
          & & GLM2 & 0.00 & 0.19 & 0.18 & 0.94 & 0.94 & 0.72 & 0.71 \\  
          & & ML & 0.01 & 0.19 & 0.18 & 0.95 & 0.94 & 0.73 & 0.72 \\  
          \hline
\multirow{6}{*}{Yes} & \multirow{2}{*}{\shortstack{Simple\\ randomization}} & DR-WLS & 0.02 & 0.21 & 0.21 & 0.95 & 0.95 & 0.82 & 0.81 \\  
          & &  DML & 0.03 & 0.26 & 0.22 & 0.96 & 0.96 & 0.86 & 0.85 \\  
              \cmidrule(lr){2-10}
& \multirow{2}{*}{\shortstack{Rerandomization}} & DR-WLS & 0.02 & 0.22 & 0.21 & 0.94 & 0.94 & 0.82 & 0.81 \\  
          & &  DML & 0.03 & 0.23 & 0.21 & 0.94 & 0.94 & 0.84 & 0.83 \\  
              \cmidrule(lr){2-10}
& \multirow{2}{*}{\shortstack{Stratified\\ rerandomization}} & DR-WLS & 0.02 & 0.22 & 0.21 & 0.94 & 0.94 & 0.82 & 0.81 \\ 
          & &  DML &  0.03 & 0.28 & 0.23 & 0.95 & 0.95 & 0.89 & 0.89 \\  
\hline
    \end{tabular}
    }
\end{table}

{To illustrate the performance of different approaches under moderate or small sample sizes, we replicate the first simulation without missing outcomes, setting $n=100$ and $n=30$. Table~\ref{tab:sim3} presents the results, mirroring the patterns in Table~\ref{tab:sim1}. Specifically, the unadjusted estimator continues to overestimate variance if rerandomization is ignored, and ANCOVA and ML estimators yield similar coverage probabilities with or without accounting for rerandomization. As sample size decreases, small biases and under-coverage emerge, reflecting the expected impact of model misspecification and numerical instability in finite samples.}

\begin{table}[ht]
    \centering
    \caption{\small Simulation results for continuous outcomes with moderate and small sample sizes (no missing outcomes).  {CP-Normal and wCI-Normal denote the coverage probability and average width of confidence intervals based on normal approximations, respectively. CP-True and wCI-True denote the coverage probability and average width of confidence intervals based on the derived asymptotic distribution that accounts for rerandomization, respectively.} }\label{tab:sim3}
    \resizebox{\textwidth}{!}{
    \begin{tabular}{cccccccccc}
    \hline
    Sample size & Randomization & Estimator & Bias & ESE & ASE$^*$ & CP-Normal & CP-True & wCI-Normal & wCI-True \\
    \hline
    \multirow{9}{*}{$n=100$}     &  \multirow{3}{*}{\shortstack{Simple\\ randomization}} & Unadjusted &  0.02 & 2.41 & 2.30 & 0.95 & 0.96 & 9.00 & 9.24 \\ 
          & & ANCOVA & -0.01 & 1.72 & 1.45 & 0.93 & 0.93 & 5.70 & 5.86 \\
          & & ML & 0.02 & 1.30 & 1.15 & 0.94 & 0.95& 4.52 & 4.64 \\  
    \cmidrule(lr){2-10}
&  \multirow{3}{*}{\shortstack{Rerandomization}} & Unadjusted & 0.02 & 1.85 & 2.30 & 0.99 & 0.96 & 9.02 & 6.77 \\  
          & & ANCOVA & -0.00 & 1.64 & 1.46 & 0.94 & 0.94 & 5.72 & 5.87 \\  
          & & ML & 0.06 & 1.31 & 1.15 & 0.95 & 0.95 & 4.50 & 4.47 \\ 
    \cmidrule(lr){2-10}
&  \multirow{3}{*}{\shortstack{Stratified\\rerandomization}} & Unadjusted &  -0.07 & 1.87 & 2.30 & 1.00 & 0.94& 9.02 & 6.66 \\ 
          & & ANCOVA &-0.07 & 1.63 & 1.46 & 0.95 & 0.95 & 5.71 & 5.86 \\ 
          & & ML & -0.04 & 1.53 & 1.31 & 0.95 & 0.95 & 5.13 & 5.06 \\ 
          \hline
    \multirow{9}{*}{$n=30$}     &  \multirow{3}{*}{\shortstack{Simple\\ randomization}} & Unadjusted &   -0.03 & 4.45 & 3.84 & 0.94 & 0.94 & 15.06 & 15.04 \\ 
          & & ANCOVA & -0.13 & 3.20 & 2.19 & 0.88 & 0.88 & 8.57 & 8.56 \\  
          & & ML &  0.02 & 3.39 & 2.80 & 0.95 & 0.95 & 10.98 & 10.97 \\ 
    \cmidrule(lr){2-10}
&  \multirow{3}{*}{\shortstack{Rerandomization}} & Unadjusted & -0.02 & 3.36 & 3.86 & 1.00 & 0.94 & 15.11 & 11.12 \\ 
          & & ANCOVA & -0.10 & 2.99 & 2.16 & 0.88 & 0.89 & 8.48 & 8.47 \\  
          & & ML &  0.08 & 3.22 & 2.81 & 0.95 & 0.94 & 11.00 & 10.48 \\ 
    \cmidrule(lr){2-10}
&  \multirow{3}{*}{\shortstack{Stratified\\rerandomization}} & Unadjusted & 0.02 & 3.30 & 3.85 & 0.99 & 0.94 & 15.09 & 10.75 \\  
          & & ANCOVA &0.01 & 2.84 & 2.16 & 0.90 & 0.90 & 8.48 & 8.47 \\ 
          & & ML & -0.03 & 3.49 & 3.09 & 0.95 & 0.94 & 12.11 & 11.44 \\  
\hline
    \end{tabular}
    }
\end{table}

\section{Data application}\label{sec:data}
The Effectiveness of Group Focused Psychosocial Support for Adults Affected by Humanitarian Crises (GroupPMPlus) study is a cluster-randomized experiment in Nepal designed to improve the mental health of people affected by humanitarian emergencies such as pandemics, war, and environmental disasters \citep{jordans2021effectiveness}. The cluster-level intervention was Group Problem Management Plus, a psychological treatment of 5 weekly sessions (versus standard care). 
In this study, 72 wards (the smallest administrative units in Nepal, representing clusters) were enrolled, consisting of 609 individuals in total. Stratified rerandomization was used for equal treatment allocation on clusters: stratification was based on gender (all individuals in the same ward have the same gender), and rerandomization was based on three binary cluster-level covariates: high or low access to mental health service, high or low disaster risk, and rural/urban status. However, we are unable to obtain detailed rerandomization parameters from the published report, including the weighting matrix or balance threshold. Therefore, we carry out our analysis assuming simple randomization: this choice will lead to conservative confidence intervals for an unadjusted analysis but has no impact on the covariate-adjusted estimators, as supported by our theory and simulations. 
The primary outcome was a continuous measure of psychological distress at the 3-month follow-up evaluated by the General Health Questionnaire (GHQ-12). The baseline variables we adjust for include the baseline GHQ-12 score and all variables used in stratified rerandomization. For this study, we implement the mixed-ANCOVA estimator (Example \ref{eg:mixed}) with individual-level data \citep{wang2021mixed}, and the unadjusted, ANCOVA, and ML estimators with cluster-level means (implemented as in the first simulation) to estimate the average treatment effect. For all estimators, we compute their point estimates, standard errors assuming simple randomization, and confidence intervals under normal approximation. 

The results are summarized in Table~\ref{tab:data1}. While all estimators have similar point estimates, their standard error estimates differ. For the unadjusted estimator, since we are unaware of the detailed rerandomization parameters, 0.78 should be a conservative estimate, leading to failure to reject the null at the 5\% level. In contrast, the standard error estimates for all three covariate-adjusted estimators should have fully accounted for the precision gain from stratified rerandomization as clarified by our theoretical results. For these estimators, it is important to highlight that validity of the standard error estimator is achieved without knowing the exact rerandomization parameters, which further endorses the recommendation of adjusting for stratification and rerandomization variables. Among the three covariate-adjusted estimators, the ANCOVA estimator appears to have the highest precision, while the machine learning estimator leads to the least variance reduction. This may be because either the sample size is relatively limited for machine learning to fully demonstrate asymptotic efficiency gain or the true data-generating distribution is nearly linear in covariates.

\begin{table}[ht]
    \centering
    \caption{Data analysis results for the GroupPMPlus study.}\label{tab:data1}
    \begin{tabular}{rrrrr}
  \hline
Estimator & Estimate & Standard error &  95\% confidence interval \\ 
  \hline
Unadjusted & -1.25 & 0.78 & (-2.77,\,\ 0.28) \\ 
  ANCOVA & -1.45 & 0.56 & (-2.55, -0.35) \\ 
  Mixed-ANCOVA & -1.35 & 0.62 & (-2.57, -0.13) \\ 
  ML & -1.44 & 0.64 & (-2.69, -0.18) \\ 
   \hline
\end{tabular}
\end{table}

\section{Discussion}\label{sec: discussion}
Covariate adjustment in randomized experiments can be achieved at the design stage, e.g., via rerandomization, and at the analysis stage, e.g., by outcome modeling. In this paper, our primary contribution is to clarify the impact of (stratified) rerandomization for the general class of M-estimators, and to further demonstrate when rerandomization can be asymptotically ignorable (and hence conventional asymptotic normality results can apply). These results provide clarifications to earlier simulation findings, for example, in cluster-randomized experiments \citep{li2016evaluation,li2017evaluation} where mixed-model ANCOVA was evaluated under rerandomization. We have also extended the theoretical development to efficient machine learning estimators, which maximally leverage baseline covariates to optimize asymptotic efficiency gain. Importantly, our results significantly expanded the existing rerandomization theory \citep{li2018asymptotic,li2020rerandomization,lu2023design,wang2023rerandomization,zhao2024no} from linear-adjusted estimators to essentially all commonly used covariate-adjusted estimators, and hence should have wide implications for randomized experiments, especially given the diversity of covariate-adjustment methods used in practice \citep{FDA2023}.

A practical question for conducting rerandomization is how to specify design parameters. 
Our general recommendation is to balance only prognostic covariates via the Mahalanobis distance and choose $t$ such that the rejection rate is no larger than 95\%. We retain our recommendation to adjust for rerandomization variables during analysis for maximum efficiency gain and convenience in inference. Finally, with our example estimators, although the choices of weighting matrices and $t$ have no impact asymptotically, they could have finite-sample implications. Alternatively, $t$ may be specified in a data-driven fashion, i.e., $t_n = f_n(\bX_1^r,\dots, \bX_n^r)$. Asymptotic analysis with an arbitrary $t_n$ is challenging as it requires special convergence results for conditional quantiles. {On this topic, \citet{wang2022rerandomization} developed asymptotic theory for sample-size-dependent thresholds \(t_n\) and discussed finite-sample diagnostics for choosing \(t_n\). More recently, \citet{WangLi2025} developed asymptotic theory for best-choice rerandomization. Although these works focus primarily on unadjusted estimators, we conjecture that their approaches can be extended to accommodate general \(M\)-estimators and machine learning estimators.}

{A potential limitation of the present theory is that it is formulated within the super-population framework routinely used for studying general covariate-adjusted estimators. This framework is useful for comparing estimators and deriving first-order properties, but it may not fully represent the finite-sample trade-offs that motivate rerandomization. For example, when covariate adjustment is constrained by model form, pre-specification requirements, or sample size, rerandomization can still improve covariate balance and may yield operating characteristics that are not fully captured by first-order asymptotic comparisons alone. A promising direction for future work is to develop refined asymptotic approximations that retain such practically relevant finite-sample features, perhaps through local asymptotic regimes in which the sequence of populations, the strength and dimension of covariate imbalance, and the degree of adjustment misspecification or regularization drift jointly.}

\section*{Acknowledgements}
Research in this article was supported by a Patient-Centered Outcomes Research Institute Award\textsuperscript{\textregistered} (PCORI\textsuperscript{\textregistered} Award ME-2022C2-27676) and the National Institute Of Allergy And Infectious Diseases of the National Institutes of Health under Award Number R00AI173395. The statements presented in this article are solely the responsibility of the authors and do not necessarily represent the official views of the National Institutes of Health or PCORI\textsuperscript{\textregistered}, its Board of Governors, or the Methodology Committee.

{
\bibliographystyle{apalike}
\bibliography{references}
}

\end{document}


\def\spacingset#1{\renewcommand{\baselinestretch}%
{#1}\small\normalsize} \spacingset{1}

\date{\vspace{-5ex}}

\maketitle

\spacingset{1.8}
\appendix

\section{Consistent variance estimators}
For M-estimators, we construct consistent variance estimators $\widehat{V}$ for $V$ in Equation~\eqref{def: vhat}, $\widehat{R}^2$ for $R^2$ in Equation~\eqref{def: rhat}, $\widehat{\widetilde{V}}$ for $\widetilde{V}$ in Equation~\eqref{def: vv-hat}, and  $\widehat{\widetilde{R}^2}$ for $\widetilde{R}^2$ in Equation~\eqref{def: rr-hat}. 
For efficient estimators, we construct consistent variance estimators $\widehat{V}$ for $V$ in Equation~\eqref{Vhat: ml}, $\widehat{R}^2$ for $R^2$ in Equation~\eqref{R2hat: ml}, $\widehat{\widetilde{V}}$ for $\widetilde{V}$ in Equation~\eqref{vhat: ml-st}, and  $\widehat{\widetilde{R}^2}$ for $\widetilde{R}^2$ in Equation~\eqref{R2hat: ml-st}. 
We do not reiterate those expressions here for brevity; the precise equations will be presented in the subsequent proofs.

\section{Proofs}
\subsection{Results under rerandomization}
\subsubsection{Lemmas}

\begin{lemma}[Conditional Lindeberg-Feller central limit theorem]\label{lem:condCLT}
Let $\{\bX_{ni}\}_{i=1}^n\subset\mathbb R^d$ be random vectors adapted to some sigma-field $\mathcal F_n$ and denote $\bS_n=\sum_{i=1}^n \bX_{ni}$. Assume (1) $\{\bX_{ni}\}_{i=1}^n$ are independent given $\mathcal F_n$, (2) $E[\bS_n\mid \mathcal F_n]=0$, (3)  $\bSigma_n=\sum_{i=1}^n E\!\big[\bX_{ni}\bX_{ni}^\top\mid \mathcal F_n\big]$ is finite almost surely, and (4)  the conditional Lindeberg condition in probability:
for each $\varepsilon>0$ and unit-lengh vector $\bu$, $\sum_{i=1}^n E\!\Big[(\bu^\top \bX_{ni})^2\,\mathbf 1\{|\bu^\top \bX_{ni}|>\varepsilon\}\ \Big|\ \mathcal F_n\Big]\xrightarrow{p}0.$

Then, for each unit $\bu$ and $t\in\mathbb R$, we have $E\!\left[e^{\,it\,\bu^\top \bS_n}\ \big|\ \mathcal F_n\right]-\exp\!\left(-\frac12 t^2\,\bu^\top\bSigma_n \bu\right)\ \xrightarrow{p}\ 0$; that is, $\bS_n$ and $N(0,\bSigma_n)$ are asymptotically equivalent conditioning on $\mathcal{F}_n$, which we denote as $\bS_n | \mathcal{F}_n \overset{\bullet}{\sim} N(0,\bSigma_n) | \mathcal{F}_n$.  
If moreover $\bSigma_n \xrightarrow{p} \bSigma$, then $\bS_n \xrightarrow{d} N(0,\bSigma)$ unconditionally.
\end{lemma}

\begin{proof}
This is the standard Lindeberg--Feller central limit theorem carried out with conditional characteristic functions (e.g.,
\citealp{aldous1978mixing} on stable convergence).
\end{proof}

\begin{lemma}\label{lemma1}
    Let $h(\bW;\btheta)$ be some measurable real-valued function $h$ on $\bW$ defined in Assumption 1 and parameters $\btheta \in \bTheta \subset \mathbb{R}^p$. 
    We assume that $\bTheta$ is a compact set, $h$ is a continuous function on $\btheta$, and  $h(\bW;\btheta)$ is dominated by an integrable function. 
    Under rerandomization and Assumption 1, we have $\sup_{\btheta \in \bTheta}||\frac{1}{n} \sum_{i=1}^n A_i h(\bW_i;\btheta) - \pi E[h(\bW_i;\btheta)] ||_2 \xrightarrow{p} 0$ and $\sup_{\btheta \in \bTheta}||\frac{1}{n} \sum_{i=1}^n (1-A_i) h(\bW_i;\btheta) - (1-\pi) E[h(\bW_i;\btheta)] ||_2 \xrightarrow{p} 0$.
\end{lemma}

\begin{proof}
Define $\mathcal{F}_n = \sigma(\bW_1,\dots,\bW_n)$ as the sigma-field generated under Assumption 1. For brevity, we further denote $U_n(A,\bW) = \sup_{\btheta \in \bTheta}||\frac{1}{n} \sum_{i=1}^n A_i h(\bW_i;\btheta) - \pi E[h(\bW_i;\btheta)]||_2$. Our goal is to show $P(U_n(A,\bW) > \varepsilon) \rightarrow 0$ for any $\varepsilon >0$, and it suffices to prove $P(U_n(A,\bW) > \varepsilon | \mathcal{F}_n) \xrightarrow{p} 0$ due to the dominated convergence theorem for convergence in probability \citep[Theorem 1.5.3]{durrett2019probability}. 

Using the definition of rerandomization, we have 
\begin{align*}
    P\left(U_n(A,\bW) > \varepsilon | \mathcal{F}_n\right) &= P\left(U_n(A^*,\bW) > \varepsilon | \mathcal{F}_n, \bI^\top \{\widehat{Var}(\bI)\}^{-1} \bI < t\right) \\
    &= \frac{ P\left(U_n(A^*,\bW) > \varepsilon, \bI^\top \{\widehat{Var}(\bI)\}^{-1} \bI < t | \mathcal{F}_n\right)}{P\left(\bI^\top \{\widehat{Var}(\bI)\}^{-1} \bI < t | \mathcal{F}_n\right)} \\
    &\le \frac{ P\left(U_n(A^*,\bW) > \varepsilon | \mathcal{F}_n\right)}{P\left(\bI^\top \{\widehat{Var}(\bI)\}^{-1} \bI < t | \mathcal{F}_n\right)},
\end{align*}
where $A^*$ represents treatment allocation under simple randomization and $\bI = \frac{1}{N_1} \sum_{i=1}^n A_i^* \bX^r_i - \frac{1}{N_0} \sum_{i=1}^n (1-A_i^*) \bX^r_i$. \underline{For the numerator term}, $A_i^* h(\bW_i;\btheta), i = 1,\dots, n$ are independent and identically distributed. Since we assume that $\bTheta$ is a compact set, $h$ is a continuous function on $\btheta$, and  $h(\bW;\btheta)$ is dominated by an integrable function, then Example 19.8 of \cite{vaart_1998} implies that $A^* h(\bW;\btheta)$ is P-Glinveko-Cantelli, which ensures that $P\left(U_n(A^*,\bW) > \varepsilon\right) \rightarrow 0$. Applying Markov's inequality, we have $P\left(U_n(A^*,\bW) > \varepsilon| \mathcal{F}_n\right) \xrightarrow{p} 0$.
\underline{For the denominator term}, we define $\bX_{in} = (A_i^*-\pi, (A_i^*-\pi)\bX_i^r)$, which are conditionally independent across $i$ given $\mathcal{F}_n$ (due to simple randomization), mean-zero, and satisfy the Lindberg condition due to the finite-second-moment assumption. Lemma~\ref{lem:condCLT} implies the conditionally asymptotic normality of $\bX_{in}$, and we apply the Delta method to get $\sqrt{n}\bI |\mathcal{F}_n \overset{\bullet}{\sim} N(0, \bV_{I,n}) |\mathcal{F}_n$, where $\bV_{I,n} = \frac{1}{\pi(1-\pi)} \frac{1}{n}\sum_{i=1}^n (X_{r,i}-\overline X_r)(X_{r,i}-\overline X_r)^\top$. Since $n\widehat{Var}(\bI) - \bV_{I,n} \xrightarrow{p}$ conditioning on $\mathcal{F}_n$, conditional Slutsky's theorem implies $\bI^\top \{\widehat{Var}(\bI)\}^{-1} \bI|\mathcal{F}_n\overset{\bullet}{\sim}\mathcal{X}^2_q|\mathcal{F}_n$, where $q$ is the dimension of $\bX^r$. Therefore, $P(\bI^\top \{\widehat{Var}(\bI)\}^{-1} \bI < t|\mathcal{F}_n) \rightarrow P(\mathcal{X}^2_q < t|\mathcal{F}_n) = P(\mathcal{X}^2_q < t)$, which is greater than zero since $t > 0$. Combining the convergence of both numerator and denominator terms, we obtain $P\left(U_n(A,\bW) > \varepsilon | \mathcal{F}_n\right) \xrightarrow{p} \frac{0}{P(\mathcal{X}^2_q < t)} = 0$. This completes the proof for $\sup_{\btheta \in \bTheta}||\frac{1}{n} \sum_{i=1}^n A_i h(\bW_i;\btheta) - \pi E[h(\bW_i;\btheta)] ||_2 \xrightarrow{p} 0$.


Following similar steps, we can get $\sup_{\btheta \in \bTheta}||\frac{1}{n} \sum_{i=1}^n (1-A_i) h(\bW_i;\btheta) - (1-\pi) E[h(\bW_i;\btheta)] ||_2 \xrightarrow{p} 0$.
\end{proof}

\begin{lemma}\label{lemma2}
    Let $\bZ_i(a) = h_a(Y_i(a), M_i(a), \bX_i)$  for some $k$-dimensional function $h_a, a \in \{0,1\}$ taking values in the Euclidean space such that $E[||\bZ_i(a)||_2^2] < \infty$. Under rerandomization and Assumption 1, we have
    \begin{align*}
        \frac{1}{\sqrt{n}}\sum_{i=1}^n \left\{A_i\bZ_i(1) + (1-A_i)\bZ_i(0) - \pi E[\bZ_i(1)] - (1-\pi) E[\bZ_i(0)]\right\} \xrightarrow{d}  \boldsymbol{\varepsilon}_k + \bfV_{UI}^\top \bfV_I^{-0.5} \bR_{q,t},
    \end{align*}
    where $\boldsymbol{\varepsilon}_k \sim N(\bzero,\bfV_U -  \bfV_{UI}^\top \bfV_I^{-1}\bfV_{UI})$, $\bR_{q,t} \sim  \bD_q|( \bD^\top  \bD < t)$,  where $\bD \sim N(\bzero, \mathbf{I}_q)$ with
    \begin{align*}
        \bfV_U &= Var\{A_i^*\bZ_i(1) + (1-A_i^*)\bZ_i(0)\} \\
        \bfV_I &= \frac{1}{\pi(1-\pi)}Var(\bX_i^r) \\
        \bfV_{UI} &= Cov\left\{\frac{A_i^*-\pi}{\pi(1-\pi)}(\bX_i^r-E[\bX_i^r]), A_i^*\bZ_i(1) + (1-A_i^*)\bZ_i(0)\right\}\\
        &= Cov\{\bX_i^r, \bZ_i(1)-\bZ_i(0)\}.
    \end{align*}
\end{lemma}
\begin{proof}
We first introduce some notation as follows:
\begin{align*}
    \bI_n(A) &= \sqrt{n} \left(\frac{\sum_{i=1} A_i \bX_i^r}{\sum_{i=1} A_i} - \frac{\sum_{i=1} (1-A_i) \bX_i^r}{\sum_{i=1} (1-A_i)}\right), \\
    \mathcal A_n &=\Big\{\, \bI_n(A^*)^\top\,\widehat{\mathrm{Var}}(\bI_n(A^*))^{-1}\,\bI_n(A^*)< t\,\Big\},\\
    \bU_n(A) &=\frac{1}{\sqrt n}\sum_{i=1}^n\!\Big\{A_i \bZ_i(1)+(1-A_i)\bZ_i(0)-\pi E\bZ(1)-(1-\pi)E\bZ(0)\Big\},\\
    \bT_n(A) &= \frac{1}{\sqrt n}\sum_{i=1}^n\!\Big\{(A_i-\pi) (\bZ_i(1)-\bZ_i(0)))\Big\},\\
    \bR_n &= \frac{1}{\sqrt n}\sum_{i=1}^n\!\Big\{\pi \bZ_i(1)+(1-\pi)\bZ_i(0)-\pi E\bZ(1)-(1-\pi)E\bZ(0)\Big\}.
\end{align*}
Then we have $\bU_n(A) = \bT_n(A) + \bR_n$. We further define $\mathcal{F}_n = \sigma(\bW_1,\dots,\bW_n)$ as the sigma-field generated under Assumption 1, and $X \overset{d}{=} Y$ if $X$ and $Y$ have the same distribution. By the definition of rerandomization, we have $g(A,\bW)| \mathcal{F}_n \overset{d}{=} g(A^*,\bW)| (\mathcal{F}_n, \mathcal{A}_n)$ for any measurable function $g$.

We first provide the conditional Normal approximation under simple randomization. To utilize Lemma~\ref{lem:condCLT}, we define $\bX_{ni} = \{A_i^*-\pi, (A_i^*-\pi)\bX_i^r, (A_i^*-\pi)(\bZ_i(1)-\bZ_i(0))\}$, which are conditionally independent across $i$ given $\mathcal{F}_n$ (due to simple randomization), mean-zero, and satisfies the Lindberg condition due to the finite-second-moment assumption. Then, Lemma~\ref{lem:condCLT} implies that $\bX_{ni}$ is asymptotically equivalent to a Gaussian distribution, conditioning on $\mathcal{F}_n$. 
Since $\bI_n(A^*), \bT_n(A^*)$ are functions of $\bX_{ni}$, the Delta method implies $(\bI_n(A^*), \bT_n(A^*)) | \mathcal{F}_n \overset{\bullet}{\sim} N(0, \bSigma_n) | \mathcal{F}_n$, where
\[
\bSigma_n=
\begin{bmatrix}
\displaystyle \bV_{I,n} & \displaystyle \bV_{IT,n}\\[4pt]
\displaystyle \bV_{IT,n}^\top & \displaystyle \bV_{T,n}
\end{bmatrix},
\]
with blocks (all $\mathcal F_n$-measurable sample moments)
\[
\bV_{I,n}=\frac{1}{\pi(1-\pi)}\cdot \frac{1}{n}\sum_{i=1}^n (\bX_i^r-\overline \bX_r)(\bX_i^r-\overline \bX_r)^\top,\quad
\bV_{IT,n}=\frac{1}{n}\sum_{i=1}^n (\bX_i^r-\overline \bX_r)\,(\bZ_i(1)-\bZ_i(0))^\top,
\]
\[
\bV_{T,n}=\pi(1-\pi)\cdot \frac{1}{n}\sum_{i=1}^n (\bZ_i(1)-Z_i(0))(\bZ_i(1)-\bZ_i(0))^\top.
\]
Following a similar procedure, we also get $n\widehat{Var}(\bI_n(A^*)) - \bV_{I,n} \xrightarrow{p} 0$ and
\begin{equation}\label{eq: ivi}
\{\bT_n(A^*), \bI_n(A^*)^\top n\widehat{Var}(\bI_n(A^*))^{-1}\bI_n(A^*)\} | \mathcal{F}_n\overset{\bullet}{\sim} \left\{\bT_n(A^*), \bI_n(A^*)^\top \bV_{I,n}^{-1}\bI_n(A^*) \right\} | \mathcal{F}_n
\end{equation}
as they are functions of $\bX_{ni}$.

We next derive the asymptotic distribution of $\bT_n(A^*) | \mathcal{F}_n, \mathcal{A}_n$. Denoting $\bvarepsilon_n = \bT_n(A^*) -  \bV_{IT,n}^\top   \bV_{I,n}^{-1}\bI_n(A^*)$, our result of $(\bI_n(A^*), \bT_n(A^*)) | \mathcal{F}_n \overset{\bullet}{\sim} N(0, \bSigma_n) | \mathcal{F}_n$ implies that $\bvarepsilon_n | \mathcal{F}_n \overset{\bullet}{\sim} N(0, \bV_{T,n} - \bV_{IT,n}^\top \bV_{I,n}^{-1}\bV_{IT,n}) | \mathcal{F}_n$ and that $\varepsilon_n$ is asymptotically independent of $I_n(A^*)$ given $\mathcal{F}_n$. Combined with Equation~\eqref{eq: ivi}, we have
\begin{align*}
    \bT_n(A^*)| \mathcal{F}_n, \mathcal{A}_n &\overset{\bullet}{\sim} \bT_n(A^*)| \mathcal{F}_n, \bI_n(A^*)^\top \bV_{I,n}^{-1}\bI_n(A^*)< t \\
    &\overset{\bullet}{\sim} \bvarepsilon_n| \mathcal{F}_n +  \bV_{IT,n}^\top   \bV_{I,n}^{-1}\bI_n(A^*)| \mathcal{F}_n, \bI_n(A^*)^\top \bV_{I,n}^{-1}\bI_n(A^*)< t \\
    &\overset{\bullet}{\sim}\bvarepsilon_n| \mathcal{F}_n +  \bV_{IT,n}^\top   \bV_{I,n}^{-0.5} \bR_{q,t} | \mathcal{F}_n,
\end{align*}
where the last step comes from $\bV_{I,n}^{-0.5} \bI_n(A^*)| \mathcal{F}_n \overset{\bullet}{\sim} \mathcal{N}(0, \bI_q)| \mathcal{F}_n$ and the definition of $\bR_{q,t}$ in the Lemma statement.

We then obtain
\begin{align*}
    \bU_n(A) |\mathcal{F}_n  &\overset{\bullet}{\sim} \bT_n(A) | \mathcal{F}_n\quad  + \quad  \bR_n  | \mathcal{F}_n \\
    &\overset{\bullet}{\sim} \bT_n(A^*)| (\mathcal{F}_n, \mathcal{A}_n) \quad + \quad \bR_n  | \mathcal{F}_n \\
    &\overset{\bullet}{\sim} (\bvarepsilon_n +  \bV_{IT,n}^\top   \bV_{I,n}^{-0.5} \bR_{q,t} + \bR_n)  | \mathcal{F}_n,
\end{align*}
where the first step is by definition, the second step is implied by rerandomization, and the last step uses the above results. 

Finally, we derive the asymptotic distribution of $\bU_n(A)$.
Conditionally on $\mathcal F_n$, $\bvarepsilon_n$ is independent of $\bI_n(A^*)$ and hence of
$\bR_{q,t}$. Since $\bR_n$ is $\mathcal F_n$-measurable, the triple
$(\bvarepsilon_n,\ \bR_{q,t},\ \bR_n)$ is conditionally independent in the sense that, given $\mathcal F_n$,
the joint conditional characteristic function factorizes:
\begin{align*}
& E\!\left[\exp\!\Big\{i\big(\bt_1^\top \bvarepsilon_n
+ \bt_2^\top \bV_{IT,n}^\top   \bV_{I,n}^{-0.5} \bR_{q,t}  + \bt_3^\top \bR_n)\big)\Big\}\ \Big|\ \mathcal F_n\right]\\
&=
\underbrace{E\!\left[\exp\{\,i \bt_1^\top \bvarepsilon_n\}\ \Big|\ \mathcal F_n\right]}_{=: \psi_{n}(\bt_1)}
\underbrace{E\!\left[\exp\{\,i \bt_2^\top  \bV_{IT,n}^\top   \bV_{I,n}^{-0.5} \bR_{q,t}\}\ \Big|\ \mathcal F_n\right]}_{=: \eta_{n}(\bt_2)}
 \exp\{\,i \bt_3^\top \bR_n\}.    
\end{align*}

For $\psi_{n}(\bt_1)$, since we have obtained $\bvarepsilon_n | \mathcal{F}_n \overset{\bullet}{\sim} N(0, V_{T,n} - \bV_{IT,n}^\top \bV_{I,n}^{-1}\bV_{IT,n}) | \mathcal{F}_n$, and by law of large numbers, $\bV_{T,n} \xrightarrow{p} Var\{(A^*-\pi)(\bZ(1)-\bZ(0)\}$, $\bV_{IT,n} \xrightarrow{p} \bV_{UI}$, $\bV_{I,n} \xrightarrow{p} \bV_I$, then $\psi_{n}(\bt_1) \xrightarrow{p} \exp\!\Big(-\tfrac12 \bt_1^\top \bSigma_e\, \bt_1\Big)$, where  $\bSigma_e = Var\{(A^*-\pi)(\bZ(1)-\bZ(0)\} - \bV_{UI}^\top \bV_I^{-1} \bV_{UI}$.
For $\eta_n(\bt_2)$, a similar prove shows that $\eta_n(\bt_2) = E\!\big[\exp\{\,i \bt_2^\top \bV_{UI}^\top \bV_I^{-1/2} \bR_{q,t}\}\big]$
Finally, by the classical (unconditional) CLT across units,
$\bR_n\ \xrightarrow{d}  N(0,\ \bSigma_r)$ with $\bSigma_r = Var\left\{\pi \bZ(1) + (1-\pi) \bZ(0)\right\}$.

To obtain the unconditional characteristic function of $\bU_n(A)$, we take the unconditional expectation and use the dominated convergence theorem to get
\begin{align*}
    & E\!\left[\exp\!\Big\{i\big(\bt_1^\top \bvarepsilon_n
+ \bt_2^\top \bV_{IT,n}^\top   \bV_{I,n}^{-0.5} \bR_{q,t}  + \bt_3^\top \bR_n)\big)\Big\}\ \right]\\
&\rightarrow \exp\!\Big(-\tfrac12 \bt_1^\top \Sigma_e\, \bt_1\Big) \eta_n(\bt_2) E[\exp\{i\bt_3^\top \bR_n\}] \\
&= \exp\!\Big\{-\tfrac12 \bt_1^\top (\bSigma_e +\bSigma_r) \bt_1\Big\} E\!\big[\exp\{\,i \bt_2^\top \bV_{UI}^\top \bV_I^{-1/2} \bR_{q,t}\}\big].
\end{align*}
This result shows that $\bU_n(A)$ coverages in distribution to $\bvarepsilon_k + \bV_{UI}^\top \bV_I^{-1/2} \bR_{q,t}$ with $\bvarepsilon_k \sim N(0, \bSigma_e +\bSigma_r)$ and $\bR_{q,t}$ being independent  of $\bvarepsilon_k$. 
Since direct algebra shows that 
\begin{align*}
    \bSigma_e +\bSigma_r &= Var\{(A^*-\pi)(\bZ(1)-\bZ(0)\} - \bV_{UI}^\top \bV_I^{-1} \bV_{UI} + Var\left\{\pi \bZ(1) + (1-\pi) \bZ(0)\right\} \\
    &= Var(A^* \bZ(1) + (1-A^*) \bZ(0)) - \bV_{UI}^\top \bV_I^{-1} \bV_{UI},
\end{align*}
we obtain the desired asymptotic distribution.
\end{proof}

\subsubsection{Proof of Theorem 1}

\begin{proof}[Proof of Theorem 1.]
The results under simple randomization can be directly obtained by Chapter 5 of \cite{vaart_1998}. Therefore, we focus on rerandomizaiton in this proof. We denote $\bO_i(a) = \{R_i(a), R_i(a)Y_i(a), a, \bX_i\}$ for $a  = 0,1$.

In the first step, we prove $\widehat{\btheta} \xrightarrow{p} \underline{\btheta}$, which implies $\widehat{\Delta} \xrightarrow{p} \underline{\Delta}$. The estimating equations for M-estimator can be rewritten as
\begin{align*}
    \bzero = \sum_{i=1}^n \bpsi(\bO_i;\btheta) = \sum_{i=1}^n [A_i\bpsi\{\bO_i(1);\btheta\} + (1-A_i) \bpsi\{\bO_i(0);\btheta\}].
\end{align*}
By regularity conditions (1) and (4), we can apply Lemma~\ref{lemma1} by setting $h(\bW_i; \btheta) = \bpsi\{\bO(a);\btheta\}$ and get $\sup_{\btheta \in \bTheta}||\frac{1}{n} \sum_{i=1}^n A_i \bpsi\{\bO_i(1);\btheta\} - \pi E[\bpsi\{\bO_i(1);\btheta\}] ||_2 \xrightarrow{p} 0$ and $\sup_{\btheta \in \bTheta}||\frac{1}{n} \sum_{i=1}^n (1-A_i) \bpsi\{\bO_i(0);\btheta\} - (1-\pi) E[\bpsi\{\bO_i(0);\btheta\}] ||_2 \xrightarrow{p} 0$. This implies that
\begin{align*}
    &\sup_{\btheta \in \bTheta}||\frac{1}{n} \sum_{i=1}^n\bpsi(\bO_i;\btheta) - \pi E[\bpsi(\bO_i(1);\btheta)] - (1-\pi)E[\bpsi(\bO_i(0);\btheta)]   ||_2 \\
    &\le \sup_{\btheta \in \bTheta}||\frac{1}{n} \sum_{i=1}^n A_i \bpsi\{\bO_i(1);\btheta\} - \pi E[\bpsi\{\bO_i(1);\btheta\}] ||_2\\
    &\quad + \sup_{\btheta \in \bTheta}||\frac{1}{n} \sum_{i=1}^n (1-A_i) \bpsi\{\bO_i(0);\btheta\} - (1-\pi) E[\bpsi\{\bO_i(0);\btheta\}] ||_2\\
    &= o_p(1) + o_p(1) \\
    &= o_p(1).
\end{align*}
Since $\frac{1}{n} \sum_{i=1}^n\bpsi(\bO_i;\widehat{\btheta})=\bzero$ and regularity condition (3) implies that $\underline{\btheta}$ is a unique solution to $\pi E[\bpsi(\bO_i(1);\btheta)] + (1-\pi)E[\bpsi(\bO_i(0);\btheta)] = \bzero$, Theorem 5.9 of \cite{vaart_1998} implies that $\widehat{\btheta} \xrightarrow{p} \underline{\btheta}$.

Next, we prove the asymptotic linearity of $\widehat{\btheta}$ following the general approach in Theorem 5.41 of \cite{vaart_1998}. By multivariate Taylor expansion of the function $\sum_{i=1}^n \bpsi\{\bO_i;\btheta\}$ on $\btheta = \underline{\btheta}$, there exists a $\widetilde{\btheta}$ (that depends on $\widehat{\btheta}$) on the line segment between $\widehat{\btheta}$ and $\underline{\btheta}$ such that
\begin{align*}
    \bzero &= \frac{1}{n}\sum_{i=1}^n \bpsi\{\bO_i;\widehat{\btheta}\} \\
    &= \frac{1}{n}\sum_{i=1}^n \bpsi(\bO_i;\underline{\btheta})+ \frac{1}{n}\sum_{i=1}^n \dot{\bpsi}(\bO_i;\underline{\btheta}) (\widehat{\btheta} - \underline{\btheta}) + \sum_{j=1}^{l+1}\frac{1}{2}(\widehat{\btheta} - \underline{\btheta})^\top \left\{\frac{1}{n}\sum_{i=1}^n \Ddot{\psi}_j(\bO_i;\widetilde{\btheta})\right\}(\widehat{\btheta} - \underline{\btheta})\boldsymbol{e}_j,
\end{align*}
where $\dot{\bpsi}(\bO_i;\underline{\btheta}) = \frac{\partial}{\partial \btheta}\bpsi(\bO_i;\btheta)\Big|_{\btheta=\underline{\btheta}}$, $\psi_j$ is the $j$-th element of $\bpsi$, $\Ddot{\psi}_j(\bO_i;\widetilde{\btheta}) = \frac{\partial^2}{\partial\btheta\partial\btheta^\top} \psi_j(\bO_i;\btheta)\Big|_{\btheta=\widetilde{\btheta}}$, and $\boldsymbol{e}_j$ is a $(l+1)$-dimensional vector with the $j$-th entry 1 and the rest 0. By regularity conditions (1), (4), and (5), we can apply Lemma 1 again to obtain $\frac{1}{n}\sum_{i=1}^n\dot{\bpsi}(\bO_i;\underline{\btheta}) = \pi E[\dot{\bpsi}\{\bO_i(1);\underline{\btheta}\}]+ (1-\pi) E[\dot{\bpsi}\{\bO_i(0);\underline{\btheta}\}] + o_p(\bone) = \bfB + o_p(\bone)$. Since $\widehat{\btheta} \xrightarrow{p} \underline{\btheta}$, Regularity (5) implies that, when $n$ is large enough,  $||\frac{1}{n}\sum_{i=1}^n \Ddot{\psi}_j(\bO_i;\widetilde{\btheta})||_2 \le \frac{1}{n}\sum_{i=1}^n v(\bO_i) = O_p(1)$ by law of large numbers. Combination of the above facts implies that 
\begin{equation}\label{eq: proof1-1}
    \frac{1}{n}\sum_{i=1}^n \bpsi(\bO_i;\underline{\btheta}) = - \{\bfB + o_p(1)\} (\widehat{\btheta} - \underline{\btheta}) - O_p(\bone) ||\widehat{\btheta} - \underline{\btheta}||_2^2.
\end{equation}
To obtain the desired asymptotic linearity, we still need to show $\sqrt{n}(\widehat{\btheta} - \underline{\btheta}) = O_p(\bone)$. To see this, regularity condition (2) states $E[||\bpsi(\bO(a);\underline{\btheta})||_2] < \infty$, and regularity condition (3) states $\pi E[\bpsi(\bO_i(1);\underline\btheta)] + (1-\pi)E[\bpsi(\bO_i(0);\underline\btheta)] = \bzero$. We can then apply Lemma~\ref{lemma2} with $\bZ_i(a) = \bpsi(\bO_i(a);\underline{\btheta})$ and get
\begin{align*}
    \frac{1}{\sqrt{n}}\sum_{i=1}^n \bpsi(\bO_i;\underline{\btheta}) = \frac{1}{\sqrt{n}}\sum_{i=1}^n [A_i\bpsi\{\bO_i(1);\underline{\btheta}\} + (1-A_i)\bpsi\{\bO_i(0);\underline{\btheta}\}] = O_p(\bone).
\end{align*}
Multiplying $\sqrt{n}$ to both sides of Equation~\eqref{eq: proof1-1}, we get
\begin{align*}
    O_p(\bone) = - \{\bfB + o_p(1)\} \sqrt{n}(\widehat{\btheta} - \underline{\btheta}) - O_p(\bone) \sqrt{n}||\widehat{\btheta} - \underline{\btheta}||_2^2.
\end{align*}
Since $\widehat{\btheta} - \underline{\btheta} = o_p(\bone)$ and $\bfB$ is invertible (regularity condition (6)), we have
\begin{align*}
     O_p(\bone) = - \{\bfB + o_p(1) + O_p(\bone)o_p(\bone)\} \sqrt{n}||\widehat{\btheta} - \underline{\btheta}||_2,
\end{align*}
implying $\sqrt{n}(\widehat{\btheta} - \underline{\btheta}) = O_p(\bone)$. Therefore, Equation~\eqref{eq: proof1-1} implies
\begin{align*}
    \sqrt{n}(\widehat{\btheta} - \underline{\btheta}) = - \frac{1}{\sqrt{n}}\sum_{i=1}^n \bfB^{-1} \bpsi(\bO_i;\underline{\btheta}) + o_p(\bone),
\end{align*}
which directly yields the asymptotic linearity of $\widehat{\Delta}$.

We then derive the asymptotic distribution for $\widehat{\Delta}$. By regularity condition (2), $IF(\bO(a))$ has finite second moments. Furthermore, $\pi E[IF(\bO(1))] + (1-\pi) E[IF(\bO(0))] = \bzero$ by the definition of $\underline{\btheta}$. Applying Lemma~\ref{lemma2} with $\bZ_i(a) = IF(\bO_i(a))$, we obtain 
\begin{align*}
    \sqrt{n}(\widehat{\Delta} - \underline{\Delta}) \xrightarrow{d} \varepsilon_1 + \bfV_{UI}^\top \bfV_I^{-0.5} \bR_{q,t},
\end{align*}
where $\varepsilon_1 \sim N(0, V-\bfV_{UI}^\top \bfV_I^{-1}\bfV_{UI})$ and $\bR_{q,t} \sim  \bD|( \bD^\top  \bD < t)$,  where $\bD \sim N(\bzero, \mathbf{I}_q)$ and $\varepsilon_1$ is independent of $\bR_{q,t}$. Since $V = Var\{A^* IF(\bO(1))+(1-A^*)IF(\bO(0))\}$, $V$ is the asymptotic variance under simple randomization. Since $\bfV_{UI}^\top \bfV_I^{-0.5}$ is a vector, Lemma A1 of \cite{li2018asymptotic} shows that $\bfV_{UI}^\top \bfV_I^{-0.5} \bR_q \overset{d}{=} \sqrt{\bfV_{UI}^\top \bfV_I^{-1} \bfV_{UI}}\ D_1|\bD^\top \bD < t$. Defining
\begin{align*}
    R^2 &= \frac{\bfV_{UI}^\top \bfV_I^{-1} \bfV_{UI}}{V}\\
    &= \frac{\pi(1-\pi)}{V} Cov[IF\{\bO(1)\}- IF\{\bO(0)\}, \bX^r]Var(\bX^r)^{-1}Cov[\bX^r, IF\{\bO(1)\}- IF\{\bO(0)\}],
\end{align*}
we get the desired asymptotic distribution. 

In the final step, we construct consistent estimators for $V$ and $R^2$ based on sandwich variance estimation. First, define $\widehat{IF}_i$ as the first element of $\left\{ \frac{1}{n}\sum_{i=1}^n \dot{\bpsi}(\bO_i; \widehat{\btheta}) \right\}^{-1} \bpsi(\bO_i; \widehat{\btheta})$ and
\begin{equation}\label{def: vhat}
    \widehat{V} = \frac{1}{n}\sum_{i=1}^n \widehat{IF}_i^2
\end{equation}
Since $\widehat{V}$ is the first-row first-column entry of 
\begin{equation*}
   \left\{ \frac{1}{n}\sum_{i=1}^n \dot{\bpsi}(\bO_i; \widehat{\btheta}) \right\}^{-1} \left\{\frac{1}{n}\sum_{i=1}^n \bpsi(\bO_i; \widehat{\btheta})\bpsi(\bO_i; \widehat{\btheta})\right\}\left\{ \frac{1}{n}\sum_{i=1}^n \dot{\bpsi}(\bO_i; \widehat{\btheta}) \right\}^{-1\ T},
\end{equation*}
To show $\widehat{V} \xrightarrow{p} V$, it suffices to show $\frac{1}{n}\sum_{i=1}^n \dot{\bpsi}(\bO_i; \widehat{\btheta}) \xrightarrow{p} \bfB$ and $\frac{1}{n}\sum_{i=1}^n \bpsi(\bO_i; \widehat{\btheta})\bpsi(\bO_i; \widehat{\btheta}) \xrightarrow{p} \pi E[\bpsi(\bO_i(1); \widehat{\btheta})\bpsi(\bO_i(1); \widehat{\btheta})^\top] + (1-\pi)E[\bpsi(\bO_i(0); \underline{\btheta})\bpsi(\bO_i(0); \underline{\btheta})^\top]$. We next detail the proof for the former argument, and the latter can be derived following a similar way. Denoting $\dot{\bpsi}_{j}(\bO_i;  \widehat{\btheta}) = \dot{\psi}_j(\bO_i;  \widehat{\btheta})$, which is also the transpose of the $j$-th row of $\dot{\bpsi}(\bO_i;  \widehat{\btheta})$. By multivariate Taylor's expansion of $\dot{\bpsi}_{j}(\bO_i;  {\btheta})$ at $\btheta = \underline{\btheta}$, we have
\begin{align*}
    \frac{1}{n} \sum_{i=1}^n \dot{\bpsi}_{j}(\bO_i;  \widehat{\btheta}) - \frac{1}{n} \sum_{i=1}^n \dot{\bpsi}_{j}(\bO_i; \underline{\btheta}) = \frac{1}{n} \Ddot{\bpsi}_j(\bO_i; \widetilde{\btheta}) (\widehat{\btheta}- \underline{\btheta})
\end{align*}
for some $\widetilde{\btheta}$ on the line segment between $\widehat{\btheta}$ and $\underline{\btheta}$. By regularity condition (5) and $\widehat{\btheta}- \underline{\btheta} = o_p(\bone)$, we get $\frac{1}{n} \sum_{i=1}^n \dot{\bpsi}_{j}(\bO_i;  \widehat{\btheta}) - \frac{1}{n} \sum_{i=1}^n \dot{\bpsi}_{j}(\bO_i; \underline{\btheta}) = o_p(\bone)$. Applying Lemma~\ref{lemma1} with $h_a(\bW_i;\btheta) =  \dot{\bpsi}_{j}(\bO_i(a); \underline{\btheta})$, we get $\frac{1}{n} \sum_{i=1}^n \dot{\bpsi}_{j}(\bO_i; \underline{\btheta}) = \bfB + o_p(1)$. Therefore, $\frac{1}{n}\sum_{i=1}^n \dot{\bpsi}(\bO_i; \widehat{\btheta}) \xrightarrow{p} \bfB$. 

For estimating $R^2$, we define 
\begin{equation}\label{def: rhat}
    \widehat{R^2} = \frac{1}{\widehat{V}} \left\{\frac{1}{n}\sum_{i=1}^n \frac{A_i-\pi}{\pi(1-\pi)}\widehat{IF}_i(\bX_i^r-\overline{\bX}^r)^\top\right\}\{n\widehat{Var}(\bI)\}^{-1} \left\{\frac{1}{n}\sum_{i=1}^n \frac{A_i-\pi}{\pi(1-\pi)}\widehat{IF}_i(\bX_i^r-\overline{\bX}^r)\right\}.
\end{equation}
The result $\widehat{R^2} \xrightarrow{p} R^2$ can be proved in a similar way to $\widehat{V} \xrightarrow{p} V$. 
\end{proof}

\subsubsection{Proof of Theorem 2}
\begin{proof}
    Since the ANCOVA and logistic regression estimators can be viewed as special cases of the DR-WLS estimator, for brevity, we focus on the proof for the DR-WLS estimator here. The proof for the mixed-ANCOVA estimator is given separately in Proposition~\ref{prop: mixed-ancova} for a clear presentation. For notation convenience, we denote $\bZ = (1, A, \bX^\top)^\top$, $\bZ^* = (1, A^*, \bX^\top)^\top$, $\bZ(1) = (1, 1, \bX^\top)^\top$, $\bZ(0) = (1, 0, \bX^\top)^\top$, $\balpha = (\alpha_0, \alpha_A, \balpha_X^\top)^\top$, $\widehat\balpha = (\widehat\alpha_0, \widehat\alpha_A, \widehat\balpha_X^\top)^\top$, $\bbeta = (\beta_0, \beta_A, \bbeta_X^\top)^\top$, $\widehat\bbeta = (\widehat\beta_0, \widehat\beta_A, \widehat\bbeta_X^\top)^\top$, $\underline{\balpha}$ as the probability limit of $\widehat{\balpha}$, $\underline{\bbeta}$ as the probability limit of $\widehat{\bbeta}$. Furthermore, we define $e(\bz) = \text{expit}(\underline{\balpha}^\top \bz)$, $\dot{e}(\bz)=\frac{\partial}{\partial \balpha}\text{expit}(\balpha^\top \bz)\Big|_{\balpha = \underline{\balpha}}$, $h(\bz) = g^{-1}(\underline{\bbeta}^\top \bz)$, and $\dot{h}(\bz) = \frac{\partial}{\partial \bbeta}g^{-1}(\underline{\bbeta}^\top \bz)\Big|_{\bbeta = \underline{\bbeta}}$ for $\bz = \bZ, \bZ^*, \bZ(1)$ or $\bZ(0)$. Finally, we denote $\pi_a = \pi^a(1-\pi)^{1-a}$ and $E^*[h(\bO)] = \pi E[h(\bO(1))] + (1-\pi)E[h(\bO(0))]$ for any integrable function $h$.

    Given the above notation, the DR-WLS estimator can be equivalently computed by solving estimating functions with $\btheta = (\Delta, \mu_1, \mu_0, \bbeta^\top, \balpha^\top)^\top$ and
    \begin{align*}
        \bpsi(\bO;\btheta) = \left(\begin{array}{c}
            f(\mu_1,\mu_0) - \Delta \\
            g^{-1}(\bbeta^\top \bZ(1)) - \mu_1 \\
            g^{-1}(\bbeta^\top \bZ(0)) - \mu_0 \\
            \frac{R}{\text{expit}(\balpha^\top \bZ)}\{Y - g^{-1}(\bbeta^\top \bZ)\} \bZ \\
            \{R-\text{expit}(\balpha^\top \bZ)\} \bZ
        \end{array}\right)
    \end{align*}
    since $g$ and $\text{expit}$ are both canonical link functions.

    We next show that the influence function for $\widehat{\Delta}$ is
    \begin{align}
    IF(\bO)& = \left(\dot{\underline{f}}_1\frac{A}{\pi} + \dot{\underline{f}}_0\frac{(1-A)}{1-\pi} - \boldsymbol{c}_1^\top \bZ\right)\frac{R\{Y-h(\bZ)\}}{ e(\bZ)} - \boldsymbol{c}_2^\top \bZ \{R - e(\bZ)\} \nonumber\\
    &\quad + \dot{\underline{f}}_1\{h(\bZ(1)) - \underline{\mu}_1\} + \dot{\underline{f}}_0\{h(\bZ(0)) - \underline{\mu}_0\}, \label{IF_formula}  
    \end{align}
where $\{\dot{\underline{f}}_1,  \dot{\underline{f}}_0\}$ is the gradient of $f$ evaluated at $(\underline{\mu}_1, \underline{\mu}_0)$,
\begin{align*}
    \boldsymbol{c}_1^\top  &= E^*\left[\left\{\dot{\underline{f}}_1\frac{A}{\pi} + \dot{\underline{f}}_0\frac{(1-A)}{1-\pi}\right\}\left(\frac{R}{e(\bZ)}-1\right)\dot{h}(\bZ) \right]^\top E^*\left[\frac{R}{e(\bZ)}\dot{h}(\bZ)\bZ^\top\right]^{-1},\\
    \boldsymbol{c}_2^\top &=  E[\dot{\underline{f}}_1\dot{h}(\bZ(1)) + \dot{\underline{f}}_0\dot{h}(\bZ(0))]^\top E^*\left[\frac{R}{e(\bZ)}\dot{h}(\bZ)\bZ^\top\right]^{-1} E^*\left[\frac{R\dot{e}(\bZ)}{e(\bZ)^2}\{Y - h(\bZ)\}\bZ^\top\right]E^*[ \dot{e}(\bZ)\bZ^\top]^{-1}.
\end{align*}
By the definition of $\bfB$, we have
\begin{align*}
    \bfB = \left(\begin{array}{ccccc}
        -1 & \dot{\underline{f}}_1 &  \dot{\underline{f}}_0 & \bzero & \bzero\\
        0 & -1 & 0 & E[\dot{h}(\bZ(1))] & \bzero \\
        0 & 0 & -1 & E[\dot{h}(\bZ(0))] & \bzero \\
        0 & 0 & 0 &  -E^*[\frac{R}{e(\bZ)} \dot{h}(\bZ) \bZ^\top] & -E^*[\frac{R\dot{e}(\bZ)}{e(\bZ)^2}\{Y-h(\bZ)\}\bZ^\top] \\
        0 & 0 & 0 & 0 & - E^*[\dot{e}(\bZ) \bZ^\top]
    \end{array}\right).
\end{align*}
Using the block matrix inversion formula, the first row of $-\bfB^{-1}$ is \\ $-\boldsymbol{e}_1^\top \bfB^{-1} = (1, \dot{\underline{f}}_1,  \dot{\underline{f}}_0, E[\dot{\underline{f}}_1\dot{h}(\bZ(1)) + \dot{\underline{f}}_0\dot{h}(\bZ(0))]^\top E^*\left[\frac{R}{e(\bZ)}\dot{h}(\bZ)\bZ^\top\right]^{-1}, -\boldsymbol{c}_2^\top)$, which implies
\begin{align*}
    IF(\bO)&=  f(\underline{\mu}_1,\underline{\mu}_0) - \underline{\Delta} + \dot{\underline{f}}_1\{h(\bZ(1)) - \underline{\mu}_1\} + \dot{\underline{f}}_0\{h(\bZ(0)) - \underline{\mu}_0\}\\
    &\quad + E[\dot{\underline{f}}_1\dot{h}(\bZ(1)) + \dot{\underline{f}}_0\dot{h}(\bZ(0))]^\top E^*\left[\frac{R}{e(\bZ)}\dot{h}(\bZ)\bZ^\top\right]^{-1} \bZ  \frac{R\{Y-h(\bZ)\}}{ e(\bZ)} - \boldsymbol{c}_2^\top \bZ \{R - e(\bZ)\}.
\end{align*}
To get the expression~\eqref{IF_formula}, it remains to show
\begin{align*}
    E[\dot{\underline{f}}_1\dot{h}(\bZ(1)) + \dot{\underline{f}}_0\dot{h}(\bZ(0))]^\top E^*\left[\frac{R}{e(\bZ)}\dot{h}(\bZ)\bZ^\top\right]^{-1} \bZ = \dot{\underline{f}}_1\frac{A}{\pi} + \dot{\underline{f}}_0\frac{(1-A)}{1-\pi} - \boldsymbol{c}_1^\top \bZ.
\end{align*}
To see this, we denote $\mathbf{M} = E^*\left[\frac{R}{e(\bZ)}\dot{h}(\bZ)\bZ^\top\right]$ and observe that $\mathbf{M}$ is symmetric since $\dot{h}(\bZ) = \frac{\partial}{\partial x} g^{-1}(x)\big|_{x=\underline{\bbeta}^\top \bZ} \bZ$. In addition, we also have
\begin{align*}
    E[\dot{\underline{f}}_1\dot{h}(\bZ(1)) + \dot{\underline{f}}_0\dot{h}(\bZ(0))] &= E^*\left[\left\{\dot{\underline{f}}_1\frac{A}{\pi} + \dot{\underline{f}}_0\frac{(1-A)}{1-\pi}\right\} \dot{h}(\bZ)\right] \\
    \mathbf{M} \left(\frac{\dot{\underline{f}}_0}{1-\pi}, \frac{\dot{\underline{f}}_1}{\pi}-\frac{\dot{\underline{f}}_0}{1-\pi}, \bzero\right)^\top &= E^*\left[\left\{\dot{\underline{f}}_1\frac{A}{\pi} + \dot{\underline{f}}_0\frac{(1-A)}{1-\pi}\right\} \frac{R}{e(\bZ)}\dot{h}(\bZ)\right].
\end{align*}
Then, by the definition of $\boldsymbol{c}_1$, we have
\begin{align*}
    \boldsymbol{c}_1^\top \bZ &=E^*\left[\left\{\dot{\underline{f}}_1\frac{A}{\pi} + \dot{\underline{f}}_0\frac{(1-A)}{1-\pi}\right\}\frac{R}{e(\bZ)} \dot{h}(\bZ)\right]^\top \mathbf{M}^{-1} \bZ - E^*\left[\left\{\dot{\underline{f}}_1\frac{A}{\pi} + \dot{\underline{f}}_0\frac{(1-A)}{1-\pi}\right\} \dot{h}(\bZ)\right]^\top  \mathbf{M}^{-1} \bZ \\
    &= \left(\frac{\dot{\underline{f}}_0}{1-\pi}, \frac{\dot{\underline{f}}_1}{\pi}-\frac{\dot{\underline{f}}_0}{1-\pi}, \bzero\right) \mathbf{M}^\top \mathbf{M}^{-1} \bZ - E[\dot{\underline{f}}_1\dot{h}(\bZ(1)) + \dot{\underline{f}}_0\dot{h}(\bZ(0))]\mathbf{M}^{-1} \bZ \\
    &= \dot{\underline{f}}_1\frac{A}{\pi} + \dot{\underline{f}}_0\frac{(1-A)}{1-\pi} - E[\dot{\underline{f}}_1\dot{h}(\bZ(1)) + \dot{\underline{f}}_0\dot{h}(\bZ(0))]\mathbf{M}^{-1} \bZ,
\end{align*}
which completes the derivation of Equation~\eqref{IF_formula}.

Given the influence function~\eqref{IF_formula}, we get
\begin{align*}
    IF(\bO(1)) - IF(\bO(0)) &= \left(\frac{\dot{\underline{f}}_1}{\pi} - \boldsymbol{c}_1^\top \bZ(1) \right)\frac{R(1)\{Y(1)-h(\bZ(1))\}}{ e(\bZ(1))} - \boldsymbol{c}_2^\top \bZ(1)\{R(1)- e(\bZ(1))\} \\
    & -  \left(\frac{\dot{\underline{f}}_0}{1-\pi} - \boldsymbol{c}_1^\top \bZ(0) \right)\frac{R(0)\{Y(0)-h(\bZ(0))\}}{ e(\bZ(0))} + \boldsymbol{c}_2^\top \bZ(0)\{R(0)- e(\bZ(0))\}.
\end{align*}
To compute $Cov(IF(\bO(1)) - IF(\bO(0)), \bX^r)$, we observe the following facts. If the missingness model is correctly specified, we have $\boldsymbol{c}_1 = \bzero$ and $E[\boldsymbol{c}_2^\top \bZ(a)\{R- e(\bZ(a))\}|\bX^r] = 0$. If the outcome regression model is correctly specified, we have $\boldsymbol{c}_2 = \bzero$ and under MAR,
\begin{align*}
    E\left[\boldsymbol{c}_1^\top \bZ(a)\frac{R(a)\{Y(a)-h(\bZ(a))\}}{ e(\bZ(a))}\bigg| \bX^r\right] = E\left[\boldsymbol{c}_1^\top \bZ(a)\frac{E[R(a)|\bZ(a)]\{E[Y(a)|\bZ(a)]-h(\bZ(a))\}}{ e(\bZ(a))}\bigg| \bX^r\right] = 0.
\end{align*}
As a result, if either of the two working models is correctly specified, we have
\begin{align}
    & Cov\{IF(\bO(1)) - IF(\bO(0)), \bX^r\}\nonumber\\
    &= Cov\left\{\frac{\dot{\underline{f}}_1}{\pi} \frac{R(1)\{Y(1)-h(\bZ(1))\}}{ e(\bZ(1))} -  \frac{\dot{\underline{f}}_0}{1-\pi} \frac{R(0)\{Y(0)-h(\bZ(0))\}}{ e(\bZ(0))}, \bX^r\right\}.\label{proof1-cov}
\end{align}
Given the above expression, if the outcome model is correctly specified, i.e., $E[Y(a)|\bZ(a)] = h(\bZ(a))$, we directly get $Cov\{IF(\bO(1)) - IF(\bO(0)), \bX^r\} = 0$, which implies $R^2=0$ by definition. If, alternatively, $\pi = 0.5$ and $\dot{\underline{f}}_1 = - \dot{\underline{f}}_0 = 1$, we get
\begin{align*}
    & Cov\{IF(\bO(1)) - IF(\bO(0)), \bX^r\}\\
    &= 4\ Cov\left\{\pi \frac{R(1)\{Y(1)-h(\bZ(1))\}}{ e(\bZ(1))} + (1-\pi) \frac{R(0)\{Y(0)-h(\bZ(0))\}}{ e(\bZ(0))}, \bX^r\right\}.
\end{align*}
Since $\pi E[\bpsi(\bO(1);\underline{\btheta}] + (1-\pi) E[\bpsi(\bO(0);\underline{\btheta}] =\bzero$, the definition of $\bpsi$ implies that
\begin{align*}
    E\left[\pi \frac{R(1)\{Y(1)-h(\bZ(1))\}}{ e(\bZ(1))} + (1-\pi) \frac{R(0)\{Y(0)-h(\bZ(0))\}}{ e(\bZ(0))}\right] &= \bzero \\
    E\left[\left\{\pi \frac{R(1)\{Y(1)-h(\bZ(1))\}}{ e(\bZ(1))} + (1-\pi) \frac{R(0)\{Y(0)-h(\bZ(0))\}}{ e(\bZ(0))}\right\}\bX^r\right] &= \bzero,
\end{align*}
which also yields $Cov\{IF(\bO(1)) - IF(\bO(0)), \bX^r\} = \bzero$ and hence $R^2=0$.

Finally, if we include treatment-covariate interaction terms, then update the definition to $\bZ = (1, A, \bX^\top, A\bX^\top)^\top$, $\bZ^* = (1, A^*, \bX^\top, A\bX^\top)^\top$, $\bZ(1) = (1, 1, \bX^\top, A\bX^\top)^\top$, $\bZ(0) = (1, 0, \bX^\top, A\bX^\top)^\top$ and update $\balpha, \bbeta$ to further include $\balpha_{AX}, \bbeta_{AX}$. Then the estimating function and influence function remain unchanged. Following a similar proof, we get the same expression of $Cov\{IF(\bO(1)) - IF(\bO(0)), \bX^r\}$ as in Equation~\eqref{proof1-cov}. By the definition of $\underline{\btheta}$ and $\bpsi$, we have
\begin{align*}
    E\left[\frac{R(a)\{Y(a)-h(\bZ(a))\}}{ e(\bZ(a))}\right] &= \bzero \\
    E\left[\frac{R(a)\{Y(a)-h(\bZ(a))\}}{ e(\bZ(a))}\bX^r\right] &= \bzero,
\end{align*}
which implies $ Cov\{IF(\bO(1)) - IF(\bO(0)), \bX^r\} = \bzero$ and hence $R^2=0$.
\end{proof}

\begin{proposition}\label{prop: mixed-ancova}
    Theorem 2 in the main paper holds for the mixed-ANCOVA estimator. 
\end{proposition}
\begin{proof}
    We inherit all notations from the proof of Theorem 2. 
    Since we are handling cluster data, we define $\bY = (\bY_{.1}, \dots, \bY_{.N})^\top$, $\bZ = (\bone_N, A\bone_N, \mathbf{X})$, $\bZ(1) = (\bone_N, \bone_N, \mathbf{X})$, $\bZ = (\bone_N, \bzero_N, \mathbf{X})$, where $\mathbf{X} = (\bX_{.1}, \dots, \bX_{.N})^\top$.
    By \cite{wang2023CRT}, the estimating equation for the mixed-model ANCOVA estimator is 
    \begin{equation}\label{eq: lmm-psi-C}
    \bpsi(\bO,\btheta) = \left(\begin{array}{c}
    f(\mu_1,\mu_0) -\Delta\\
    \mu_1 - \frac{1}{N}\bone_{N}^\top\bZ(1)\bbeta \\
    \mu_0 - \frac{1}{N}\bone_{N}^\top\bZ(0)\bbeta \\
    \bZ^{\top}\bfV(\bY - \bZ\bbeta)  \\
    -\textrm{tr}(\bfV) + (\bY - \bZ\bbeta)^\top \bfV^2(\bY - \bZ\bbeta)\\
  -\bone_{N}^\top\bfV\bone_{N} + (\bY - \bZ\bbeta)^\top \bfV\bone_{N}\bone_{N}^\top \bfV(\bY - \bZ\bbeta)    
   \end{array}\right),
\end{equation}
where $\bfV = (\sigma^2 \bfI_N + \tau^2 \bone_N \bone_N^\top)^{-1} = \frac{1}{\sigma^2} \bfI_N -  \frac{\tau^2}{\sigma^2(\sigma^2+N\tau^2)}\bone_{N}\bone_{N}^\top$, and the influence function for $\widehat{\Delta}$ is
\begin{align*}
    IF(\bO) &= \left\{\dot{\underline{f}}_1\frac{A}{\pi} + \dot{\underline{f}}_0\frac{(1-A)}{1-\pi}\right\} E\left[\frac{N}{ \underline\sigma_C^2 + N\underline\tau_C^2}\right]^{-1} \frac{1}{\underline\sigma_C^2 + N\underline\tau_C^2} \bone_N^\top (\bY - \bZ\underline{\bbeta})  \\
    &\quad + \dot{\underline{f}}_1\left\{\frac{1}{N}\bone_{N}^\top\bZ(1)\underline{\bbeta} - \underline{\mu}_1\right\} + \dot{\underline{f}}_0\left\{\frac{1}{N}\bone_{N}^\top\bZ(0)\underline{\bbeta} - \underline{\mu}_0\right\}.
\end{align*}
Then, we get
\begin{align*}
    IF(\bO(1)) - IF(\bO(0)) & = \dot{\underline{f}}_1\frac{1}{\pi}E\left[\frac{N}{ \underline\sigma_C^2 + N\underline\tau_C^2}\right]^{-1} \frac{1}{\underline\sigma_C^2 + N\underline\tau_C^2} \bone_N^\top \{\bY(1) - \bZ(1)\underline{\bbeta}\} \\
    &\quad - \dot{\underline{f}}_0\frac{1}{1-\pi}E\left[\frac{N}{ \underline\sigma_C^2 + N\underline\tau_C^2}\right]^{-1} \frac{1}{\underline\sigma_C^2 + N\underline\tau_C^2} \bone_N^\top \{\bY(0) - \bZ(0)\underline{\bbeta}\}.
\end{align*}
Given the definition of $\bpsi$, since $\bX^r$ is a cluster-level covariate included in $\mathbf{X}$ as $\bX^r \bone_{N}^\top$,  we have
\begin{align*}
    \bzero &= \pi E[\bX^r \bone_{N}^\top \bfV (\bY(1) - \bZ(1)\underline{\bbeta})] + (1-\pi) E[\bX^r \bone_{N}^\top \bfV (\bY(0) - \bZ(0)\underline{\bbeta})] \\
    &= \pi E\left[\frac{1}{\underline\sigma_C^2 + N\underline\tau_C^2} \bX^r \bone_N^\top (\bY(1) - \bZ(1)\underline{\bbeta}) \right] + (1-\pi) E\left[\frac{1}{\underline\sigma_C^2 + N\underline\tau_C^2} \bX^r \bone_N^\top (\bY(0) - \bZ(0)\underline{\bbeta}) \right].
\end{align*}
As a result, if $\dot{\underline{f}}_1 = - \dot{\underline{f}}_0 = 1$ and $\pi = 1-\pi$, we get 
\begin{align*}
   & E[\{IF(\bO(1)) - IF(\bO(0))\}\bX^r] \\
   &= 4E\left[\frac{N}{ \underline\sigma_C^2 + N\underline\tau_C^2}\right]^{-1} \sum_{a=0}^1 \pi^a(1-\pi)^a E\left[\frac{1}{\underline\sigma_C^2 + N\underline\tau_C^2} \bX^r \bone_N^\top (\bY(a) - \bZ(a)\underline{\bbeta}) \right] \\
   &= \bzero.
\end{align*}
Following a similar proof, we also get $E[IF(\bO(1)) - IF(\bO(0))] = 0$. This implies, under condition (1), we have $Cov(IF(\bO(1)) - IF(\bO(0)), \bX^r) = \bzero$. 

If the mixed-ANCOVA model is correctly specified, we directly have $E[\bY(a)|\bZ(a)] = \bZ(a)\bbeta$, which implies $Cov(IF(\bO(1)) - IF(\bO(0)), \bX^r) = \bzero$. 

If the mixed-ANCOVA model further includes treatment-covariate interaction terms, the formula of $\{IF(\bO(1)) - IF(\bO(0))\}$ remains the same except that $\bZ(1) = (\bone_N, \bone_N, \mathbf{X}, \mathbf{X})$ and $\bZ(0) = (\bone_N, \bzero_N, \mathbf{X}, \bzero)$. Then the definition of $\bpsi$ implies that $E\left[\frac{1}{\underline\sigma_C^2 + N\underline\tau_C^2} \bX^r \bone_N^\top (\bY(a) - \bZ(a)\underline{\bbeta}) \right] = \bzero$, resulting in $Cov(IF(\bO(1)) - IF(\bO(0)), \bX^r) = \bzero$.
\end{proof}

\subsubsection{Rerandomization based on other balance criteria (distance functions)}
Consider rerandomization using $\bI^\top \widehat{\bfH}^{-1} \bI < t$. We assume $n\widehat{\bfH} \xrightarrow{p} \underline{\bfH}$. 
Since we have shown that $\bI \xrightarrow{d} N(0,\bfV_I)$, then the Delta method implies that $\bI^\top \bfV_{I}^{1/2}\underline{\bfH}^{-1} \bfV_{I}^{1/2}\bI \xrightarrow{d} \bD^\top \bfV_{I}^{1/2} \underline{\bfH}^{-1} \bfV_{I}^{1/2}\bD$, which is a generalized chi-squared distribution. By the assumption $n\widehat{\bfH} \xrightarrow{p} \underline{\bfH}$, the Slutsky's theorem implies that $\bI^\top \widehat{\bfH}^{-1} \bI \xrightarrow{d} \bD^\top \bfV_{I}^{1/2} \underline{\bfH}^{-1} \bfV_{I}^{1/2}\bD$, the same generalized chi-squared distribution. Therefore, $P(\bD^\top \bfV_{I}^{1/2} \underline{\bfH}^{-1} \bfV_{I}^{1/2}\bD < t) > 0$ if $t > 0$. Therefore, Lemma~\ref{lemma1} remains the same except that $\widehat{Var}(\bI)$ is replaced by $\widehat{\bfH}$ in the proof. Similarly, Lemma~\ref{lemma2} remains the same except that $\bR_{qt}$ is substituted by $\widetilde{\bR}_{\bfH,t} = \bD|(\bD^\top\bfV_{I}^{1/2} \underline{\bfH}^{-1} \bfV_{I}^{1/2}\bD <t)$. Following the proof of Theorem 1, we get the desired weak convergence result. 

To show rerandomization does not increase variance, we first give a proposition below.

\begin{proposition}
    Let $X, Y$ be two independent, positive random variables. Then (1) $E[X|X<t]$ is nondecreasing in $t$ and (2) $E[X|X+Y<t] \le E[X|X<t]$ for any $t >0$.
\end{proposition}
\begin{proof}
For (1), we have
\begin{align*}
    \frac{d}{dt}E[X|X<t] &= \frac{d}{dt} \frac{\int_{0}^t x f(x)dx}{\int_{0}^t  f(x)dx} = \frac{tf(t)\int_{0}^t  f(x)dx - \int_{0}^t x f(x)dx f(t)}{\left(\int_{0}^t  f(x)dx\right)^2} = \frac{f(t)\int_{0}^t  (t-x)f(x)dx}{\left(\int_{0}^t  f(x)dx\right)^2} \ge 0.
\end{align*}
For (2), (1) implies that $E[X|X<t-y] \le E[X|X<t]$ for any $y \ge0$, which implies
\begin{align*}
    \int_{0}^{t-y} x f(x) dx \le E[X|X<t] \int_{0}^{t-y} f(x) dx.
\end{align*}
Therefore, we have
\begin{align*}
    E[X|X+Y<t] &= \frac{\int_{0}^t \int_{0}^{t-y} x f(x) dx f(y) dy}{\int_{0}^t \int_{0}^{t-y}  f(x) dx f(y) dy} \\
    &\le \frac{\int_{0}^t \int_{0}^{t-y} E[X|X<t] \int_{0}^{t-y} f(x) dx dx f(y) dy}{\int_{0}^t \int_{0}^{t-y}  f(x) dx f(y) dy} \\
    &= E[X|X<t].
\end{align*}
\end{proof}
Returning to $\widetilde{\bR}_{\bfH,t}$, we aim to show $Var(\bD)-Var(\widetilde{\bR}_{\bfH,t})$ is positive semi-definite. Letting $\mathbf{P}\mathbf{\Lambda}\mathbf{P}^\top$ be the singular value decomposition of $\bfV_{I}^{1/2} \underline{\bfH}^{-1} \bfV_{I}^{1/2}$, we denote $\widetilde{\bD} = \mathbf{P}^\top \bD$, we have $\widetilde{\bD} \sim N(0,\bfI_q)$ and  $Var(\widetilde{\bR}_{\bfH,t}) = \mathbf{P}^\top Var(\widetilde{\bD}| \widetilde{\bD}^\top \mathbf{\Lambda}\widetilde{\bD} < t)\mathbf{P}$, which implies that it suffices to prove $Var(\widetilde{\bD}) - Var(\widetilde{\bD}| \widetilde{\bD}^\top \mathbf{\Lambda}\widetilde{\bD} < t)$ is positive definite. 
Since $\widetilde{\bD} \sim N(0,\bfI_q)$, we have $E[\widetilde{D}_i|\widetilde{D}_j= d_j, \widetilde{\bD}^\top \mathbf{\Lambda}\widetilde{\bD} < t] = E[\widetilde{D}_i|\sum_{j'\ne j} \lambda_{j'} D_{j'}^2 < t - \lambda_j d_j^2] = E[-\widetilde{D}_i|\sum_{j'\ne j} \lambda_{j'} D_{j'}^2 < t - \lambda_j d_j^2] = 0$, which implies $E[\widetilde{D}_i\widetilde{D}_j|\widetilde{\bD}^\top \mathbf{\Lambda}\widetilde{\bD} < t] = 0$ for $i\ne j$.  Therefore, it is sufficient to show $E[\widetilde{D}_i^2|\widetilde{\bD}^\top \mathbf{\Lambda}\widetilde{\bD} < t] \le E[\widetilde{D}_i^2]$. Now, using the above proposition with $X = \widetilde{D}_i^2$ and $Y = \sum_{j\ne i} \lambda_j \widetilde{D}_j^2$, we have $X$ is independent of $Y$, and then $E[\widetilde{D}_i^2|\widetilde{\bD}^\top \mathbf{\Lambda}\widetilde{\bD} < t] \le E[\widetilde{D}_i^2|\lambda_i\widetilde{D}_i^2<t] \le E[\widetilde{D}_i^2]$, which completes the proof.


\subsubsection{A counterexample showing rerandomization may increase variance in finite-samples}

We consider a randomized experiment with $N_1 = 1$ and $N_0=2$. Let $X^r$ follow a dichotomous distribution with $P(X^r=2) = p,\ p < 0.5$ and $P(X^r=-1)=1-p$, $Y(a) \sim N(0, 1 + M a I\{X^r = -1\})$ with $M$ being a positive constant, and $t=1$ for rerandomization. For the unadjusted estimator, $\widehat{\Delta} = \sum_{i=1}^n {A_iY_i/N_1 - (1-A_i)Y_i/N_0}$, we denote its variance under simple randomization as $V^*$ and under rerandomization as $V$. Without loss of generality, our derivation is conditioned on $N_1$ being a fixed number. 

We next enumerate all possible values of $X_{1:3} = (X_1^r,X_2^r, X_3^r)$. When $X_{1:3} = (2,2,2)$ or $(-1,-1,-1)$, the covariance imbalance $I$ is always zero, yielding the same variance under simple or rerandomization. When $X_{1:3} = (-1,-1,2)$, we have $I = 3$ (if the one with $X^r =2$ is assigned to treatment) or $I=-1.5$ (if the one with $X^r =2$ is assigned to control). Since $\widehat{V}(I) = 3$, only the second scenario will be picked by rerandomization with $t=1$. Therefore, conditioning on $X_{1:3} = (-1,-1,2)$, the variance of $\widehat{\Delta}$ is $ M+1.5$ under rerandomization. Since the other scenario gives variance 1.5, the variance of $\widehat{\Delta}$ under simple randomization is $1.5/3 + 2(M+1.5)/3 = 2M/3 + 1.5$. Therefore, these covariate value specifications make rerandomization increase the variance. Finally, when $X_{1:3} = (2,2,-1)$, a similar procedure gives the variance of $\widehat{\Delta}$ as 1.5 under rerandomization and $M/3 + 1.5$ under simple randomization. Marginalizing over the distribution of $X_{1:3}$ and noting $P(X_{1:3} = (-1,-1,2)) = 3p(1-p)^2$ and $P(X_{1:3} = (2,2,-1)) = 3p^2(1-p)$, we have
\begin{align*}
    V^*-V &= 3p(1-p)^2 \left\{(M + 1.5) - \left(\frac{2M}{3} + 1.5\right)\right\} + 3p^2(1-p) \left\{1.5 - \left(\frac{M}{3} + 1.5\right)\right\} \\
    &= Mp(1-p)^2 - Mp^2(1-p) \\
    &= Mp(1-p)(1-2p),
\end{align*}
 which is positive when $M>0$ and $p< 0.5$.  Notably, the same results also hold if ANCOVA is used to adjust for $X^r$, following the same proof.

\subsection{Results under stratified rerandomization}
\subsubsection{Lemmas}

\begin{lemma}\label{lemma:stratified}
Let $\tA_1,\dots, \tA_n$ be treatment allocation under stratified randomization, $\bh(\bW_i)$ denote a multidimensional function of $\bW_i$ with finte second moments, and $\mathcal{F}_n$ denote the sigma field generated by $\bW_1,\dots, \bW_n$. We have $n^{-1/2} \sum_{i=1}^n (\tA_i - \pi) \bh(\bW_i) | \mathcal{F}_n \overset{\bullet}{\sim} N(0, \widetilde{\bSigma}_n) | \mathcal{F}_n$, where $\widetilde{\bSigma}_n = \frac{\pi(1-\pi)}{n}\sum_{s \in \mathcal{S}} \sum_{i: S_i = s} (\bh(\bW_i) - \overline{\bh}_s)(\bh(\bW_i) - \overline{\bh}_s)^\top$ with \\
$\overline{\bh}_s = (\sum_{i=1}^n I\{S_i=s\})^{-1}(\sum_{i: S_i = s} \bh(\bW_i))$.
\end{lemma}

\begin{proof}
By conditioning on $\mathcal{F}_n$, $\bh(\bW_i)$ is deterministic, and the only randomness comes from treatment assignment. Under stratified randomization, treatment assignment is independent across strata and follows a complete randomization within each stratum. Therefore, the standard finite-sample central limit theorem \citep{li2017general} can be applied to each stratum, and we combine all strata to obtain the joint convergence. See \cite{pirondini2022covariate} for a detailed proof. 
\end{proof}

\begin{lemma}\label{lemma3}
Let $\mathcal{F}_n$ denote the sigma field generated by $\bW_1,\dots, \bW_n$.
    Under stratified randomization, $\sqrt{n}\widetilde{\bI}|\mathcal{F}_n \overset{\bullet}{\sim} N(0, \widetilde{\bfV}_{I,n})|\mathcal{F}_n$ and $n\widehat{Var}(\widetilde{\bI}) - \widetilde{\bfV}_{I,n} \xrightarrow{p}0$ conditioning on $\mathcal{F}_n$,  where $\widetilde{\bfV}_{I,n} = \frac{1}{n\pi(1-\pi)}\sum_{s \in \mathcal{S}} \sum_{i: S_i = s} (\bX_i^r - \overline{\bX}_s^r)(\bX_i^r - \overline{\bX}_s^r)^\top$. Furthermore, define
    \begin{align*}
        \widetilde{\bI}^{\dagger} = \sum_{s \in \mathcal{S}} \widehat{p}_s\sum_{i=1}^n \left[\frac{I\{S_i=s\}\tA_i\bX_i^r}{\sum_{i=1}^n I\{S_i=s\}\tA_i} - \frac{I\{S_i=s\}(1-\tA_i)\bX_i^r}{\sum_{i=1}^n I\{S_i=s\}(1-\tA_i)}\right]
    \end{align*}
    where $\widehat{p}_s=\frac{1}{n}\sum_{i=1}^n I\{S_i=s\}$. Here, $\widetilde{\bI}^{\dagger}$ is used for computing the overall Mahalanobis distance in \cite{wang2022rerandomization}. Then $\sqrt{n}\widetilde{\bI}^{\dagger}|\mathcal{F}_n \overset{\bullet}{\sim} N(0, \widetilde{\bfV}_{I,n})|\mathcal{F}_n$, implying that $\widetilde{\bI}^{\dagger}$ is asymptotically equivalent to $\widetilde{\bI}$.
\end{lemma}

\begin{proof}
 Setting $\bh(\bW_i) = (1,\bX_i^r)$, Lemma~\ref{lemma:stratified} implies the conditional Gaussian approximation of  $(A_i-\pi,(A_i-\pi)\bX_i^r)$. Since $\widetilde{\bI}$ is a deterministic function of $(A_i-\pi,(A_i-\pi)\bX_i^r)$ conditioning on $\mathcal{F}_n$, we apply the Delta method and obtain the desired conditional convergence of $\sqrt{n}\widetilde{\bI}$ and $n\widehat{Var}(\widetilde{\bI})$. For $\widetilde{\bI}^{\dagger}$, we set $\bh(\bW_i) = (I\{S_i=1\}, I\{S_i=1\}\bX_i^r,\dots, I\{S_i=|\mathcal{S}|\}, I\{S_i=|\mathcal{S}|\}\bX_i^r)$, and the same procedure yeilds the desired convergence.
\end{proof}

\begin{lemma}\label{lemma4}
    Let $h(\bW;\btheta)$ be some measurable real-valued function $h$ on $\bW$ defined in Assumption 1 and parameters $\btheta \in \bTheta \subset \mathbb{R}^p$. 
    We assume that $\bTheta$ is a compact set, $h$ is a continuous function on $\btheta$, and  $h(\bW;\btheta)$ is dominated by an integrable function. 
    Under stratified rerandomization and Assumption 1, we have $\sup_{\btheta \in \bTheta}||\frac{1}{n} \sum_{i=1}^n A_i h(\bW_i;\btheta) - \pi E[h(\bW_i;\btheta)] ||_2 \xrightarrow{p} 0$ and $\sup_{\btheta \in \bTheta}||\frac{1}{n} \sum_{i=1}^n (1-A_i) h(\bW_i;\btheta) - (1-\pi) E[h(\bW_i;\btheta)] ||_2 \xrightarrow{p} 0$.
\end{lemma}

\begin{proof}
This proof is similar to the proof for Lemma~\ref{lemma1}. The major difference is in that we use Lemma~1 from the Supplementary Material of \cite{wang2023model} to handle stratified randomization and use Lemma~\ref{lemma3} to handle the distribution $\widetilde{\bI}^\top \{\widehat{Var}(\widetilde{\bI})\}^{-1} \widetilde{\bI}$.

Define $\mathcal{F}_n = \sigma(\bW_1,\dots,\bW_n)$ as the sigma-field generated under Assumption 1. For brevity, we further denote $U_n(A,\bW) = \sup_{\btheta \in \bTheta}||\frac{1}{n} \sum_{i=1}^n A_i h(\bW_i;\btheta) - \pi E[h(\bW_i;\btheta)]||_2$. Our goal is to show $P(U_n(A,\bW) > \varepsilon) \rightarrow 0$ for any $\varepsilon >0$, and it suffices to prove $P(U_n(A,\bW) > \varepsilon | \mathcal{F}_n) \xrightarrow{p} 0$ due to the dominated convergence theorem for convergence in probability \citep[Theorem 1.5.3]{durrett2019probability}. 

Using the definition of rerandomization, we have 
\begin{align*}
    P\left(U_n(A,\bW) > \varepsilon | \mathcal{F}_n\right) &= P\left(U_n(\tA,\bW) > \varepsilon | \mathcal{F}_n, \widetilde{\bI}^\top \{\widehat{Var}(\widetilde{\bI})\}^{-1} \widetilde{\bI} < t\right) \\
    &= \frac{ P\left(U_n(\tA,\bW) > \varepsilon, \widetilde{\bI}^\top \{\widehat{Var}(\widetilde{\bI})\}^{-1} \widetilde{\bI} < t | \mathcal{F}_n\right)}{P\left(\widetilde{\bI}^\top \{\widehat{Var}(\widetilde{\bI})\}^{-1} \widetilde{\bI} < t | \mathcal{F}_n\right)} \\
    &\le \frac{ P\left(U_n(\tA,\bW) > \varepsilon | \mathcal{F}_n\right)}{P\left(\widetilde{\bI}^\top \{\widehat{Var}(\widetilde{\bI})\}^{-1} \widetilde{\bI} < t | \mathcal{F}_n\right)},
\end{align*}
where $\tA$ represents treatment allocation under stratified randomization and $\widetilde{\bI} = \frac{1}{N_1} \sum_{i=1}^n \tA_i \bX^r_i - \frac{1}{N_0} \sum_{i=1}^n (1-\tA_i) \bX^r_i$.
\underline{For the numerator term}, we define $\{h(\bW^{(s)};\btheta): \btheta \in \bTheta\}$ as the conditional process of $\{h(\bW;\btheta): \btheta \in \bTheta\}$ given $S=s$ for any $s$ in the support of $S$ (see Definition 1 in the Supplementary Material of \cite{wang2023model} for details). Since $\bW^{(s)}$ lies in the support of $\bW$, the regularity condition (4) implies that the map $\btheta \mapsto h(\bW^{(s)};\btheta)$ is continuous in its support and dominated by an integrable function. Then, Example 19.8 of \cite{vaart_1998} implies that $\{h(\bW^{(s)};\btheta): \btheta \in \bTheta\}$ is P-Glivenko-Cantelli. Furthermore, $\sup_{\btheta \in \bTheta} ||E[h(\bW^{(s)};\btheta)]||_2 < \infty$ since $h(\bW^{(s)};\btheta)$ is dominated by an integrable function. As a result, Lemma 1 (2) in the Supplementary Material of \cite{wang2023model} implies that $ P\left(U_n(\tA,\bW)\right) \rightarrow 0$. Applying Markov's inequality, we have $P\left(U_n(\tA,\bW) > \varepsilon| \mathcal{F}_n\right) \xrightarrow{p} 0$.
\underline{For the denominator term}, Lemma~\ref{lemma3} and Slutsky's theorem implies $\widetilde{\bI}^\top \{\widehat{Var}(\widetilde{\bI})\}^{-1} \widetilde{\bI}|\mathcal{F}_n\overset{\bullet}{\sim}\mathcal{X}^2_q|\mathcal{F}_n$, where $q$ is the dimension of $\bX^r$. Therefore, $P(\widetilde{\bI}^\top \{\widehat{Var}(\widetilde{\bI})\}^{-1} \widetilde{\bI} < t|\mathcal{F}_n) \rightarrow P(\mathcal{X}^2_q < t|\mathcal{F}_n) = P(\mathcal{X}^2_q < t)$, which is greater than zero since $t > 0$. Combining the convergence of both numerator and denominator terms, we obtain $P\left(U_n(A,\bW) > \varepsilon | \mathcal{F}_n\right) \xrightarrow{p} \frac{0}{P(\mathcal{X}^2_q < t)} = 0$. This completes the proof for $\sup_{\btheta \in \bTheta}||\frac{1}{n} \sum_{i=1}^n A_i h(\bW_i;\btheta) - \pi E[h(\bW_i;\btheta)] ||_2 \xrightarrow{p} 0$ under stratified rerandomization. Following a similar proof, we can get $\sup_{\btheta \in \bTheta}||\frac{1}{n} \sum_{i=1}^n (1- A_i) h(\bW_i;\btheta) - (1-\pi) E[h(\bW_i;\btheta)] ||_2 \xrightarrow{p} 0$.
\end{proof}

\begin{lemma}\label{lemma5}
    Let $\bZ_i(a) = h_a(Y_i(a), M_i(a), \bX_i)$  for some $k$-dimensional function $h_a, a \in \{0,1\}$ taking values in the Euclidean space such that $E[||\bZ_i(a)||_2^2] < \infty$. Under stratified rerandomization and Assumption 1, we have
    \begin{align*}
        \frac{1}{\sqrt{n}}\sum_{i=1}^n \left\{A_i\bZ_i(1) + (1-A_i)\bZ_i(0) - \pi E[\bZ_i(1)] - (1-\pi) E[\bZ_i(0)]\right\} \xrightarrow{d}  \widetilde{\boldsymbol{\varepsilon}}_k + \widetilde{\bfV}_{UI}^\top \widetilde{\bfV}_I^{-0.5} \bR_{q,t},
    \end{align*}
    where $\tbepsilon_k \sim N(\bzero,\widetilde{\bfV}_U -  \widetilde{\bfV}_{UI}^\top \widetilde{\bfV}_I^{-1}\widetilde{\bfV}_{UI})$, $\bR_{q,t} \sim  \bD_q|( \bD^\top  \bD < t)$,  where $\bD \sim N(\bzero, \mathbf{I}_q)$ with
    \begin{align*}
        \widetilde{\bfV}_U &= Var\{A_i^*\bZ_i(1) + (1-A_i^*)\bZ_i(0)\} - \pi(1-\pi) Var(E[\bZ(1)-\bZ(0)|S]) \\
        \widetilde{\bfV}_I &= \frac{1}{\pi(1-\pi)}E[Var(\bX_i^r|S)] \\
        \widetilde{\bfV}_{UI} &= E[Cov\{\bX_i^r, \bZ_i(1)-\bZ_i(0)|S\}].
    \end{align*}
\end{lemma}

\begin{proof}
This proof is similar to the proof for Lemma~\ref{lemma2}. The major difference is that we use Lemma~\ref{lemma:stratified} to obtain the weak convergence under stratified randomization. We use the same notation as in the proof for Lemma~\ref{lemma2}, which are
\begin{align*}
    \bI_n(A) &= \sqrt{n} \left(\frac{\sum_{i=1} A_i \bX_i^r}{\sum_{i=1} A_i} - \frac{\sum_{i=1} (1-A_i) \bX_i^r}{\sum_{i=1} (1-A_i)}\right), \\
    \mathcal A_n &=\Big\{\, \bI_n(\tA)^\top\,\widehat{\mathrm{Var}}(\bI_n(\tA))^{-1}\,\bI_n(\tA)< t\,\Big\},\\
    \bU_n(A) &=\frac{1}{\sqrt n}\sum_{i=1}^n\!\Big\{A_i \bZ_i(1)+(1-A_i)\bZ_i(0)-\pi E\bZ(1)-(1-\pi)E\bZ(0)\Big\},\\
    \bT_n(A) &= \frac{1}{\sqrt n}\sum_{i=1}^n\!\Big\{(A_i-\pi) (\bZ_i(1)-\bZ_i(0)))\Big\},\\
    \bR_n &= \frac{1}{\sqrt n}\sum_{i=1}^n\!\Big\{\pi \bZ_i(1)+(1-\pi)\bZ_i(0)-\pi E\bZ(1)-(1-\pi)E\bZ(0)\Big\}.
\end{align*}
Then we have $\bU_n(A) = \bT_n(A) + \bR_n$. We further define $\mathcal{F}_n = \sigma(\bW_1,\dots,\bW_n)$ as the sigma-field generated under Assumption 1, and $X \overset{d}{=} Y$ if $X$ and $Y$ have the same distribution. By the definition of rerandomization, we have $g(A,\bW)| \mathcal{F}_n \overset{d}{=} g(\tA,\bW)| (\mathcal{F}_n, \mathcal{A}_n)$ for any measurable function $g$.

We first provide the conditional Normal approximation under stratified randomization. To utilize Lemma~\ref{lemma:stratified}, we define $\bh(\bW_i)= (1, \bX_i^r, \bZ_i(1) - \bZ_i(0))$, which has finite second moment. Then, Lemma~\ref{lemma:stratified} and the Delta method implies $(\bI_n(\tA), \bT_n(\tA)) | \mathcal{F}_n \overset{\bullet}{\sim} N(0, \widetilde{\bSigma}_n) | \mathcal{F}_n$, where
\[
\widetilde{\bSigma}_n=
\begin{bmatrix}
\displaystyle \tbfV_{I,n} & \displaystyle \tbfV_{IT,n}\\[4pt]
\displaystyle \tbfV_{IT,n}^\top & \displaystyle \tbfV_{T,n}
\end{bmatrix},
\]
with blocks (all $\mathcal F_n$-measurable sample moments)
\begin{align*}
    \tbfV_{I,n} &=\frac{1}{n\pi(1-\pi)}\sum_{s \in \mathcal{S}} \sum_{i: S_i = s} (\bX_i^r - \overline{\bX}_s^r)(\bX_i^r - \overline{\bX}_s^r)^\top,\\
\tbfV_{IT,n} &=\frac{1}{n}\sum_{s \in \mathcal{S}}\sum_{i: S_i = s} (\bX_i^r - \overline{\bX}_s^r)\,(\bZ_i(1)-\bZ_i(0)-\overline{\bZ(1)}_s+\overline{\bZ(0)}_s)^\top,  \\
\tbfV_{T,n}&=\frac{\pi(1-\pi)}{n}\sum_{s \in \mathcal{S}}\sum_{i: S_i = s} (\bZ_i(1)-\bZ_i(0)-\overline{\bZ(1)}_s+\overline{\bZ(0)}_s)(\bZ_i(1)-\bZ_i(0)-\overline{\bZ(1)}_s+\overline{\bZ(0)}_s)^\top.
\end{align*}
Following a similar procedure, we get
\begin{equation}\label{eq: ivi}
\{\bT_n(\tA), \bI_n(\tA)^\top n\widehat{Var}(\bI_n(\tA))^{-1}\bI_n(\tA)\} | \mathcal{F}_n\overset{\bullet}{\sim} \left\{\bT_n(\tA), \bI_n(\tA)^\top \bV_{I,n}^{-1}\bI_n(\tA) \right\} | \mathcal{F}_n.
\end{equation}

We next derive the asymptotic distribution of $\bT_n(\tA) | \mathcal{F}_n, \mathcal{A}_n$. Denoting $\tbepsilon_n = \bT_n(\tA) -  \tbfV_{IT,n}^\top   \tbfV_{I,n}^{-1}\bI_n(\tA)$, our result of $(\bI_n(\tA), \bT_n(\tA)) | \mathcal{F}_n \overset{\bullet}{\sim} N(0, \bSigma_n) | \mathcal{F}_n$ implies that $\bvarepsilon_n | \mathcal{F}_n \overset{\bullet}{\sim} N(0, \tbfV_{T,n} - \tbfV_{IT,n}^\top \tbfV_{I,n}^{-1}\tbfV_{IT,n}) | \mathcal{F}_n$ and that $\varepsilon_n$ is asymptotically independent of $\bI_n(\tA)$ given $\mathcal{F}_n$. Combined with Equation~\eqref{eq: ivi}, we have
\begin{align*}
    \bT_n(\tA)| \mathcal{F}_n, \mathcal{A}_n &\overset{\bullet}{\sim} \bT_n(\tA)| \mathcal{F}_n, \bI_n(\tA)^\top \tbfV_{I,n}^{-1}\bI_n(\tA)< t \\
    &\overset{\bullet}{\sim} \tbepsilon_n| \mathcal{F}_n +  \tbfV_{IT,n}^\top   \tbfV_{I,n}^{-1}\bI_n(\tA)| \mathcal{F}_n, \bI_n(\tA)^\top \tbfV_{I,n}^{-1}\bI_n(\tA)< t \\
    &\overset{\bullet}{\sim}\tbepsilon_n| \mathcal{F}_n +  \tbfV_{IT,n}^\top   \tbfV_{I,n}^{-0.5} \bR_{q,t} | \mathcal{F}_n,
\end{align*}
where the last step comes from $\tbfV_{I,n}^{-0.5} \bI_n(\tA)| \mathcal{F}_n \overset{\bullet}{\sim} \mathcal{N}(0, \bI_q)| \mathcal{F}_n$ and the definition of $\bR_{q,t}$ in the Lemma statement.

We then obtain
\begin{align*}
    \bU_n(A) |\mathcal{F}_n  &\overset{\bullet}{\sim} \bT_n(A) | \mathcal{F}_n\quad  + \quad  \bR_n  | \mathcal{F}_n \\
    &\overset{\bullet}{\sim} \bT_n(\tA)| (\mathcal{F}_n, \mathcal{A}_n) \quad + \quad \bR_n  | \mathcal{F}_n \\
    &\overset{\bullet}{\sim} (\tbepsilon_n +  \tbfV_{IT,n}^\top   \tbfV_{I,n}^{-0.5} \bR_{q,t} + \bR_n)  | \mathcal{F}_n,
\end{align*}
where the first step is by definition, the second step is implied by rerandomization, and the last step uses the above results. 

Finally, we derive the asymptotic distribution of $\bU_n(\tA)$.
Conditionally on $\mathcal F_n$, $\tbepsilon_n$ is independent of $\bI_n(\tA)$ and hence of
$\bR_{q,t}$. Since $\bR_n$ is $\mathcal F_n$-measurable, the triple
$(\tbepsilon_n,\ \bR_{q,t},\ \bR_n)$ is conditionally independent in the sense that, given $\mathcal F_n$,
the joint conditional characteristic function factorizes:
\begin{align*}
& E\!\left[\exp\!\Big\{i\big(\bt_1^\top \tbepsilon_n
+ \bt_2^\top \tbfV_{IT,n}^\top   \tbfV_{I,n}^{-0.5} \bR_{q,t}  + \bt_3^\top \bR_n)\big)\Big\}\ \Big|\ \mathcal F_n\right]\\
&=
\underbrace{E\!\left[\exp\{\,i \bt_1^\top \tbepsilon_n\}\ \Big|\ \mathcal F_n\right]}_{=: \psi_{n}(\bt_1)}
\underbrace{E\!\left[\exp\{\,i \bt_2^\top  \tbfV_{IT,n}^\top   \tbfV_{I,n}^{-0.5} \bR_{q,t}\}\ \Big|\ \mathcal F_n\right]}_{=: \eta_{n}(\bt_2)}
 \exp\{\,i \bt_3^\top \bR_n\}.    
\end{align*}

For $\psi_{n}(\bt_1)$, since we have obtained $\tbepsilon_n | \mathcal{F}_n \overset{\bullet}{\sim} N(0, \tbfV_{T,n} - \tbfV_{IT,n}^\top \tbfV_{I,n}^{-1}\tbfV_{IT,n}) | \mathcal{F}_n$, and by law of large numbers, $\tbfV_{T,n} \xrightarrow{p} \pi(1-\pi)E[Var\{(\bZ(1)-\bZ(0)|S\}]$, $\tbfV_{IT,n} \xrightarrow{p} \tbfV_{UI}$, $\tbfV_{I,n} \xrightarrow{p} \tbfV_I$, then $\psi_{n}(\bt_1) \xrightarrow{p} \exp\!\Big(-\tfrac12 \bt_1^\top \bSigma_e\, \bt_1\Big)$, where  $\bSigma_e = \pi(1-\pi)E[Var\{(\bZ(1)-\bZ(0)|S\}] - \tbfV_{UI}^\top \tbfV_I^{-1} \tbfV_{UI}$.
For $\eta_n(\bt_2)$, a similar prove shows that $\eta_n(\bt_2) = E\!\big[\exp\{\,i \bt_2^\top \tbfV_{UI}^\top \tbfV_I^{-1/2} \bR_{q,t}\}\big]$
Finally, by the classical (unconditional) CLT across units,
$\bR_n\ \xrightarrow{d}  N(0,\ \bSigma_r)$ with $\bSigma_r = Var\left\{\pi \bZ(1) + (1-\pi) \bZ(0)\right\}$.

To obtain the unconditional characteristic function of $\bU_n(A)$, we take the unconditional expectation and use the dominated convergence theorem to get
\begin{align*}
    & E\!\left[\exp\!\Big\{i\big(\bt_1^\top \tbepsilon_n
+ \bt_2^\top \tbfV_{IT,n}^\top   \tbfV_{I,n}^{-0.5} \bR_{q,t}  + \bt_3^\top \bR_n)\big)\Big\}\ \right]\\
&\rightarrow \exp\!\Big(-\tfrac12 \bt_1^\top \bSigma_e\, \bt_1\Big) \eta_n(\bt_2) E[\exp\{i\bt_3^\top \bR_n\}] \\
&= \exp\!\Big\{-\tfrac12 \bt_1^\top (\bSigma_e +\bSigma_r) \bt_1\Big\} E\!\big[\exp\{\,i \bt_2^\top \tbfV_{UI}^\top \tbfV_I^{-1/2} \bR_{q,t}\}\big].
\end{align*}
This result shows that $\bU_n(A)$ coverages in distribution to $\tbepsilon_k + \tbfV_{UI}^\top \tbfV_I^{-1/2} \bR_{q,t}$ with $\tbepsilon_k \sim N(0, \bSigma_e +\bSigma_r)$ and $\bR_{q,t}$ being independent  of $\tbepsilon_k$. 
Since direct algebra shows
\begin{align*}
    \bSigma_e +\bSigma_r &= \pi(1-\pi)E[Var\{(\bZ(1)-\bZ(0)|S\}] - \tbfV_{UI}^\top \tbfV_I^{-1} \tbfV_{UI} + Var\left\{\pi \bZ(1) + (1-\pi) \bZ(0)\right\} \\
    &= Var(A^* \bZ(1) + (1-A^*) \bZ(0)) - \pi(1-\pi)Var(E[\bZ(1)-\bZ(0)|S]) - \tbfV_{UI}^\top \tbfV_I^{-1} \tbfV_{UI},
\end{align*}
we obtain the desired asymptotic distribution.
\end{proof}

\subsubsection{Proof of Theorem 3}
\begin{proof}[Proof of Theorem 3.]
The proof for consistency, asymptotic normality, and weak convergence is the same as the proof of Theorem 1 except that Lemmas~\ref{lemma1} and \ref{lemma2} are substituted by Lemmas~\ref{lemma4} and \ref{lemma5}, respectively. 

As a result, we obtain 
\begin{align*}
    \sqrt{n}(\widehat{\Delta} - \underline{\Delta}) \xrightarrow{d} \widetilde{\varepsilon}_1 + \widetilde{\bfV}_{UI}^\top \widetilde{\bfV}_I^{-0.5} \bR_{q,t},
\end{align*}
where $\varepsilon_1 \sim N(0, \widetilde{V}-\widetilde{\bfV}_{UI}^\top \widetilde{\bfV}_I^{-1}\widetilde{\bfV}_{UI})$ and $\bR_{q,t} \sim  \bD|( \bD^\top  \bD < t)$,  where $\bD \sim N(\bzero, \mathbf{I}_q)$ and $\varepsilon_1$ is independent of $\bR_{q,t}$.
Furthermore, Lemma~\ref{lemma5} implies that
\begin{align*}
    \widetilde{V} &= V - \pi(1-\pi) E[E[IF(\bO(1)) - IF(\bO(0))|S]^2] \\
    \widetilde{\bfV}_I &= \frac{1}{\pi(1-\pi)}E[Var(\bX_i^r|S)] \\
    \widetilde{\bfV}_{UI} &= E[Cov\{\bX_i^r, IF(\bO(1)) - IF(\bO(0))|S\}],
\end{align*}
where $V = Var\{A^* IF(\bO(1))+(1-A^*)IF(\bO(0))\}$ is the asymptotic variance under simple randomization. Since $\widetilde{\bfV}_{UI}^\top \widetilde{\bfV}_I^{-0.5}$ is a vector, Lemma A1 of \cite{li2018asymptotic} shows that $\widetilde{\bfV}_{UI}^\top \widetilde{\bfV}_I^{-0.5} \bR_q \overset{d}{=} \sqrt{\widetilde{\bfV}_{UI}^\top \widetilde{\bfV}_I^{-1} \widetilde{\bfV}_{UI}}\ D_1|\bD^\top \bD < t$. Defining
\begin{align*}
    \widetilde{R}^2 &= \frac{\widetilde{\bfV}_{UI}^\top \widetilde{\bfV}_I^{-1} \widetilde{\bfV}_{UI}}{\widetilde{V}}.
\end{align*}
we get the desired asymptotic distribution. 

To consistently estimate $\widetilde{V}$ and $\widetilde{R}^2$, we define 
$\widehat{p}_s= n^{-1}\sum_{i=1}^n I\{S_i=s\}$ and 
$\widehat{IF}_i$ as the first element of $\left\{ \frac{1}{n}\sum_{i=1}^n \dot{\bpsi}(\bO_i; \widehat{\btheta}) \right\}^{-1} \bpsi(\bO_i; \widehat{\btheta})$. We define
\begin{align}
    \widehat{\widetilde{V}} &= \widehat{V} - \pi(1-\pi) \sum_{s \in \mathcal{S}} \widehat{p}_s \widehat{d}_s^2, \label{def: vv-hat}\\
    \widehat{\widetilde{R}^2} &= \frac{1}{\widehat{\widetilde{V}}} \widehat{C}^\top \widehat{Var}(\widetilde{\bI})^{-1} \widehat{C}, \label{def: rr-hat}
\end{align}
where
\begin{align*}
    \widehat{V} & = \frac{1}{n} \sum_{i=1}^n\widehat{IF}_i^2, \\
    \widehat{d}_s &= \frac{1}{n} \sum_{i=1}^n \frac{I\{S_i=s\}} {\widehat{p}_s} \frac{A_i-\pi}{\pi(1-\pi)}  \widehat{IF}_i, \\
    \widehat{C} &=  \sum_{s \in \mathcal{S}} \widehat{p}_s \left[ \frac{1}{n}\sum_{i=1}^n \frac{I\{S_i=s\}}{\widehat{p}_s}  \frac{A_i-\pi}{\pi(1-\pi)}  \widehat{IF}_i \bX_i^r -\widehat{d}_s\overline{\bX}_s^r \right], \\
    \overline{\bX}_s^r &= \frac{1}{n} \sum_{i=1}^n \frac{I\{S_i=s\}} {\widehat{p}_s}\bX_i^r. 
\end{align*}
Following a similar proof to the last part of the proof for Theorem 1 and using Lemma~\ref{lemma4} along with the continuous mapping theorem, we obtain $\widehat{V} \xrightarrow{p} V$, $\widehat{p}_s \xrightarrow{p} P(S=s)$, $\widehat{d}_s \xrightarrow{p} E[IF(\bO(1))-IF(\bO(0)|S]$, $\overline{\bX}_s^r \xrightarrow{p} E[\bX^r|S]$, $\widehat{C} \xrightarrow{p} E[Cov\{\bX^r, IF(\bO(1))-IF(\bO(0)|S\}]$. Combining these results with Lemma~\ref{lemma3}, we obtain $\widehat{\widetilde{V}} \xrightarrow{p} \widetilde{V}$ and $\widehat{\widetilde{R}^2} \xrightarrow{p} \widetilde{R}^2$.
\end{proof}

\subsubsection{Proof of Theorem 4}
\begin{proof}[Proof of Theorem 4.]
By Theorem 2 and 3, the influence function for the DR-WLS estimator is given in Equation~\eqref{IF_formula} and
\begin{align*}
    IF(\bO(1)) - IF(\bO(0)) &= \left(\frac{\dot{\underline{f}}_1}{\pi} - \boldsymbol{c}_1^\top \bZ(1) \right)\frac{R(1)\{Y(1)-h(\bZ(1))\}}{ e(\bZ(1))} - \boldsymbol{c}_2^\top \bZ(1)\{R(1)- e(\bZ(1))\} \\
    & -  \left(\frac{\dot{\underline{f}}_0}{1-\pi} - \boldsymbol{c}_1^\top \bZ(0) \right)\frac{R(0)\{Y(0)-h(\bZ(0))\}}{ e(\bZ(0))} + \boldsymbol{c}_2^\top \bZ(0)\{R(0)- e(\bZ(0))\}.
\end{align*}
When the outcome missing model is correctly specified, we have $\boldsymbol{c}_1 = \bzero$, and $E[\boldsymbol{c}_2^\top \bZ(a)\{R(a)- e(\bZ(a)\}|\bX] = E[\boldsymbol{c}_2^\top \bZ(a)E[R(a)- e(\bZ(a))|\bZ(a)]|\bX] = \bzero$. When the outcome model is correctly specified, we have $\boldsymbol{c}_2 = \bzero$, and under MAR,
\begin{align*}
    E\left[\boldsymbol{c}_1^\top \bZ(a) \frac{R(a)\{Y(a)-h(\bZ(a))\}}{ e(\bZ(a))}\bigg| \bX \right] = E\left[\boldsymbol{c}_1^\top \bZ(a) \frac{E[R(a)|\bZ(a)]\{E[Y(a)|\bZ(a)]-h(\bZ(a))\}}{ e(\bZ(a))}\bigg|\bX\right] = 0.
\end{align*}
Therefore, when at least one of the two working models is correctly specified, we have
\begin{equation}\label{eq: if1-if0}
    E[IF(\bO(1)) - IF(\bO(0))|\bX] = E\left[\frac{\dot{\underline{f}}_1}{\pi} \frac{R(1)\{Y(1)-h(\bZ(1))\}}{ e(\bZ(1))} -  \frac{\dot{\underline{f}}_0}{1-\pi} \frac{R(0)\{Y(0)-h(\bZ(0))\}}{ e(\bZ(0))}\bigg|\bX\right].
\end{equation}
We next show $E[\{IF(\bO(1)) - IF(\bO(0))\}\widetilde{\bX}] =\bzero$ under each of the condition (1), (2), or (3), where $\widetilde{\bX} = (\bX^r, S)$. Under condition (1), we have $\pi = 0.5$ and $\dot{\underline{f}}_1 = - \dot{\underline{f}}_0 =1$, and so the definition of $\bpsi(\bO_i;\btheta)$ implies that
\begin{align*}
 E[\{IF(\bO(1)) - IF(\bO(0))\}\widetilde{\bX}] &= 4 E\left[\left\{\pi \frac{R(1)\{Y(1)-h(\bZ(1))\}}{ e(\bZ(1))} + (1-\pi) \frac{R(0)\{Y(0)-h(\bZ(0))\}}{ e(\bZ(0))}\right\}\widetilde{\bX}\right]\\
     &= \bzero.
\end{align*}
Under condition (2), Equation~\eqref{eq: if1-if0} remains correct. Furthermore, the definition of $\bpsi(\bO_i;\btheta)$ now implies that $E\left[\frac{R(a)\{Y(a)-h(\bZ(a))\}}{ e(\bZ(a))}\widetilde{\bX}\right] = \bzero$. Therefore, we can get $E[\{IF(\bO(1)) - IF(\bO(0))\}\widetilde{\bX}] = \bzero$. Under condition (3) and MAR, we have $E[IF(\bO(1)) - IF(\bO(0))|\bX] = \bzero$, which also implies $E[\{IF(\bO(1)) - IF(\bO(0))\}\widetilde{\bX}] =\bzero$.

Finally, since $E[\{IF(\bO(1)) - IF(\bO(0))\}\widetilde{\bX}] =\bzero$ and $S$ is a categorical variable, we get $E[\{IF(\bO(1)) - IF(\bO(0))\}\bX^r] = \bzero$ and $E[IF(\bO(1)) - IF(\bO(0))|S] = \bzero$, yielding 
\begin{align*}
    & E[Cov\{IF(\bO(1)) - IF(\bO(0)), \bX^r|S\}]\\
    &= E[\{IF(\bO(1)) - IF(\bO(0))\}\bX^r] - E[E[IF(\bO(1)) - IF(\bO(0))|S]E[\bX^r|S]] \\
    &= \bzero.
\end{align*}
As a result, $\widetilde{R}^2 = 0$ and $\widetilde{V} = V$.

For the mixed-ANCOVA estimator, the proof is the same as the proof of Proposition~\ref{prop: mixed-ancova} with $\bX^r$ replaced by $(\bX^r, S)$, and we omitted the proof.
\end{proof}

\subsection{Efficient estimation with machine learning}
\subsubsection{Proof of Theorem 5}

\begin{lemma}[Convergence with random indices]\label{lemma: random indices}
Let $X_n$ and $Y_n$ be two independent random sequences defined in the same probability space. We assume $Y_n$ take values in positive integers, $E[X_n^2] = o(n^{-r})$ and $Y_n/n \xrightarrow{p} c$ for positive constants $r$ and $c$. Then $E[X_{Y_n}^2|Y_n] = o_p(n^{-r})$.
\end{lemma}

\begin{proof}
Since  $E[X_n^2] = o(n^{-r})$, for any $\varepsilon >0$, there exists an integer $N_{\varepsilon}$ such that $E[n^r X_n^2] < \varepsilon$ for any $n \ge N_{\varepsilon}$. Then, when $n > 2N_{\varepsilon}/c$, we have
\begin{align*}
    P(E[n^rX_{Y_n}^2|Y_n] > \varepsilon) &= \sum_{i=1}^{\infty} P(E[n^rX_{Y_n}^2|Y_n] > \varepsilon, Y_n = i) \\
    &= \sum_{i=1}^{\infty} P(E[n^rX_i^2|Y_n=i] > \varepsilon, Y_n = i) \\
    &= \sum_{i=1}^{\infty} I\{E[n^rX_i^2] > \varepsilon\} P(Y_n = i) \\
    &= \sum_{i < N_{\varepsilon}} I\{E[n^rX_i^2] > \varepsilon\} P(Y_n = i) \\
    &\le P(Y_n < N_{\varepsilon}) \\
    &\le P(Y_n < nc/2) \\
    &\le P(|Y_n/n - c| > c/2) \\
    &\rightarrow 0,
\end{align*}
where the last line results from $Y_n/n \xrightarrow{p} c$. This implies $E[X_{Y_n}^2|Y_n] = o_p(n^{-r})$.
\end{proof}

\begin{lemma}\label{lemma: O_a-k}
Given Assumption 4 (1) and defining $\mathcal{O}_{a,-k}^* = \{(R_iY_i,Y_i,\bX_i): i = 1,\dots, n;i \notin \mathcal{I}_k; A_i^* =a\}$ for $A_i^*$ generated by simple randomization, we have
\begin{align*}
    E[\{\widehat{\kappa}_a(\bX_i;\mathcal{O}_{a,-k}^*) - E[R(a)|\bX_i]\}^2|\mathcal{O}^*_{a, -k}] &= o_p(n^{-1/2}),\\
    E[\{\widehat{\eta}_a(\bX_i;\mathcal{O}_{a,-k}^*) - E[Y(a)|\bX_i]\}^2|\mathcal{O}^*_{a, -k}] &= o_p(n^{-1/2}).
\end{align*}
\end{lemma}

\begin{proof}
Denoting $N_{a,-k} = \sum_{i=1}^n I\{A_i^*=a, i \notin \mathcal{I}_k\}$ and $Z= E[\{\widehat{\kappa}_a(\bX_i;\mathcal{O}_{a,-k}^*) - E[R(a)|\bX_i]\}^2|\mathcal{O}^*_{a, -k}]$, we have $P(n^{1/2} Z > \varepsilon)=E[P(n^{1/2} Z > \varepsilon|N_{a,-k})]$. By the dominated convergence theorem, it suffices to show $P(n^{1/2} Z > \varepsilon|N_{a,-k}) = o_p(1)$. By the Markov's inequality, we have $P(n^{1/2} Z > \varepsilon|N_{a,-k}) \le \varepsilon^{-1} E[n^{1/2} Z|N_{a,-k}]$, and we only need to show $E[Z|N_{a,-k}] = o_p(n^{-1/2})$. Since $N_{a,-k}$ is a function of $\mathcal{O}_{a,-k}^*$, it is equivalent to show $E[\{\widehat{\kappa}_a(\bX_i;\mathcal{O}_{a,-k}^*) - E[R(a)|\bX_i]\}^2|N_{a, -k}]= o_p(n^{-1/2})$. We next apply Lemma~\ref{lemma: random indices} to prove this result.

Define $X_n = \widehat{\kappa}_a(\bX_i; \{\bO_i(a)\}_{i=1}^n) - E[R(a)|\bX_i]$, where $\bO_i(a)= \{R_i(a)Y_i(a), R_i(a), \bX_i\}$,  and $Y_n = N_{a,-k}$. Since $\mathcal{O}_{a,-k}^* = \{\bO_i(a): i = 1,\dots, n;i \notin \mathcal{I}_k; A_i^* =a\}$ and the sample splitting and simple treatment assignment is independent of $\{\bO_i(a)\}_{i=1}^n$, we have $\mathcal{O}_{a,-k}^*|N_{a,-k} \overset{d}{=} \{\bO_i(a): i = 1,\dots, N_{a,-k}\}|N_{a,-k}$. This implies that $E[X_{Y_n}^2|Y_n] = E[\{\widehat{\kappa}_a(\bX_i;\mathcal{O}_{a,-k}^*) - E[R(a)|\bX_i]\}^2|N_{a, -k}]$, and our goal is to show $E[X_{Y_n}^2|Y_n] = o_p(1)$. To see this, we observe that (1) $Y_n/n \xrightarrow{p} \pi^a(1-\pi)^{1-a}(K-1)/K$, (2) $Y_n$ is independent of $\{\bO_i(a)\}_{i=1}^n$ and hence $X_n$, (3) $E[X_n^2] = o_p(n^{-1/2})$ by Assumption 4 (1). Therefore, Lemma~\ref{lemma: random indices} implies $E[X_{Y_n}^2|Y_n] = o_p(1)$. 

For the second result about the convergence of $\widehat{\eta}_a$, we follow the same procedure above and can get the desired result.
\end{proof}

\begin{proof}[Proof of Theorem 5.]
   Define $\pi_a = \pi^a(1-\pi)^{1-a}$, $\mu_a = E[Y(a)]$ and
    \begin{align*}
        D_i(a,\kappa_a,\eta_a) = \frac{I\{A_i=a\}}{\pi_a}\frac{R_i(a)}{\kappa_a(\bX_i)}\{Y_i(a)  - \eta_a(\bX_i)\} + \eta_a(\bX_i).
    \end{align*}
    We further define $\underline{\kappa}_a(\bX) = E[R(a)|\bX]$ and $\underline{\eta}_a(\bX) = E[Y(a)|\bX]$. Then the efficient influence function for $\Delta^*$ becomes
    \begin{align*}
      EIF(\bO_i) =  \sum_{a=0}^1 f'_a \{D_i(a,\underline{\kappa}_a,\underline{\eta}_a)-\mu_a\}.
    \end{align*}
    For the asymptotic linearity, then it suffices to prove $\sqrt{n}(\widehat{\mu}_a^{\textup{Eff}}-\mu_a) = n^{-1/2} \sum_{i=1}^n \{D_i(a, \underline{\kappa}_a,\underline{\eta}_a)-\mu_a\}$  for each $a = 0,1$. 
    For brevity, we denote $\widehat{\kappa}_{a,-k}(\bX) = \widehat{\kappa}_{a}(\bX;\mathcal{O}_{a,-k})$ and $\widehat{\eta}_{a,-k}(\bX) = \widehat{\eta}_{a}(\bX;\mathcal{O}_{a, -k})$. Then we get $\widehat{\mu}_a^{\textup{Eff}} = K^{-1} \sum_{k=1}^K |\mathcal{I}_k|^{-1} \sum_{i \in \mathcal{I}_k} D_i(a, \widehat{\kappa}_{a,-k}, \widehat{\eta}_{a,-k})$ and
    \begin{align*}
        \sqrt{n}(\widehat{\mu}_a^{\textup{Eff}}-\mu_a) &= K^{-1/2} \sum_{k=1}^K |\mathcal{I}_k|^{-1/2} \sum_{i \in \mathcal{I}_k} \{D_i(a, \widehat{\kappa}_{a, -k}, \widehat{\eta}_{a, -k}) - \mu_a \} \\
        &= K^{-1/2} \sum_{k=1}^K |\mathcal{I}_k|^{-1/2} \sum_{i \in \mathcal{I}_k} \{D_i(a, \underline{\kappa}_{a}, \underline{\eta}_{a}) - \mu_a \} \\
        &\quad +  K^{-1/2} \sum_{k=1}^K  \underbrace{|\mathcal{I}_k|^{-1/2}\sum_{i \in \mathcal{I}_k} \{D_i(a, \widehat{\kappa}_{a, -k}, \widehat{\eta}_{a, -k}) - D_i(a, \underline{\kappa}_{a}, \underline{\eta}_{a})\}}_{{\textup{denoted as}\ U_k}}.
    \end{align*}
Since $K$ is fixed, it suffices to prove $U_k = o_p(1)$ for all $k$. To this end, we decompose
\begin{align*}
    U_k = U_{k1} + U_{k2} + U_{k3},
\end{align*}
where
\begin{align*}
    U_{k1} &= |\mathcal{I}_k|^{-1/2} \sum_{i \in \mathcal{I}_k} \frac{I\{A_i=a\}}{\pi_a} \frac{\underline{\kappa}_a(\bX_i) - \widehat{\kappa}_{a,-k}(\bX_i)}{\underline{\kappa}_a(\bX_i)\widehat{\kappa}_{a,-k}(\bX_i)}\{Y_i(a) - \underline{\eta}_a(\bX_i)\},\\
    U_{k2} &= |\mathcal{I}_k|^{-1/2} \sum_{i \in \mathcal{I}_k} \left(\frac{I\{A_i=a\}}{\pi_a} \frac{R_i(a)}{\underline{\kappa}_a(\bX_i)}-1\right)  \{\underline{\eta}_a(\bX_i)-\widehat{\eta}_{a,-k}(\bX_i)\},\\
    U_{k3} &= |\mathcal{I}_k|^{-1/2} \sum_{i \in \mathcal{I}_k} \frac{I\{A_i=a\}}{\pi_a} \frac{\underline{\kappa}_a(\bX_i) - \widehat{\kappa} _{a,-k}(\bX_i)}{\underline{\kappa}_a(\bX_i)\widehat{\kappa}_a(\bX_i)}\{\underline{\eta}_{a}(\bX_i)-\widehat{\eta}_{a,-k}(\bX_i)\},
\end{align*}
and aim to show $U_{kj} = o_p(1)$ for $j=1,2,3$.  For this purpose, we denote $\mathcal{F}_n$ is the sigma field generated by $\bW_1,\dots,\bW_n$. If we can show that, for any $\varepsilon>0$, $P(|U_{kj}|>\varepsilon|\mathcal{F}_n) \xrightarrow{p} 0$, then the dominated convergence theorem for convergence in probability \citep[Theorem 1.5.3]{durrett2019probability} implies $E[|U_{kj}|]\rightarrow 0$ and hence $U_{kj} = o_p(1)$.

Under rerandomization, we define $Q^*=\bI^*{}^\top \widehat{Var}(\bI) \bI^*$, $\widehat{\kappa} _{a,-k}^*(\bX_i) = \widehat{\kappa} _a(\bX_i; \mathcal{O}^*_{a,-k})$, and $\widehat{\eta} _{a,-k}^*(\bX_i) = \widehat{\eta}_a(\bX_i; \mathcal{O}^*_{a, -k})$, where $\mathcal{O}_{a,-k}^* = \{(R_iY_i,Y_i,\bX_i): i = 1,\dots, n;i \notin \mathcal{I}_k; A_i^* =a\}$  is the same as $\mathcal{O}_{a,-k}$ except that $A_i$ is replaced by $A_i^*$, where $A_i^*$ is generated by simple randomization. 
For $j= 1,2,3$, we further define 
\begin{align*}
    U_{k1}^* &= |\mathcal{I}_k|^{-1/2} \sum_{i \in \mathcal{I}_k} \frac{I\{A_i^*=a\}}{\pi_a} \frac{\underline{\kappa}_a(\bX_i) - \widehat{\kappa}_{a,-k}^*(\bX_i)}{\underline{\kappa}_a(\bX_i)\widehat{\kappa}_{a,-k}^*(\bX_i)}\{Y_i(a) - \underline{\eta}_a(\bX_i)\},\\
    U_{k2}^* &= |\mathcal{I}_k|^{-1/2} \sum_{i \in \mathcal{I}_k} \left(\frac{I\{A_i^*=a\}}{\pi_a} \frac{R_i(a)}{\underline{\kappa}_a(\bX_i)}-1\right)  \{\underline{\eta}_a(\bX_i)-\widehat{\eta}_{a,-k}^*(\bX_i)\},\\
    U_{k3}^* &= |\mathcal{I}_k|^{-1/2} \sum_{i \in \mathcal{I}_k} \frac{I\{A_i^*=a\}}{\pi_a} \frac{\underline{\kappa}_a(\bX_i) - \widehat{\kappa} _{a,-k}^*(\bX_i)}{\underline{\kappa}_a(\bX_i)\widehat{\kappa}_{a,-k}^*(\bX_i)}\{\underline{\eta}_{a}(\bX_i)-\widehat{\eta}_{a,-k}^*(\bX_i)\}, 
\end{align*}
i.e., $U^*_{kj}$ is the same as $U_{kj}$ except that $I\{A_i=a\}$ is replaced by $I\{A^*_i=a\}$ and $\{\widehat{\kappa}_{a,-k}(\bX_i), \widehat{\eta}_{a,-k}(\bX_i)\}$ are replaced by $\{\widehat{\kappa}_{a,-k}^*(\bX_i), \widehat{\eta}_{a,-k}^*(\bX_i)\}$. By the rerandomization scheme, we have $U_{kj}|\mathcal{F}_n\overset{d}{=} U^*_{kj}|( \mathcal{F}_n, Q^* < t)$ for $j=1,2,3$. 
Then, we have
\begin{align*}
   P(|U_{kj}|> \varepsilon|\mathcal{F}_n) &= P(|U_{kj}^*|> \varepsilon|\mathcal{F}_n, Q^*<t)\\
   &= P(Q^*<t|\mathcal{F}_n)^{-1} P(|U_{kj}^*|> \varepsilon, Q^*<t|\mathcal{F}_n) \\
   &\le  P(Q^*<t|\mathcal{F}_n)^{-1} P(|U_{kj}^*|> \varepsilon|\mathcal{F}_n).
\end{align*}
Since Lemma~\ref{lem:condCLT} leads to $\lim_{n \rightarrow \infty}P(Q^*<t) = P(\mathcal{X}^2_q < t|\mathcal{F}_n) > 0$, the above derivation yields that $P(|U_{kj}^*|> \varepsilon|\mathcal{F}_n) = o_p(1)$ implies $P(|U_{kj}|> \varepsilon|\mathcal{F}_n) =o_p(1)$ and hence $U_{kj} = o_p(1)$. By Markov's inequality, it suffices to show $U_{kj}^* = o_p(1)$. For this purpose, we observe $P(|U_{kj}^*|> \varepsilon) = E[P(|U_{kj}^*|> \varepsilon|\mathcal{O}^*_{a,-k})]$. By the dominated convergence theorem, it suffices to show $P(|U_{k1}^*|> \varepsilon|\mathcal{O}^*_{a,-k}) = o_p(1)$. Since the Markov's inequality shows $P(|U_{kj}^*|> \varepsilon|\mathcal{O}^*_{a,-k}) \le \varepsilon^{-2} E[U_{kj}^*{}^2|\mathcal{O}^*_{a,-k}]$, it suffices to show $E[U_{kj}^*{}^2|\mathcal{O}^*_{a,-k}] = o_p(1)$ for $j=1,2,3$.

For showing $E[U_{k1}^*{}^2|\mathcal{O}^*_{a,-k}] = o_p(1)$, since $\widehat{\kappa}_{a,-k}^*(\bX_i)$ is a deterministic function of $\bX_i$ given $\mathcal{O}_{a,-k}^*$, and $\{A_i^*,Y_i(a), R_i(a), \bX_i: i \in \mathcal{I}_k\}$ is independent of $\mathcal{O}_{a,-k}^*$, we have
\begin{align*}
    E[U_{k1}^*|\mathcal{O}_{a,-k}^*] &= |\mathcal{I}_k|^{-1/2} \sum_{i \in \mathcal{I}_k} E\left[\frac{I\{A_i^*=a\}}{\pi_a} \frac{\underline{\kappa}_a(\bX_i) - \widehat{\kappa}_{a,-k}^*(\bX_i)}{\underline{\kappa}_a(\bX_i)\widehat{\kappa}_{a,-k}^*(\bX_i)}\{Y_i(a) - \underline{\eta}_a(\bX_i)\}\bigg|\mathcal{O}_{a,-k}^*\right] \\
    &= |\mathcal{I}_k|^{-1/2} \sum_{i \in \mathcal{I}_k} E\left[\frac{\underline{\kappa}_a(\bX_i) - \widehat{\kappa}_{a,-k}^*(\bX_i)}{\underline{\kappa}_a(\bX_i)\widehat{\kappa}_{a,-k}^*(\bX_i)}\{Y_i(a) - \underline{\eta}_a(\bX_i)\}\bigg|\mathcal{O}_{a,-k}^*\right] \\
    &= |\mathcal{I}_k|^{-1/2} \sum_{i \in \mathcal{I}_k} E\left[\frac{\underline{\kappa}_a(\bX_i) - \widehat{\kappa}_{a,-k}^*(\bX_i)}{\underline{\kappa}_a(\bX_i)\widehat{\kappa}_{a,-k}^*(\bX_i)}\{E[Y_i(a)|\bX_i,\mathcal{\mathcal{O}}_{a,-k}^*] - \underline{\eta}_a(\bX_i)\}\bigg|\mathcal{O}_{a,-k}^*\right] \\
    &= |\mathcal{I}_k|^{-1/2} \sum_{i \in \mathcal{I}_k} E\left[\frac{\underline{\kappa}_a(\bX_i) - \widehat{\kappa}_{a,-k}^*(\bX_i)}{\underline{\kappa}_a(\bX_i)\widehat{\kappa}_{a,-k}^*(\bX_i)}\{E[Y_i(a)|\bX_i] - \underline{\eta}_a(\bX_i)\}\bigg|\mathcal{O}_{-k}^*\right] \\    
    &= 0.
\end{align*}
Since $\{(A_i^*, Y_i(a), \bX_i), i \in \mathcal{I}_k\}$ are independent and identically distributed and are independent of $\mathcal{O}^*_{a,-k}$,  then $\{(A_i^*, Y_i(a), \bX_i), i \in \mathcal{I}_k\}$ are independent and identically distributed given  $\mathcal{O}^*_{a,-k}$.
Therefore, 
\begin{align*}
    E[U_{k1}^*{}^2|\mathcal{O}^*_{a, -k}] &= E[\{U_{k1}^* - E[U_{k1}^*|\mathcal{O}_{a,-k}^*]\}^2|\mathcal{O}^*_{a,-k}] \\
    &= Var(U_{k1}^*|\mathcal{O}^*_{a,-k}) \\
    &= Var\left\{\frac{I\{A_i^*=a\}}{\pi_a} \frac{\underline{\kappa}_a(\bX_i) - \widehat{\kappa}_{a,-k}^*(\bX_i)}{\underline{\kappa}_a(\bX_i)\widehat{\kappa}_{a,-k}^*(\bX_i)}\{Y_i(a) - \underline{\eta}_a(\bX_i)\}\bigg|\mathcal{O}^*_{a,-k}\right\} \\
    &= E\left[\left[\frac{I\{A_i^*=a\}}{\pi_a} \frac{\underline{\kappa}_a(\bX_i) - \widehat{\kappa}_{a,-k}^*(\bX_i)}{\underline{\kappa}_a(\bX_i)\widehat{\kappa}_{a,-k}^*(\bX_i)}\{Y_i(a) - \underline{\eta}_a(\bX_i)\}\right]^2\bigg|\mathcal{O}^*_{a,-k}\right] \\
    &\le C\times E[\{\underline{\kappa}_a(\bX_i) - \widehat{\kappa}_{a,-k}^*(\bX_i)\}^2|\mathcal{O}^*_{a,-k}],
\end{align*}
where $C$ is a constant and the last inequality comes from  Assumptions 2 and 4(2) that $\{\underline{\kappa}_a(\bX_i)\}^{-1}$, $\{\widehat{\kappa}_{a,-k}^*(\bX_i)\}^{-1}$, and $E[Y_i^2(a)|\bX_i]$ are all uniformly bounded. 
By Lemma~\ref{lemma: O_a-k}, we have $E[\{\underline{\kappa}_a(\bX_i) - \widehat{\kappa}_{a,-k}^*(\bX_i)\}^2|\mathcal{O}^*_{a,-k}] = o_p(1)$, which implies $U_{k1} = o_p(1)$.

For $U_{k2}^*$, following a similar proof to $U_{k1}^*$, we have $E[U_{k2}^*|\mathcal{O}_{a,-k}^*] = 0$ and
\begin{align*}
    E[U_{k2}^*{}^2|\mathcal{O}^*_{a,-k}] &= E[\{U_{k2}^* - E[U_{k2}^*|\mathcal{O}_{a,-k}^*]\}^2|\mathcal{O}^*_{a,-k}] \\
    &= Var(U_{k2}^*|\mathcal{O}^*_{a,-k}) \\
    &= Var\left[\left(\frac{I\{A_i^*=a\}}{\pi_a} \frac{R_i(a)}{\underline{\kappa}_a(\bX_i)}-1\right)  \{\underline{\eta}_a(\bX_i)-\widehat{\eta}_{a,-k}^*(\bX_i)\}\bigg|\mathcal{O}^*_{a,-k}\right] \\
    &= E\left[\left(\frac{I\{A_i^*=a\}}{\pi_a} \frac{R_i(a)}{\underline{\kappa}_a(\bX_i)}-1\right)^2  \{\underline{\eta}_a(\bX_i)-\widehat{\eta}_{a,-k}^*(\bX_i)\}^2\bigg|\mathcal{O}^*_{a,-k}\right] \\
    &\le C\times E[\{\underline{\eta}_a(\bX_i)-\widehat{\eta}_{a,-k}^*(\bX_i)\}^2|\mathcal{O}^*_{a,-k}],
\end{align*}
where $C$ is a constant and the last inequality comes from Assumption 2 that $\{\underline{\kappa}_a(\bX_i)\}^{-1}$ is uniformly bounded. 
By Lemma~\ref{lemma: O_a-k}, we have $E[\{\underline{\eta}_a(\bX_i) - \widehat{\eta}_{a,-k}^*(\bX_i)\}^2|\mathcal{O}^*_{a,-k}] = o_p(1)$, which implies $U_{k2} = o_p(1)$.

For showing $E[U_{k3}^*{}^2|\mathcal{O}^*_{a,-k}] = o_p(1)$, the proof is slightly different. Since $\{(A_i^*,\bX_i): i\in \mathcal{I}_k\}$ are independent and identically distributed given $\mathcal{O}_{a,-k}^*$ and $\widehat{\kappa}_{a,-k}, \widehat{\eta}_{a,-k}$ are fixed functions given $\mathcal{O}_{a,-k}^*$, we have
\begin{align*}
    E[U_{k3}^*|\mathcal{O}^*_{a,-k}] &= \sqrt{|\mathcal{I}_k|} E\left[\frac{\underline{\kappa}_a(\bX_i) - \widehat{\kappa} _{a,-k}^*(\bX_i)}{\underline{\kappa}_a(\bX_i)\widehat{\kappa}_{a,-k}^*(\bX_i)}\{\underline{\eta}_{a}(\bX_i)-\widehat{\eta}_{a,-k}^*(\bX_i)\}\bigg|\mathcal{O}^*_{a,-k}\right] \\
    &\le C \sqrt{|\mathcal{I}_k|} E\left[\{\underline{\kappa}_a(\bX_i) - \widehat{\kappa} _{a,-k}^*(\bX_i)\}\{\underline{\eta}_{a}(\bX_i)-\widehat{\eta}_{a,-k}^*(\bX_i)\}\bigg|\mathcal{O}^*_{a,-k}\right] \\
    &\le  C \sqrt{|\mathcal{I}_k|} E\left[\{\underline{\kappa}_a(\bX_i) - \widehat{\kappa} _{a,-k}^*(\bX_i)\}^2\bigg|\mathcal{O}^*_{a,-k}\right]^{1/2}E\left[\{\underline{\eta}_{a}(\bX_i)-\widehat{\eta}_{a,-k}^*(\bX_i)\}^2\bigg|\mathcal{O}^*_{a,-k}\right]^{1/2}\\
    &=  C \sqrt{|\mathcal{I}_k|} o_p(|\mathcal{I}_k|^{-1/4})  o_p(|\mathcal{I}_k|^{-1/4})\\
    &= o_p(1),
\end{align*}
where the first inequality results from the regularity condition that $\{\underline{\kappa}_a(\bX_i)\}^{-1}, \{\widehat{\kappa} _{a,-k}^*(\bX_i)\}^{-1}$ are uniformly bounded, the second inequality is implied by the Cauchy-Schwarz inequality, and the next equation comes from Lemma~\ref{lemma: O_a-k}. We then have
\begin{align*}
    E[U_{k3}^*{}^2|\mathcal{O}^*_{a,-k}] &= Var(U_{k3}^*|\mathcal{O}^*_{a,-k}) + E[U_{k3}^*|\mathcal{O}^*_{a,-k}]^2\\
    &= Var\left[ \frac{I\{A_i^*=a\}}{\pi_a} \frac{\underline{\kappa}_a(\bX_i) - \widehat{\kappa} _{a,-k}^*(\bX_i)}{\underline{\kappa}_a(\bX_i)\widehat{\kappa}_{a,-k}^*(\bX_i)}\{\underline{\eta}_{a}(\bX_i)-\widehat{\eta}_{a,-k}^*(\bX_i)\}\bigg| \mathcal{O}_{a,-k}^*\right] +o_p(1)\\
    &\le E\left[\left[\frac{I\{A_i^*=a\}}{\pi_a} \frac{\underline{\kappa}_a(\bX_i) - \widehat{\kappa} _{a,-k}^*(\bX_i)}{\underline{\kappa}_a(\bX_i)\widehat{\kappa}_{a,-k}^*(\bX_i)}\{\underline{\eta}_{a}(\bX_i)-\widehat{\eta}_{a,-k}^*(\bX_i)\}\right]^2\bigg| \mathcal{O}_{a,-k}^*\right] + o_p(1) \\
    &\le C\times E[\{\underline{\eta}_{a}(\bX_i)-\widehat{\eta}_{a,-k}^*(\bX_i)\}^2|\mathcal{O}_{a,-k}^*] + o_p(1)\\
    &= C\times o_p(1) + o_p(1)\\
    &= o_p(1).
\end{align*}
Therefore, we have $U_{k3} = o_p(1)$. This completes the proof for showing the remainder term $U_k = o_p(1)$ and hence asymptotic linearity under rerandomization. 

Then the asymptotic distribution is a direct result of Theorem 1, which also implies consistency. When $\bX$ includes $\bX^r$ for covariate adjustment, we can compute
\begin{align*}
    EIF(\bO(1)) - EIF(\bO(0)) = \sum_{a=0}^1 (-1)^{1-a} f_a'\left[ \frac{1}{\pi_a} \frac{R(a)}{E[R(a)|\bX]}\{Y(a) - E[Y(a)|\bX]\} \right],
\end{align*}
which implies
\begin{align*}
    E\left[EIF(\bO(1)) - EIF(\bO(0))|\bX\right] &= \sum_{a=0}^1 (-1)^{1-a} f_a' \frac{1}{\pi_a} E\left[\frac{R(a)}{E[R(a)|\bX]}\{Y(a) - E[Y(a)|\bX]\}\bigg|\bX\right] \\
    &= 0
\end{align*}
by the missing at random assumption. Hence $Cov(EIF(\bO(1)) - EIF(\bO(0)), \bX) = 0$. If $\bX$ contains $\bX^r$ as a subset, then $Cov(EIF(\bO(1)) - EIF(\bO(0)), \bX^r) = 0$, implying $R^2=0$.

Finally, for consistent variance estimation, we define
\begin{align*}
    \widehat{EIF}_i(a) = \frac{I\{A_i=a\}}{\pi^a(1-\pi)^{1-a}}\frac{R_i}{\widehat{\kappa}_a(\bX_i)}\{Y_i - \widehat{\eta}_a(\bX_i)\} + \widehat{\eta}_a(\bX_i)-\widehat{\mu}_a^{\textup{dml}},
\end{align*}
and $\widehat{EIF}_{i}= f_1' \widehat{EIF}_{i}(1)+ f_0'\widehat{EIF}_{i}(0)$. Under rerandomization, we construct
\begin{align}
    \widehat{V} &= \frac{1}{K}\sum_{k=1}^K \frac{1}{|\mathcal{I}_k|} \sum_{i \in \mathcal{I}_k} \widehat{EIF}_{i}^2, \label{Vhat: ml}\\
    \widehat{R^2} &= \frac{1}{\widehat{V}} \widehat{C}^\top \{n \widehat{Var}(\bI)\}^{-1} \widehat{C},\label{R2hat: ml}
\end{align}
where $\widehat{C} = \frac{1}{K}\sum_{k=1}^K \frac{1}{|\mathcal{I}_k|} \sum_{i \in \mathcal{I}_k} \frac{A_i-\pi}{\pi(1-\pi)} \widehat{EIF}_{i} (\bX_i^r - \overline{\bX}^r)$.

To show consistency of $\widehat{V}$ to $V = Var\{A^* EIF(\bO(1))+(1-A^*)EIF(\bO(0))\}$, we define $Q_k = \frac{1}{|\mathcal{I}_k|} \sum_{i \in \mathcal{I}_k} \widehat{EIF}_{i}^2 - EIF(\bO_i)^2$. Following a similar proof to showing $U_k = o_p(1)$, we can get $Q_k = o_p(1)$. Since Lemma~\ref{lemma1} implies that $\frac{1}{|\mathcal{I}_k|} \sum_{i \in \mathcal{I}_k} EIF(\bO_i)^2 \xrightarrow{p} V$, we have $\widehat{V} \xrightarrow{p} V$. The proof of $ \widehat{R^2} \xrightarrow{p} R^2$ is similar and hence omitted here.

\end{proof}

\subsubsection{Proof of Theorem 6}
\begin{proof}
    We inherit all notation from the proof of Theorem 5. In addition, we define $T_i$ as the random variable encoding the fold that subject $i$ belongs to and define     
    \begin{align*}
        D_i(a,s,k,\kappa_a,\eta_a) = I\{S_i=s, T_i=k\} \left[\frac{I\{A_i=a\}}{\pi_a}\frac{R_i(a)}{\kappa_a(\bX_i)}\{Y_i(a)  - \eta_a(\bX_i)\} + \eta_a(\bX_i)\right],
    \end{align*}
    which satisfy $D_i(a,s,k,\kappa_a,\eta_a) = I\{S_i=s, T_i=k\} D(a, \kappa_a,\eta_a)$. 
    For asymptotic linearity, we then need to show $\sqrt{n}(\widehat{\mu}_a^{\textup{Eff}} - \mu_a) = n^{-1/2} \sum_{k=1}^K \sum_{s \in \mathcal{S}}\sum_{i=1}^n \{D(a, s,k,\underline{\kappa}_a,\underline{\eta}_a) - \mu_a\} + o_p(1)$. 

    Defining $\widehat{\kappa}_{as,-k}(\bX_i) = \widehat{\kappa}_a(\bX_i; \mathcal{O}_{as,-k})$ and $\widehat{\eta}_{as,-k}(\bX_i) = \widehat{\eta}_a(\bX_i; \mathcal{O}_{as,-k})$ with $ \mathcal{O}_{as,-k} = \{(R_iY_i, R_i, \bX_i): A_i = a, S_i=s, T_i\ne k\}$, we have $\widehat{\mu}_a^{\textup{Eff}} = n^{-1}\sum_{k=1}^K \sum_{s \in \mathcal{S}}\sum_{i=1}^n D_i(a, s,k,\widehat{\kappa}_{as,-k},\widehat{\eta}_{as,-k})$ and
    \begin{align*}
        \sqrt{n}(\widehat{\mu}_a^{\textup{Eff}} - \mu_a) &= n^{-1/2} \sum_{k=1}^K \sum_{s \in \mathcal{S}}\sum_{i=1}^n \{D_i(a, s,k,\widehat{\kappa}_{as,-k},\widehat{\eta}_{as,-k}) - \mu_a\} \\
        &= n^{-1/2} \sum_{k=1}^K \sum_{s \in \mathcal{S}}\sum_{i=1}^n \{D(a, s,k,\underline{\kappa}_a,\underline{\eta}_a) - \mu_a\} \\
        &\quad +\sum_{k=1}^K \sum_{s \in \mathcal{S}} \underbrace{n^{-1/2} \sum_{i=1}^n \{D(a, s,k,\widehat{\kappa}_{as,-k},\widehat{\eta}_{as,-k})-D(a, s,k,\underline{\kappa}_a,\underline{\eta}_a)\}}_{\textup{denoted as}\ U_{sk}}.
    \end{align*}
    Since $K$ and $\mathcal{S}$ are fixed, it suffices to prove $U_{sk}=o_p(1)$ for asymptotic linearity. For this goal, we decompose $U_{sk}$ as $U_{sk} = U_{sk1} + U_{sk2} + U_{sk3}$, where
    \begin{align*}
    U_{sk1} &= n^{-1/2} \sum_{i=1}^n \frac{I\{A_i=a,S_i=s, T_i=k\}}{\pi_a} \frac{\underline{\kappa}_{a}(\bX_i) - \widehat{\kappa}_{as,-k}(\bX_i)}{\underline{\kappa}_{a}(\bX_i)\widehat{\kappa}_{as,-k}(\bX_i)}\{Y_i(a) - \underline{\eta}_{a}(\bX_i)\},\\
    U_{sk2} &= n^{-1/2} \sum_{i=1}^n I\{S_i=s, T_i=k\}\left(\frac{I\{A_i=a\}}{\pi_a} \frac{R_i(a)}{\underline{\kappa}_{a}(\bX_i)}-1\right)  \{\underline{\eta}_{a}(\bX_i)-\widehat{\eta}_{as,-k}(\bX_i)\},\\
    U_{sk3} &= n^{-1/2} \sum_{i=1}^n \frac{I\{A_i=a,S_i=s, T_i=k\}}{\pi_a} \frac{\underline{\kappa}_a(\bX_i) - \widehat{\kappa} _{as,-k}(\bX_i)}{\underline{\kappa}_a(\bX_i)\widehat{\kappa} _{as,-k}(\bX_i)}\{\underline{\eta}_{a}(\bX_i)-\widehat{\eta}_{as,-k}(\bX_i)\}.
\end{align*}
This decomposition motivates us to show $U_{skj} = o_p(1)$ for $j=1,2,3$. If we can show that, for any $\varepsilon>0$, $P(|U_{skj}|>\varepsilon|\mathcal{F}_n) \xrightarrow{p} 0$, then the dominated convergence theorem for convergence in probability \citep[Theorem 1.5.3]{durrett2019probability} implies $E[|U_{skj}|]\rightarrow 0$ and hence $U_{skj} = o_p(1)$.

 We define $\widetilde{Q}=\widetilde{\bI}^\top \widehat{Var}(\widetilde{\bI}) \widetilde{\bI}$, $\widetilde{\kappa}_{as,-k}(\bX_i) = \widehat{\kappa}_{a}(\bX_i; \widetilde{\mathcal{O}}_{as,-k})$, and $\widetilde{\eta}_{as,-k}(\bX_i) = \widehat{\eta}_a(\bX_i; \widetilde{\mathcal{O}}_{as,-k})$, where $\widetilde{\mathcal{O}}_{as, -k}=\{(R_i(a)Y_i(a), R_i(a), \bX_i): \tA_i = a, S_i=s, T_i\ne k\}$ is the same as $\mathcal{O}_{as, -k}$ except that $A_i$ is replaced by $\tA_i$, where $\tA_i$ is generated by stratified randomization. We further define
     \begin{align*}
    \widetilde{U}_{sk1} &= n^{-1/2} \sum_{i=1}^n \frac{I\{\tA_i=a,S_i=s, T_i=k\}}{\pi_a} \frac{\underline{\kappa}_{a}(\bX_i) - \widetilde{\kappa}_{as,-k}(\bX_i)}{\underline{\kappa}_{a}(\bX_i)\widetilde{\kappa}_{as,-k}(\bX_i)}\{Y_i(a) - \underline{\eta}_{a}(\bX_i)\},\\
    \widetilde{U}_{sk2} &= n^{-1/2} \sum_{i=1}^n I\{S_i=s, T_i=k\}\left(\frac{I\{\tA_i=a\}}{\pi_a} \frac{R_i(a)}{\underline{\kappa}_{a}(\bX_i)}-1\right)  \{\underline{\eta}_{a}(\bX_i)-\widetilde{\eta}_{as,-k}(\bX_i)\},\\
    \widetilde{U}_{sk3} &= n^{-1/2} \sum_{i=1}^n \frac{I\{\tA_i=a,S_i=s, T_i=k\}}{\pi_a} \frac{\underline{\kappa}_a(\bX_i) - \widetilde{\kappa} _{as,-k}(\bX_i)}{\underline{\kappa}_a(\bX_i)\widetilde{\kappa} _{as,-k}(\bX_i)}\{\underline{\eta}_{a}(\bX_i)-\widetilde{\eta}_{as,-k}(\bX_i)\}.
\end{align*}
By the stratified rerandomization scheme, we have $U_{skj}|\mathcal{F}_n \overset{d}{=} \widetilde{U}_{skj}|( \mathcal{F}_n, \widetilde{Q} < t)$ for $j=1,2,3$. 
Then, we have
\begin{align*}
   P(|U_{skj}|> \varepsilon|\mathcal{F}_n) &= P(|\widetilde{U}_{skj}|> \varepsilon|\mathcal{F}_n, \widetilde{Q}<t)\\
   &= P(\widetilde{Q}<t|\mathcal{F}_n)^{-1} P(|\widetilde{U}_{skj}|> \varepsilon, \widetilde{Q}<t|\mathcal{F}_n) \\
   &\le  P(\widetilde{Q}<t|\mathcal{F}_n)^{-1} P(|\widetilde{U}_{skj}|> \varepsilon|\mathcal{F}_n).
\end{align*}
Since Lemma~\ref{lemma3} leads to $P(\widetilde{Q}<t|\mathcal{F}_n) \xrightarrow{p} P(\mathcal{X}^2_q < t) > 0$, the above derivation yields that $P(|\widetilde{U}_{skj}|> \varepsilon|\mathcal{F}_n) = o_p(1)$ implies $ P(|U_{skj}|> \varepsilon|\mathcal{F}_n) =o_p(1)$ and hence $U_{skj} = o_p(1)$. By Markov's inequality, it suffices to show $\widetilde{U}_{skj}= o_p(1)$. 
This result implies that the question reduces to controlling the remainder terms under stratified randomization (without rerandomization).

For showing $\widetilde{U}_{skj}= o_p(1)$, we next perform a series of transformations of $\widetilde{U}_{skj}$ such that the resulting quantities are easier to handle. These transforms are the same as those made in the proof of Theorem 4.1 of \cite{rafi2023efficient} and also resemble those made in \cite{wang2023model}, representing a standard yet complex method to deal with correlations introduced by stratified randomization. To start with, we define
\begin{align*}
    N_s &= \sum_{i=1}^n I\{S_i=s\}, \\
    N_{sk} &= \sum_{i=1}^n I\{S_i=s, T_i =k\}, \\
    N_{as} &= \sum_{i=1}^n I\{\tA_i =a, S_i=s\},\\
    N_{ask} &= \sum_{i=1}^n I\{\tA_i=a,S_i=s, T_i=k\}, \\
    N_{as,-k} &= \sum_{i=1}^n I\{\tA_i=a,S_i=s, T_i\ne k\}.
\end{align*}
Then, we independently generate $\{\bW_i(s)\}_{i=1}^n$ following definition~\ref{def: conditional-W} and define
\begin{align*}
    \widetilde{U}_{sk1}^\dagger &= n^{-1/2} \sum_{i=1}^n \frac{I\{\tA_i=a,S_i=s, T_i=k\}}{\pi_a} \frac{\underline{\kappa}_{as}(\bX_i(s)) - \widetilde{\kappa}_{as,-k}^\dagger(\bX_i(s))}{\underline{\kappa}_{as}(\bX_i(s))\widetilde{\kappa}_{as,-k}^\dagger(\bX_i(s))}\{Y_i(a,s) - \underline{\eta}_{as}(\bX_i(s))\},\\
    \widetilde{U}_{sk2}^\dagger &= n^{-1/2} \sum_{i=1}^n I\{S_i=s, T_i=k\}\left(\frac{I\{\tA_i=a\}}{\pi_a} \frac{R_i(a,s)}{\underline{\kappa}_{as}(\bX_i(s))}-1\right)  \{\underline{\eta}_{as}(\bX_i(s))-\widetilde{\eta}_{as,-k}(\bX_i(s))\},\\
    \widetilde{U}_{sk3}^\dagger &= n^{-1/2} \sum_{i=1}^n \frac{I\{\tA_i=a,S_i=s, T_i=k\}}{\pi_a} \frac{\underline{\kappa}_{as}(\bX_i(s)) - \widetilde{\kappa} _{as,-k}(\bX_i(s))}{\underline{\kappa}_{as}(\bX_i(s))\widetilde{\kappa} _{as,-k}(\bX_i(s))}\{\underline{\eta}_{as}(\bX_i(s))-\widetilde{\eta}_{as,-k}(\bX_i(s))\},
\end{align*}
where $\underline{\kappa}_{as}(\bX_i(s))= E[Y_i(a,s)|\bX_i(s)]$,  $\underline{\eta}_{as}(\bX_i(s))= E[R_i(a,s)|\bX_i(s)]$, $\widetilde{\kappa}_{as,-k}^\dagger(\bX_i(s)) = \widehat{\kappa}_{a}(\bX_i(s); \mathcal{O}^\dagger_{as,-k})$ and $\widetilde{\eta}_{as,-k}^\dagger(\bX_i(s)) = \widehat{\eta}_{a}(\bX_i(s); \mathcal{O}^\dagger_{as,-k})$ with $\mathcal{O}^\dagger_{as,-k} = \{R_i(a,s)Y_i(a,s), R_i(a,s), \bX_i(s): \tA_i = a, S_i=s, T_i\ne k\}$. Since $\bW_i(s)\overset{d}{=} \bW_i|(S_i=s)$ and $\{\bW_i(s)\}_{i=1}^n$  is independent of $\{\tA_i, S_i, T_i\}_{i=1}^n$, we have $\{W_i:S_i=s\}|\{\tA_i, S_i, T_i\}_{i=1}^n \overset{d}{=} \{W_i(s):S_i=s\}|\{\tA_i, S_i, T_i\}_{i=1}^n$. Since $ \widetilde{U}_{skj}$ is deterministic function of $\{W_i:S_i=s\}$ and $\{\tA_i, S_i, T_i\}_{i=1}^n$, we have $ \widetilde{U}_{skj}|\{\tA_i, S_i, T_i\}_{i=1}^n \overset{d}{=} \widetilde{U}_{skj}^\dagger|\{\tA_i, S_i, T_i\}_{i=1}^n$, which implies $\widetilde{U}_{skj} \overset{d}{=} \widetilde{U}_{skj}^\dagger$. Our goal is hence to prove $\widetilde{U}_{skj}^\dagger = o_p(1)$.

Next, we observe that, if we permute the indices of $i=1,\dots,n$ based on $\{\tA_i, S_i, T_i\}_{i=1}^n$,  the distribution of $\widetilde{U}_{skj}^\dagger$ remains unchanged because  $\{\bW_i(s)\}_{i=1}^n$  is independent of $\{\tA_i, S_i, T_i\}_{i=1}^n$.
Therefore, we apply a permutation such that $\{i:S_i=s\}$ are mapped to $\{1,\dots, N_s\}$, $\{i:S_i=s, T_i =k\}$ are mapped to $\{1,\dots, N_{sk}\}$, and  $\{i:\tA_i=a,S_i=s, T_i=k\}$ are mapped to $\{1,\dots, N_{ask}\}$. With such permutations, we have
\begin{align*}
    \widetilde{U}_{sk1}^\dagger &\overset{d}{=} n^{-1/2} \sum_{i=1}^{N_{ask}} \frac{1}{\pi_a} \frac{\underline{\kappa}_{as}(\bX_i(s)) - \widetilde{\kappa}_{as,-k}^\dagger(\bX_i(s))}{\underline{\kappa}_{as}(\bX_i(s))\widetilde{\kappa}_{as,-k}^\dagger(\bX_i(s))}\{Y_i(a,s) - \underline{\eta}_{as}(\bX_i(s))\},\\
    \widetilde{U}_{sk2}^\dagger &\overset{d}{=} n^{-1/2}  \left(\sum_{i=1}^{N_{ask}}\frac{1}{\pi_a} \frac{R_i(a,s)}{\underline{\kappa}_{as}(\bX_i(s))}\{\underline{\eta}_{as}(\bX_i(s))-\widetilde{\eta}_{as,-k}^\dagger(\bX_i(s))\}-\sum_{i=1}^{N_{sk}} \{\underline{\eta}_{as}(\bX_i(s))-\widetilde{\eta}_{as,-k}^\dagger(\bX_i(s))\}\right)  ,\\
    \widetilde{U}_{sk3}^\dagger &\overset{d}{=} n^{-1/2} \sum_{i=1}^{N_{ask}} \frac{1}{\pi_a} \frac{\underline{\kappa}_{as}(\bX_i(s)) - \widetilde{\kappa}^\dagger_{as,-k}(\bX_i(s))}{\underline{\kappa}_{as}(\bX_i(s))\widetilde{\kappa}^\dagger_{as,-k}(\bX_i(s))}\{\underline{\eta}_{as}(\bX_i(s))-\widetilde{\eta}_{as,-k}^\dagger(\bX_i(s))\},
\end{align*}
where 
$\mathcal{O}^\dagger_{as,-k} = \{R_i(a,s)Y_i(a,s), R_i(a,s), \bX_i(s): i = N_{sk}+1,\dots, N_{s}; \tA_i = a\}$. 
To further simplify the indices among observations in $\mathcal{O}^\dagger_{as,-k}$, we independently generate $\{\bW_i^{\ddagger}(s)\}_{i=1}^n$ following definition~\ref{def: conditional-W} such that $\{\bW_i(s)^{\ddagger}\}_{i=1}^n$ is independent of $\{\bW_i(s)\}_{i=1}^n$ and $\{\tA_i,S_i,T_i\}_{i=1}^n$. We define $\mathcal{O}^\ddagger_{as,-k} = \{R_i^{\ddagger}(a,s)Y_i^{\ddagger}(a,s), R_i^{\ddagger}(a,s), \bX_i^{\ddagger}(s): i = 1,\dots, N_{as,-k}\}$ and $\bN = (N_{ask}, N_{as,-k}, N_{s},N_{sk})$. By definition, we have $\mathcal{O}^\ddagger_{as,-k}|\bN \overset{d}{=} \mathcal{O}^\dagger_{as,-k}|\bN$, and $\mathcal{O}^\ddagger_{as,-k}$ is independent of $\{\bW_i(s)\}_{i=1}^n$ given $\bN$. We next define $\widetilde{\kappa}_{as,-k}^\ddagger(\bX_i(s)) = \widehat{\kappa}_{a}(\bX_i(s); \mathcal{O}^\ddagger_{as,-k})$, $\widetilde{\eta}_{as,-k}^\ddagger(\bX_i) = \widehat{\eta}_{a}(\bX_i(s); \mathcal{O}^\ddagger_{as,-k})$, and
\begin{align*}
    \widetilde{U}_{sk1}^\ddagger &= n^{-1/2} \sum_{i=1}^{N_{ask}} \frac{1}{\pi_a} \frac{\underline{\kappa}_{as}(\bX_i(s)) - \widetilde{\kappa}_{as,-k}^\ddagger(\bX_i(s))}{\underline{\kappa}_{as}(\bX_i(s))\widetilde{\kappa}_{as,-k}^\ddagger(\bX_i(s))}\{Y_i(a,s) - \underline{\eta}_{as}(\bX_i(s))\},\\
    \widetilde{U}_{sk2}^\ddagger &= n^{-1/2}  \left(\sum_{i=1}^{N_{ask}}\frac{1}{\pi_a} \frac{R_i(a,s)}{\underline{\kappa}_{as}(\bX_i(s))}\{\underline{\eta}_{as}(\bX_i(s))-\widetilde{\eta}_{as,-k}^\ddagger(\bX_i(s))\}-\sum_{i=1}^{N_{sk}} \{\underline{\eta}_{as}(\bX_i(s))-\widetilde{\eta}_{as,-k}^\ddagger(\bX_i(s))\}\right)  ,\\
    \widetilde{U}_{sk3}^\ddagger &= n^{-1/2} \sum_{i=1}^{N_{ask}} \frac{1}{\pi_a} \frac{\underline{\kappa}_{as}(\bX_i(s)) - \widetilde{\kappa}^\ddagger_{as,-k}(\bX_i(s))}{\underline{\kappa}_{as}(\bX_i(s))\widetilde{\kappa}^\ddagger_{as,-k}(\bX_i(s))}\{\underline{\eta}_{as}(\bX_i(s))-\widetilde{\eta}_{as,-k}^\ddagger(\bX_i(s))\}
\end{align*}
Here, $\widetilde{U}_{skj}^\ddagger$ differs from $\widetilde{U}_{skj}^\dagger$ only in that $(\widetilde{\kappa}_{as,-k}^\ddagger, \widetilde{\eta}_{as,-k}^\ddagger)$ substitutes $(\widetilde{\kappa}_{as,-k}^\dagger, \widetilde{\eta}_{as,-k}^\dagger)$. Since $\mathcal{O}_{sk}^\dagger = \{R_i(a,s)Y_i(a,s), R_i(a,s): i=1,\dots, N_{sk}\}$ is independent of $\mathcal{O}_{as,-k}^\dagger$ given $\bN$, we have $P(\mathcal{O}_{sk}^\dagger, \mathcal{O}_{as,-k}^\dagger) = P(\mathcal{O}_{sk}^\dagger|\bN)P(\mathcal{O}_{as,-k}^\dagger|\bN)P(\bN) = P(\mathcal{O}_{sk}^\dagger|\bN)P(\mathcal{O}_{as,-k}^\ddagger|\bN)P(\bN) = P(\mathcal{O}_{sk}^\dagger, \mathcal{O}_{as,-k}^\ddagger)$. Since $\widetilde{U}_{skj}^\ddagger$ is a deterministic function of $(\mathcal{O}_{sk}^\dagger, \mathcal{O}_{as,-k}^\ddagger)$ and $\widetilde{U}_{skj}^\dagger$ is the same deterministic function of $(\mathcal{O}_{sk}^\dagger, \mathcal{O}_{as,-k}^\dagger)$, we have $\widetilde{U}_{skj}^\ddagger \overset{d}{=} \widetilde{U}_{skj}^\dagger$, which means we only need to show $\widetilde{U}_{skj}^\ddagger = o_p(1)$.
Since $P(\widetilde{U}_{skj}^\ddagger > \varepsilon) = E[P(\widetilde{U}_{skj}^\ddagger > \varepsilon|\mathcal{O}_{as,-k}^\ddagger, \bN)]$, the dominated convergence theorem shows that it suffices to prove $P(\widetilde{U}_{skj}^\ddagger > \varepsilon|\mathcal{O}_{as,-k}^\ddagger, \bN) = o_p(1)$. Since Markov's inequality implies $P(\widetilde{U}_{skj}^\ddagger > \varepsilon|\mathcal{O}_{as,-k}^\ddagger, \bN) \le \varepsilon^{-2}E[\widetilde{U}_{skj}^\ddagger{}^2|\mathcal{O}_{as,-k}^\ddagger, \bN]$, we only need to show $E[\widetilde{U}_{skj}^\ddagger{}^2|\mathcal{O}_{as,-k}^\ddagger,\bN] = o_p(1)$. To this end, we deal with each $\widetilde{U}_{skj}^\ddagger$ separately, where we frequently apply a useful fact that $\bW_i(s), i=1,\dots, N_{sk}$ are independent and identically distributed given $(\mathcal{O}_{as,-k}^\ddagger,\bN)$, which is because $\{\bW_i(s)\}_{i=1}^n, \{\bW_i(s)^\ddagger\}_{i=1}^n, \bN$ are mutually independent. 

For $\widetilde{U}_{sk1}^\ddagger$, since $\bW_i(s), i=1,\dots, N_{sk}$ are independent and identically distributed given $(\mathcal{O}_{as,-k}^\ddagger,\bN)$, we have
\begin{align*}
   & E[\widetilde{U}_{sk1}^\ddagger|\mathcal{O}_{as,-k}^\ddagger,\bN]\\
   &= \frac{N_{ask}}{\pi_a\sqrt{n}} E\left[\frac{\underline{\kappa}_{as}(\bX_i(s)) - \widetilde{\kappa}_{as,-k}^\ddagger(\bX_i(s))}{\underline{\kappa}_{as}(\bX_i(s))\widetilde{\kappa}_{as,-k}^\ddagger(\bX_i(s))}\{Y_i(a,s) - \underline{\eta}_{as}(\bX_i(s))\}|\mathcal{O}_{as,-k}^\ddagger,\bN\right] \\
    &=\frac{N_{ask}}{\pi_a\sqrt{n}} E\left[\frac{\underline{\kappa}_{as}(\bX_i(s)) - \widetilde{\kappa}_{as,-k}^\ddagger(\bX_i(s))}{\underline{\kappa}_{as}(\bX_i(s))\widetilde{\kappa}_{as,-k}^\ddagger(\bX_i(s))}\{E[Y_i(a,s)|\bX_i(s),\mathcal{O}_{as,-k}^\ddagger,\bN]  - \underline{\eta}_{as}(\bX_i(s))\}|\mathcal{O}_{as,-k}^\ddagger,\bN\right].
\end{align*}
Since $(Y_i(a,s), \bX_i(s))$ is independent of $(\mathcal{O}_{as,-k}^\ddagger,\bN)$, we have $E[Y_i(a,s)|\bX_i(s),\mathcal{O}_{as,-k}^\ddagger,\bN]=E[Y_i(a,s)|\bX_i(s)] = \underline{\eta}_{as}(\bX_i(s))$, implying $E[\widetilde{U}_{sk1}^\ddagger|\mathcal{O}_{as,-k}^\ddagger,\bN] = 0$. Then, 
\begin{align*}
    E[\widetilde{U}_{sk1}^\ddagger{}^2|\mathcal{O}_{as,-k}^\ddagger,\bN] &= Var(\widetilde{U}_{sk1}^\ddagger|\mathcal{O}_{as,-k}^\ddagger,\bN) \\
    &= \frac{N_{ask}}{n\pi_a^2} Var\left[\frac{\underline{\kappa}_{as}(\bX_i(s)) - \widetilde{\kappa}_{as,-k}^\ddagger(\bX_i(s))}{\underline{\kappa}_{as}(\bX_i(s))\widetilde{\kappa}_{as,-k}^\ddagger(\bX_i(s))}\{Y_i(a,s) - \underline{\eta}_{as}(\bX_i(s))\}\bigg|\mathcal{O}_{as,-k}^\ddagger,\bN\right] \\
    &\le \frac{N_{ask}}{n\pi_a^2} E\left[\left[\frac{\underline{\kappa}_{as}(\bX_i(s)) - \widetilde{\kappa}_{as,-k}^\ddagger(\bX_i(s))}{\underline{\kappa}_{as}(\bX_i(s))\widetilde{\kappa}_{as,-k}^\ddagger(\bX_i(s))}\{Y_i(a,s) - \underline{\eta}_{as}(\bX_i(s))\}\right]^2\bigg|\mathcal{O}_{as,-k}^\ddagger,\bN\right] \\
    &\le  C\frac{1}{\pi_a^2} E\left[\{\underline{\kappa}_{as}(\bX_i(s)) - \widetilde{\kappa}_{as,-k}^\ddagger(\bX_i(s))\}^2\bigg|\mathcal{O}_{as,-k}^\ddagger,\bN\right] \numberthis\label{Eq: kappa-convergence-s}\\
    &= o_p(1),
\end{align*}
where the second last line results from the fact that $N_{ask}<n$ and Assumptions 2 and 4 (2) that $E[Y^2(a)|\bX], \{\underline{\kappa}_{as}(\bX)\}^{-1}, \{\widehat{\kappa}_{as,-k}(\bX)\}^{-1}$ are uniformly bounded, and the last line result from Lemma~\ref{lemma: O_as-k}.

For $\widetilde{U}_{sk2}^\ddagger$, since $\bW_i(s), i=1,\dots, N_{sk}$ are independent and identically distributed given $(\mathcal{O}_{as,-k}^\ddagger,\bN)$, we have
\begin{align*}
    E[\widetilde{U}_{sk2}^\ddagger|\mathcal{O}_{as,-k}^\ddagger,\bN] &=   \frac{N_{ask}}{\sqrt{n}\pi_a} E\left[\frac{R_i(a,s)}{\underline{\kappa}_{as}(\bX_i(s))}\{\underline{\eta}_{as}(\bX_i(s))-\widetilde{\eta}_{as,-k}^\ddagger(\bX_i(s))\}\bigg|\mathcal{O}_{as,-k}^\ddagger,\bN\right] \\
    &\quad - \frac{N_{sk}}{\sqrt{n}}  E[\underline{\eta}_{as}(\bX_i(s))-\widetilde{\eta}_{as,-k}^\ddagger(\bX_i(s))|\mathcal{O}_{as,-k}^\ddagger,\bN] \\
    &= \left(\frac{N_{ask}}{\sqrt{n}\pi_a} - \frac{N_{sk}}{\sqrt{n}}\right) E[\underline{\eta}_{as}(\bX_i(s))-\widetilde{\eta}_{as,-k}^\ddagger(\bX_i(s))|\mathcal{O}_{as,-k}^\ddagger,\bN],
\end{align*}
where the last line result from the fact that $E[R_i(a,s)|\bX_i(s), \mathcal{O}_{as,-k}^\ddagger,\bN] = E[R_i(a,s)|\bX_i(s)] = \underline{\kappa}_{as}(\bX_i(s))$. By stratified randomization and sample splitting scheme, $N_{ask}-\pi_a N_{sk}$ is zero (or at least bounded by a constant independent of $n$). In addition, Lemma~\ref{lemma: O_as-k} implies that $|E[\underline{\eta}_{as}(\bX_i(s))-\widetilde{\eta}_{as,-k}^\ddagger(\bX_i(s))|\mathcal{O}_{as,-k}^\ddagger,\bN]| = o_p(1)$, we have $E[\widetilde{U}_{sk2}^\ddagger|\mathcal{O}_{as,-k}^\ddagger,\bN] = O_p(1)o_p(1) =o_p(1)$. Moving forward, we have
\begin{align*}
    & E[\widetilde{U}_{sk2}^\ddagger{}^2|\mathcal{O}_{as,-k}^\ddagger,\bN]\\ &= Var(\widetilde{U}_{sk2}^\ddagger|\mathcal{O}_{as,-k}^\ddagger,\bN) + E[\widetilde{U}_{sk2}^\ddagger|\mathcal{O}_{as,-k}^\ddagger,\bN]^2 \\
    &= \frac{N_{ask}}{n}Var\left[\left\{\frac{1}{\pi_a} \frac{R_i(a,s)}{\underline{\kappa}_{as}(\bX_i(s))}-1\right\}\{\underline{\eta}_{as}(\bX_i(s))-\widetilde{\eta}_{as,-k}^\ddagger(\bX_i(s))\}\bigg|\mathcal{O}_{as,-k}^\ddagger,\bN\right] \\
    &\quad + \frac{N_{sk}-N_{ask}}{n}Var\left[\underline{\eta}_{as}(\bX_i(s))-\widetilde{\eta}_{as,-k}^\ddagger(\bX_i(s))\bigg|\mathcal{O}_{as,-k}^\ddagger,\bN\right] + o_p(1) \\
    &\le E\left[\left(\left\{\frac{1}{\pi_a} \frac{R_i(a,s)}{\underline{\kappa}_{as}(\bX_i(s))}-1\right\}^2+1\right)\{\underline{\eta}_{as}(\bX_i(s))-\widetilde{\eta}_{as,-k}^\ddagger(\bX_i(s))\}^2\bigg|\mathcal{O}_{as,-k}^\ddagger,\bN\right] + o_p(1)\\
    &\le C\times E\left[\{\underline{\eta}_{as}(\bX_i(s))-\widetilde{\eta}_{as,-k}^\ddagger(\bX_i(s))\}^2\bigg|\mathcal{O}_{as,-k}^\ddagger,\bN\right] + o_p(1),\numberthis\label{Eq: eta-convergence-s}.
\end{align*}
In the above derivations, the first inequality uses the fact that $N_{ask} < n$, $N_{sk}-N_{ask} < n$, and variance is smaller than the second moment. The last inequality results from the regularity condition that $\ \{\underline{\kappa}_{as}(\bX)\}^{-1}$ is uniformly bounded. Applying Lemma~\ref{lemma: O_as-k}, we obtain $ E[\widetilde{U}_{sk2}^\ddagger{}^2|\mathcal{O}_{as,-k}^\ddagger,\bN] = o_p(1)$.

For $\widetilde{U}_{sk3}^\ddagger$,  since $\bW_i(s), i=1,\dots, N_{sk}$ are independent and identically distributed given $(\mathcal{O}_{as,-k}^\ddagger,\bN)$, we have
\begin{align*}
    & E[\widetilde{U}_{sk3}^\ddagger|\mathcal{O}_{as,-k}^\ddagger,\bN] \\
    &= \frac{N_{ask}}{\sqrt{n}\pi_a} E\left[\frac{\underline{\kappa}_{as}(\bX_i(s)) - \widetilde{\kappa}^\ddagger_{as,-k}(\bX_i(s))}{\underline{\kappa}_{as}(\bX_i(s))\widetilde{\kappa}^\ddagger_{as,-k}(\bX_i(s))}\{\underline{\eta}_{as}(\bX_i(s))-\widetilde{\eta}_{as,-k}^\ddagger(\bX_i(s))\}\bigg|\mathcal{O}_{as,-k}^\ddagger,\bN\right]\\
    &\le C \sqrt{n} E\left[\{\underline{\kappa}_{as}(\bX_i(s)) - \widetilde{\kappa}^\ddagger_{as,-k}(\bX_i(s))\}\{\underline{\eta}_{as}(\bX_i(s))-\widetilde{\eta}_{as,-k}^\ddagger(\bX_i(s))\}\bigg|\mathcal{O}_{as,-k}^\ddagger,\bN\right] \\
    &\le C \sqrt{n} E\left[\{\underline{\kappa}_{as}(\bX_i(s)) - \widetilde{\kappa}^\ddagger_{as,-k}(\bX_i(s))\}^2\bigg|\mathcal{O}_{as,-k}^\ddagger,\bN\right]^{1/2}E\left[\{\underline{\eta}_{as}(\bX_i(s))-\widetilde{\eta}_{as,-k}^\ddagger(\bX_i(s))\}^2\bigg|\mathcal{O}_{as,-k}^\ddagger,\bN\right]^{1/2} \\
    &= C \sqrt{n} o_p(n^{-1/4})o_p(n^{-1/4}) \\
    &= o_p(1),
\end{align*}
where the third line comes from $N_{ask} < n$ and the regularity condition that $\underline{\kappa}_a, \widehat{\kappa}_a$ are uniformly bounded, the fourth line uses the Cauchy-Schwarz inequality, and the fifth line is implied by Lemma~\ref{lemma: O_as-k}. Then, we have
\begin{align*}
    & E[\widetilde{U}_{sk3}^\ddagger{}^2|\mathcal{O}_{as,-k}^\ddagger,\bN]\\ &= Var(\widetilde{U}_{sk3}^\ddagger|\mathcal{O}_{as,-k}^\ddagger,\bN) + E[\widetilde{U}_{sk3}^\ddagger|\mathcal{O}_{as,-k}^\ddagger,\bN]^2 \\
    &= \frac{N_{ask}}{n}Var\left[\frac{1}{\pi_a} \frac{\underline{\kappa}_{as}(\bX_i(s)) - \widetilde{\kappa}^\ddagger_{as,-k}(\bX_i(s))}{\underline{\kappa}_{as}(\bX_i(s))\widetilde{\kappa}^\ddagger_{as,-k}(\bX_i(s))}\{\underline{\eta}_{as}(\bX_i(s))-\widetilde{\eta}_{as,-k}^\ddagger(\bX_i(s))\}\bigg|\mathcal{O}_{as,-k}^\ddagger,\bN\right] + o_p(1) \\
    &\le E\left[\left[\frac{1}{\pi_a} \frac{\underline{\kappa}_{as}(\bX_i(s)) - \widetilde{\kappa}^\ddagger_{as,-k}(\bX_i(s))}{\underline{\kappa}_{as}(\bX_i(s))\widetilde{\kappa}^\ddagger_{as,-k}(\bX_i(s))}\{\underline{\eta}_{as}(\bX_i(s))-\widetilde{\eta}_{as,-k}^\ddagger(\bX_i(s))\}\right]^2\bigg|\mathcal{O}_{as,-k}^\ddagger,\bN\right] + o_p(1) \\
    &\le C\times E\left[\{\underline{\eta}_{as}(\bX_i(s))-\widetilde{\eta}_{as,-k}^\ddagger(\bX_i(s))\}^2\bigg|\mathcal{O}_{as,-k}^\ddagger,\bN\right] + o_p(1),
\end{align*}
where the first inequality is because $N_{ask} < n$ and variance is smaller than the second moment, and the second inequality is because $ \{\underline{\kappa}_{as}(\bX)\}^{-1}, \{\widehat{\kappa}_{as,-k}(\bX)\}^{-1}$ are uniformly bounded. Applying Lemma~\ref{lemma: O_as-k},we get $E[\widetilde{U}_{sk3}^\ddagger{}^2|\mathcal{O}_{as,-k}^\ddagger,\bN] =o_p(1)$. Therefore, we complete the proof for $U_{sk} = o_p(1)$. This completes the proof for asymptotic linearity under stratified rerandomization. 

Then the asymptotic distribution is a direct result of Theorem 3, which also implies consistency. Next, we can compute
\begin{align*}
    EIF(\bO(1)) - EIF(\bO(0)) = \sum_{a=0}^1 (-1)^{1-a} f_a'\left[ \frac{1}{\pi_a} \frac{R(a)}{E[R(a)|\bX]}\{Y(a) - E[Y(a)|\bX]\} \right],
\end{align*}
which implies
\begin{align*}
    E\left[EIF(\bO(1)) - EIF(\bO(0))|\bX\right] &= \sum_{a=0}^1 (-1)^{1-a} f_a' \frac{1}{\pi_a} E\left[\frac{R(a)}{E[R(a)|\bX]}\{Y(a) - E[Y(a)|\bX]\}\bigg|\bX\right] \\
    &= 0
\end{align*}
If $\bX$ contains $\bX^r$ and $S$ as a subset, we have
\begin{align*}
    & E[Cov\{IF(\bO(1)) - IF(\bO(0)), \bX^r|S\}]\\
    &= E[\{IF(\bO(1)) - IF(\bO(0))\}\bX^r] - E[E[IF(\bO(1)) - IF(\bO(0))|S]E[\bX^r|S]] \\
    &= \bzero,
\end{align*}
which implies $R^2=0$.

Finally, we construct consistent variance estimators as
\begin{align}
       \widehat{\widetilde{V}} &= \widehat{V} - \pi(1-\pi) \sum_{s \in \mathcal{S}} \widehat{p}_s \widehat{d}_s^2, \label{vhat: ml-st}\\
    \widehat{\widetilde{R}^2} &= \frac{1}{\widehat{\widetilde{V}}} \widehat{C}^\top \widehat{Var}(\widetilde{\bI})^{-1} \widehat{C},\label{R2hat: ml-st}
\end{align}
where $\widehat{p}_s = n^{-1}\sum_{i=1}^n I\{S_i=s\}$ and
\begin{align*}
    \widehat{V} &= \frac{1}{K}\sum_{k=1}^K \sum_{s \in \mathcal{S}} \widehat{p}_s\frac{1}{|\mathcal{I}_{sk}|} \sum_{i \in \mathcal{I}_{sk}} \widehat{EIF}_{i}^2, \\
    \widehat{d}_s &= \frac{1}{K}\sum_{k=1}^K \frac{1}{|\mathcal{I}_{sk}|}\sum_{i \in \mathcal{I}_{sk}}\frac{A_i-\pi}{\pi(1-\pi)}  \widehat{EIF}_i, \\
    \widehat{C} &=  \sum_{s \in \mathcal{S}} \widehat{p}_s \left[ \frac{1}{K}\sum_{k=1}^K \frac{1}{|\mathcal{I}_{sk}|}\sum_{i \in \mathcal{I}_{sk}}\frac{A_i-\pi}{\pi(1-\pi)}  \widehat{EIF}_i\bX_i^r -\widehat{d}_s\overline{\bX}_s^r \right], \\
    \widehat{EIF}_i &= f_1'\widehat{EIF}_i(1) + f_0'\widehat{EIF}_i(0), \\
    \widehat{EIF}_i(a) &= \frac{I\{A_i=a\}}{\pi^a(1-\pi)^{1-a}}\frac{R_i}{\widehat{\kappa}_a(\bX_i)}\{Y_i - \widehat{\eta}_a(\bX_i)\} + \widehat{\eta}_a(\bX_i)-\widehat{\mu}_a^{\textup{dml}},\\
    \overline{\bX}_s^r &= \frac{1}{n} \sum_{i=1}^n \frac{I\{S_i=s\}} {\widehat{p}_s}\bX_i^r. 
\end{align*}
Consistency of $\widehat{\widetilde{V}}$ and $\widehat{\widetilde{R}^2}$ can be proof similar to Theorem 5.
\end{proof}

\begin{definition}\label{def: conditional-W}
    For each $s \in \mathcal{S}$, we define the conditional distribution in strautm $s$ as $\mathcal{P}^{\bW(s)} = \mathcal{P}^{\bW|(S=s)}$ on $\bW(s) = \{R(a,s), Y(a,s), \bX(s): a = 0,1\}$, where $\mathcal{P}^{\bW}$ is defined in Assumption 1. We further define $\{\bW_i(s)\}_{i=1}^n$ as $n$ independent samples from $\mathcal{P}^{\bW(s)}$, where $\{\bW_i(s)\}_{i=1}^n$ is also independent of $\{\bW_i\}_{i=1}^n$ and $\{\tA_i, T_i\}_{i=1}^n$, where $T_i$ is the random variable for sample splitting.
\end{definition}

\begin{lemma}\label{lemma: O_as-k}
Under Assumptions 1-4, we have
\begin{align*}
    E\left[\{\underline{\kappa}_{as}(\bX_i(s)) - \widetilde{\kappa}_{as,-k}^\ddagger(\bX_i(s))\}^2\bigg|\mathcal{O}_{as,-k}^\ddagger,\bN\right] &= o_p(n^{-1/2}), \\
    E\left[\{\underline{\eta}_{as}(\bX_i(s)) - \widetilde{\eta}_{as,-k}^\ddagger(\bX_i(s))\}^2\bigg|\mathcal{O}_{as,-k}^\ddagger,\bN\right] &= o_p(n^{-1/2})
\end{align*}
for the conditional expectations defined in Equation~\eqref{Eq: kappa-convergence-s}.
\end{lemma}
\begin{proof}
Denoting $Z= E\left[\{\underline{\kappa}_{as}(\bX_i(s)) - \widetilde{\kappa}_{as,-k}^\ddagger(\bX_i(s))\}^2\bigg|\mathcal{O}_{as,-k}^\ddagger,\bN\right]$, we have\\ $P\{n^{1/2}Z>\varepsilon\} = E[P(n^{1/2}Z>\varepsilon|N_{as,-k})]$. By the dominated convergence theorem, it suffices to show $P(n^{1/2}Z>\varepsilon|N_{as,-k}) = o_p(1)$. By Markov's inequality, we have $P(n^{1/2}Z>\varepsilon|N_{as,-k}) \le\varepsilon^{-1} E[n^{1/2}Z|N_{as,-k}]$, and we only need to show $E[Z|N_{as,-k}] = o_p(n^{-1/2})$. We next apply Lemma~\ref{lemma: random indices} to prove this result.

Defining $X_n = \underline{\kappa}_{as}(\bX_i(s)) - \widehat{\kappa}_{a}(\bX_i(s); \{\bO_i^\ddagger(a,s)\}_{i=1}^n)$ and $Y_n = N_{as,-k}$, where $\bO_i^\ddagger(a,s)=\{R_i^\ddagger(a,s) Y_i^\ddagger(a,s), R_i^\ddagger(a,s),\bX_i^\ddagger(s)\}$. We observe that $\widetilde{\kappa}_{as,-k}^\ddagger(\bX_i(s)) =  \widehat{\kappa}_{a}(\bX_i(s); \{\bO_i^\ddagger(a,s)\}_{i=1}^{N_{as,-k}})$, which implies $ E[Z|N_{as,-k}] = E[X_{Yn}|Y_n]$. To show $E[X_{Yn}^2|Y_n] = o_p(n^{-1/2})$, we verify the conditions of Lemma~\ref{lemma: random indices}. By definition, we have $N_{as,-k}$ is independent of $\bX_i(s)$ and $\{\bO_i^\ddagger(a,s)\}_{i=1}^n$. Furthermore, we have $N_{as,-k}/n \xrightarrow{p} \pi_a P(S=s)(K-1)/K$ under stratified randomization. Finally, Assumption 4 (1) implies that $E[X_n^2] = o_p(n^{-1/2})$. Therefore, Lemma~\ref{lemma: random indices} implies the desired result for $\kappa$. Following the same procedure, we can obtain the convergence result for $\eta$. 
\end{proof}

\subsubsection{Efficiency invariance to randomization schemes}

Recall $\bW=(Y(0),Y(1),R(0),R(1), X)$. For brevity, let $\bT = (Y(0),Y(1),R(0),R(1))$ denote the vector of potential outcomes, and let $A_{1:n}, \bX_{1:n}, \bT_{1:n}$ denote $\{A_i\}_{i=1}^n, \{\bX_i\}_{i=1}^n, \{\bT_i\}_{i=1}^n$, respectively.  We further define a \emph{design} $D$ as a known, parameter-free conditional law
$G^D(A_{1:n}\mid \bX_{1:n})$, e.g.\ simple randomization with fixed $\pi$ or rerandomization with $\pi(x)=P(A=1\mid X=x)$.  We assume conditional ignobility, i.e., $A\perp \bT |\bX$, which automatically holds under simple randomization. This is also true under rerandomization, because $A_{1:n}\perp \bT_{1:n} | \bX_{1:n}$ and hence, for any Borel set $B$,
\begin{align*}
    P(A_i \in B | \bX_i, \bT_i) &= E[P(A_i\in B |\bX_{1:n}, \bT_i)|\bX_{1:n}, \bT_i] \\
    &= E[P(A_i\in B |\bX_{1:n})|\bX_i, \bT_i]\\
    &= E[P(A_i\in B |\bX_{1:n})|\bX_i] \quad \textup{(by Assumption 1)} \\
    &= P(A_i \in B | \bX_i).
\end{align*}

Below we give the formal statement that different randomization schemes provide the same semiparametric efficiency bound using the average treatment effect $\Delta = E[Y(1)-Y(0)]$ as an example; the proposition directly generalizes to other estimands.

\begin{proposition}[Efficiency invariance to parameter-free designs with MAR]\label{prop:inv-paramfree-R}
For the observed-data experiments $\bO_1,\dots, \bO_n$ induced by $(P,G^D)$ for any parameter-free design $D$ satisfying $A\perp \bT |\bX$, we have, under Assumptions 1-2,

\noindent (i) The observed-data tangent space $\mathcal T$ for $\Delta$ is the same for all such $D$.

\noindent (ii) The efficient influence function for $\Delta$ in this causal model is
\begin{align*}
\phi(\bO)&=
\frac{A\,R}{\pi(\bX)\,E[R(1)|\bX]}\{Y-E[Y(1)|\bX]\}
\;-\;\frac{(1-A)\,R}{(1-\pi(\bX))\,E[R(0)|\bX]}\{Y-E[Y(1)|\bX]\} \\
&\quad + E[Y(1)|\bX]-E[Y(0)|\bX]-\Delta.
\end{align*}

\noindent (iii) Hence the semiparametric efficiency bound equals $Var\{\phi(\bO)\}$ and is identical across all parameter-free designs $D$.
\end{proposition}

\begin{proof}
Under design $D$, the joint likelihood factors as
\[
L_n^D(P)\ =\
\underbrace{G^D\!\big(A_{1:n}\mid X_{1:n}\big)}_{\text{known, parameter-free design}}
\cdot \prod_{i=1}^n f_X(\bX_i)\cdot
\prod_{i=1}^n f_{T\mid X}(\bT_i\mid \bX_i)
\cdot \prod_{i=1}^n f_{O_i\mid X_i,A_i,T_i}
\]
where $f_{O_i\mid  X_i,A_i,T_i}$ is deterministic and enforces consistency:
if $A_i=a$ and $R_i(a)=1$, then $R_i=1$ and $Y_i=Y_i(a)$; if $R_i(a)=0$, then $R_i=0$ and $Y_i$ is absent.
Let $\{P_\varepsilon\}$ be any regular parametric submodel of $P$.
Because $G^D$ is parameter-free, its derivative over $\varepsilon$ is zero. Thus, the observed-data score equals the sum of $n$ per-unit scores
that live in the span of scores for $f_X$ and for the conditional law of $T$ given $X$ mapped through the observation mechanism.
Therefore, the observed-data tangent space $\mathcal T$ does not depend on $D$, proving (i).

With the same tangent space across designs, the estimand $\Delta(P)$ is pathwise differentiable, and its canonical gradient is the well-known augmented inverse-probability weighted form of $\phi(\bO)$. Finally, we have $Var(\sum_{i=1}^n \phi(\bO_i)) = \sum_{i=1}^n Var(\phi(\bO_i))$. This is because $T_i \perp A_{1:n}|\bX_{1:n}$ and Assumptions 1-2 imply that $\bT_1,\dots, \bT_n$ are conditional independent given $(\bX_{1:n}, A_{1:n})$ and hence, for $i\ne j$,
\begin{align*}
   E[\phi(\bO_i)\phi(\bO_j)] &= E[E[\phi(\bO_i)\phi(\bO_j)|A_{1:n}, 
   \bX_{1:n}]] \\
   &= E[E[\phi(\bO_i)|A_{1:n}, \bX_{1:n}]E[\phi(\bO_j)|A_{1:n}, \bX_{1:n}]] \\
   &= E[\{E[Y(1)|\bX_i]-E[Y(0)|\bX_i]-\Delta\}\{E[Y(1)|\bX_j]-E[Y(0)|\bX_j]-\Delta\}]\\
   &=0,
\end{align*}
thereby removing all interaction terms. Consequently, the Fisher information is still $n$ times the per-unit information, exactly as under simple randomization, and hence the semiparametric efficiency bound $Var\{\phi(\bO)\}$ is the same under every parameter-free design $D$, proving (iii).
\end{proof}

\subsection{Finite-sample results}

{
\bibliographystyle{apalike}
\bibliography{references}
}